\RequirePackage{amsmath}
\documentclass[preprint,12pt,authoryear]{elsarticle}
\usepackage{amsmath}
\usepackage{amssymb}
\usepackage[a4paper,margin=0.89in,footskip=0.5in]{geometry}
\usepackage{url}
\usepackage[colorlinks,citecolor=blue,urlcolor=blue,linkcolor=blue]{hyperref} 
\usepackage{graphicx,multirow,wasysym}
\usepackage{lmodern}
\usepackage{amsfonts,amsmath}
\usepackage{paralist}
\usepackage{enumerate}
\usepackage{bbm}
\usepackage{tikz}
\usepackage{verbatim}
\usepackage[shortlabels]{enumitem}
\usepackage{comment}
\usepackage{booktabs}
\usepackage{lipsum}
\usepackage{tabularx}
\usepackage{xmpmulti}
\usepackage{multicol}
\usepackage{float}
\floatstyle{boxed}
\restylefloat{figure}
\usepackage[colorinlistoftodos]{todonotes}
\usepackage{collectbox}

\usepackage{placeins}
\makeatletter
\AtBeginDocument{%
	\expandafter\renewcommand\expandafter\subsection\expandafter
	{\expandafter\@fb@secFB\subsection}%
	\newcommand\@fb@secFB{\FloatBarrier
		\gdef\@fb@afterHHook{\@fb@topbarrier \gdef\@fb@afterHHook{}}}%
	\g@addto@macro\@afterheading{\@fb@afterHHook}%
	\gdef\@fb@afterHHook{}%
}
\makeatother

\usepackage{natbib}
\setcitestyle{square,numbers}

\newtheorem{theorem}{Theorem}
\newtheorem{lemma}[theorem]{Lemma}
\newdefinition{remark}{Remark}
\newproof{proof}{Proof}
\newtheorem{definition}{Definition}

\makeatletter
\newcommand{\mybox}{%
	\collectbox{%
		\setlength{\fboxsep}{1pt}%
		\fbox{\BOXCONTENT}%
	}%
}
\makeatother

\newcommand{\remove}[1]{}

\journal{arXiv}
\begin{document}
	
\begin{frontmatter}
\title{Traceable Policy-Based Signatures and Instantiation from Lattices}

\author[1]{Yanhong Xu\corref{cor1}
}
\ead{yanhong.xu1@ucalgary.ca}

\author[1]{Reihaneh Safavi-Naini}
\ead{rei@ucalgary.ca}

\author[2]{Khoa Nguyen}
\ead{khoantt@ntu.edu.sg}

\author[2]{Huaxiong Wang}
\ead{hxwang@ntu.edu.sg}

\cortext[cor1]{Corresponding author}

\address[1]{Department of Computer Science, 
	University of Calgary, 
	2500 University Drive, NW
	Calgary T2N 1N4, Canada}
\address[2]{School of Physical and Mathematical Sciences, Nanyang Technological University, 
	21 Nanyang Link, Singapore 637371}

\begin{abstract}
Policy-based signatures (PBS) were proposed by Bellare and Fuchsbauer (PKC 2014) to allow an {\em authorized} member of an organization to sign a message on behalf of the organization. The user's authorization is determined by a policy managed by the organization's trusted authority, while the signature preserves the privacy of the organization's policy. Signing keys in PBS do not include user identity information and thus can be passed to others, violating the intention of employing PBS to restrict users' signing capability.

In this paper, we introduce the notion of {\em traceability} for PBS by including user identity in the signing key such that the trusted authority will be able to open a suspicious signature and recover the signer's identity should the needs arise. We provide rigorous definitions and stringent security notions of traceable PBS (TPBS), capturing the properties of PBS suggested by Bellare-Fuchsbauer and resembling the ``full traceability'' requirement for group signatures put forward by Bellare-Micciancio-Warinschi (Eurocrypt 2003). As a proof of concept, we provide a modular construction of TPBS, based on a signature scheme, an encryption scheme and a zero-knowledge proof system. Furthermore, to demonstrate the feasibility of achieving TPBS from concrete, quantum-resistant assumptions, we give an instantiation based on  lattices.
\end{abstract}

\end{frontmatter}

\section{Introduction}
 Policy-based signatures (PBS) were introduced by Bellare and Fuchsbauer~\cite{BF14PKC} to allow authorized users in an organization to sign messages on behalf of the organization, while keeping the internal authorization policy  of the organization private.  A signature in PBS will be verified with respect to the organization's public key and so does not leak any identity information, nor it reveals the policy that is applied to the signed message.
PBS is an attractive privacy-preserving primitive in practice as, similar to group signature~\cite{CH91EC}, it allows members of an organization to sign on behalf of the organization, but enables organization to enforce fine-grained policies for signing messages without revealing to the outside world.  From a theoretical viewpoint, PBS is also a powerful primitive: it captures and implies a number of anonymity-oriented authentication systems, serving as an umbrella notion that unifies many existing notions. 
In particular, it was shown~\cite{BF14PKC} to imply group signatures~\cite{CH91EC,BMW03EC}, ring signatures~\cite{RST01AC}, anonymous credential~\cite{CL01EC},  anonymous proxy signatures~\cite{FP08SCN}, attribute-based signatures~\cite{MPR11CT-RSA}, and certain variants of functional signatures~\cite{BGI14PKC,BMS16PKC}.

Bellare and Fuchsbauer~\cite{BF14PKC} initially required  two  basic security requirements for PBS: {\em  indistinguishability}  and \emph{unforgeability.} Indistinguishability captures privacy of the policy under which a message is signed, and  requires the verifier not to be able to distinguish which of the two candidate keys is used for signing a message $m$, knowing that the message satisfies the underlying policies of both keys. Unforgeability demands that it is infeasible to produce a valid signature on a message without possessing a key for a policy that permits the given message\footnote{Note that a message could be authorized under multiple policies.}. 
Bellare  and Fuchsbauer however argued that a usual indistinguishability notion is insufficient for some applications and a typical definition of unforgeability can lead to technical difficulties. They therefore proposed two stringent notions of {\em simulatability} and {\em extractability}  that are proven to imply indistinguishability and unforgeability. 
They also provided generic constructions of PBS satisfying the proposed strong notions of security, and demonstrated a concrete pairing-based instantiation.  In their constructions, a  trusted authority issues a signing key $sk_p$ for a policy $p$  by certifying $p$ via an ordinary signature scheme. If a message $m$ complies with policy $p$, then the holder of $sk_p$  can create a PBS on $m$ by generating a zero-knowledge proof of knowledge of a valid certificate on certain hidden policy $p$ such that  the given message $m$ conforms  to $p$.

We observe that signing keys in PBS are associated only with policies and does not contain any identifying information of users. This pitfall allows the keys to be easily misused without penalty: key holders can freely pass their keys to anyone they choose to, enabling them to sign on behalf of the organization. Such a limitation completely opens PBS to insider adversaries who can share and/or exchange their keys without any repercussion,  and so completely bypass the organization's security/privacy regulations.

\smallskip
\noindent{\sc Protection against key sharing.} 
To protect against key sharing one may attach the user's identity to the policy. That is to issue   signing keys for user $\mathsf{id}$ on $ \mathsf{id}\|p$ for all plausible policies (a user may be authorized for multiple policies), ensuring that the user's identity is part of the key. This, however,  achieves privacy of identity in the same way we achieve it for policy.  Therefore, the user can still share  (or  exchange)  their signing keys without any real penalty as they can always claim it is lost or stolen.

One may consider the ``all-or-nothing'' non-transferability approach  proposed by Camenisch and Lysyanskaya~\cite{CL01EC}, 
where sharing a single credential  would reveal all credentials of the user and would lead to the user's loss of its identity. In PBS, one can link  all signing keys of a user to provide such a feature, but still, it does not prevent malicious users from sharing their keys, although in this case they would have to share all their keys at the same time.

One can however make use of the approach that has been employed in a related primitive, group signatures (GS)~\cite{CH91EC,BMW03EC}, to provide  {\em traceability}.
 In GS, a member of an organization has an individualized secret key which they can use to sign on behalf  the organization while keeping the signer's identity hidden.
 The system, however, allows the presence of an {\em opener} (who could be the same as the secret key issuer or be a separate trusted entity), to open the signature and reveal the signer's identity - in cases of disputes. This is achieved by including an encryption of the user identity, encrypted  with the public key of the opener, in the group signature.  This is the approach that we will take in this paper.

 \smallskip

 \noindent {\sc Our Results. }In this work, we aim to equip PBS with a tracing functionality, and study how to build such an enhanced scheme in a modular manner based on generic assumptions, as well as, concrete lattice-based assumptions. Our contributions are three-fold. 
\smallskip

\noindent
{\bf Formalization.} 
 We propose the primitive of traceable policy-based signatures (TPBS), an extension of PBS in which the identity of the user who has signed a message can be  recovered by an opening algorithm,  while privacy of the signing policy is preserved even against the opening authority. 
In a nutshell, TPBS enrichs PBS with a reasonable traceability feature, reminiscent to group signatures. 
This traceability aspect deters users from sharing or exchanging 
their signing keys (and hence bypassing the organization's regulations), and in general, keeps users accountable for their potentially inappropriate actions. We  formalize the  security notions of simulatability and extractability for 
TPBS  inline with the corresponding notions in~\cite{BF14PKC}, giving additional power to the adversary by modelling their capabilities to ask for  opening of signatures.  In particular,  in the  definition of  simulatability, we provide the adversary with  
 access to an opening oracle, so that to capture 
 the situation when the adversary sees the results of previous openings.  In the  definition of extractability, we also include the requirement that an adversary who has queried signing keys of  a set of users,  be unable to output a valid signature that cannot be opened or traced to a member of the queried group.

\smallskip

\noindent 
{\bf Generic construction.} As a proof of concept, we provide generic construction of TPBS, which is akin to a modular design of PBS by Bellare and Fuchsbauer and which requires essentially the same cryptographic building blocks as the latter.  Namely, we assume the existence  of  
a signature scheme that is existentially unforgeable under chosen message attacks (EU-CMA),  a  non-interactive zero-knowledge (NIZK) proof that is simulation-sound extractable (SE), and a public-key encryption scheme that is indistinguishable under chosen ciphertext attacks (IND-CCA). Let us briefly review how our construction works. The setup algorithm generates the public parameters and key pairs for the underlying building blocks.  The signing key of a user  $\mathsf{id}$ will be a set of signatures, one per $\mathsf{id}\|p$,   for all policies that  the user  is authorized to sign. To sign a message $m$, a user $\mathsf{id}$ first encrypts  its identity,  and provides   a NIZK proof of knowledge of a signature on $\mathsf{id}\|p$ such that $p$ permits the message $m$ \emph{and} $\mathsf{id}$ has been correctly encrypted to a given ciphertext. The signature contains the ciphertext as well as the resulting NIZK proof. On input the decryption key, one can open a valid signature  and  determine the  signer's identity.  We prove  this generic construction satisfies our proposed security notions and the opening algorithm correctly reveals the user~$\mathsf{id}$.

\smallskip

\noindent 
{\bf Lattice-based instantiation. }Next, to demonstrate the feasibility of achieving TPBS from concrete and well-studied computational assumptions, we provide an instantiation of TPBS based on the hardness of the Short Integer Solution (\textsf{SIS}) problem ~\cite{Ajtai96STOC} and the Learning With Errors (\textsf{LWE}) problem~\cite{Regev05STOC}. Our scheme hence can rely on the worst-case hardness of the Shortest Independent Vector Problem ($\mathsf{SIVP}_\gamma$) over standard lattices with small polynomial approximation factors $\gamma$, and has the potential to be resistant against quantum computers. 

To design a lattice-based TPBS following the above blueprint, we need to choose suitable technical ingredients from lattices and implement some non-trivial refinements, especially for the associated zero-knowledge proof system, so that to make these ingredients work smoothly together.  
First, we seek for techniques to build SE-NIZKs for expressive languages in the lattice setting.  To this end, we employ zero-knowledge techniques for \textsf{SIS} and \textsf{LWE} statements that operate within Stern's framework~\cite{Stern96IT}, and make use of the Fiat-Shamir (FS) transform~\cite{FS86C} to obtain an SE-NIZK proof in the random oracle model~\footnote{The proof of security of this approach is given in~\cite{FKMV12IndoC}.}.    We then need to choose an appropriate policy language,  a secure signature scheme and a secure encryption scheme, that are reasonably efficient and that are compatible with the zero-knowledge layer.  
Regarding policy language, we will use the one suggested by Cheng et al.~\cite{CNW16DCC} in their construction of the first PBS from lattices.
 The authors argued that their language  captures policies that are used in real-life scenarios: 
the policy language is defined for an exponential-size message space and supports a polynomial-size policy space. Moreover, a policy can  permit many messages, and a message can simultaneously satisfy many policies. 
Concretely, a policy $\mathbf{p}\in\{0,1\}^{\ell_2}$ is said to permit a message $\mathbf{m}\in\{0,1\}^n$ if there exists a witness $\mathbf{q}\in\{0,1\}^d$ such that $\mathbf{G}_1\cdot \mathbf{p}+\mathbf{G}_2\cdot \mathbf{q}=\mathbf{m}\bmod 2$ for given matrices $\mathbf{G}_1,\mathbf{G}_2$ and appropriately  chosen parameters. 
As for EU-CMA signature scheme, we will employ 
Boyen's signature~\cite{Boyen10PKC} that features much smaller key sizes than the Bonsai signature~\cite{CHKP10EC} used in Cheng et al.'s PBS scheme.  As for the encryption layer,  we will start with  the identity-based encryption (IBE) scheme by Gentry, Peikert, and Vaikuntanathan (GPV)~\cite{GPV08STOC}, and then apply the CHK transform~\cite{CHK04EC} to get an IND-CCA secure encryption scheme.

Since the chosen ingredients (Cheng et al.'s policy language, Boyen's signature, GPV-IBE encryption) are known to be compatible with Stern-like zero-knowledge protocols~\cite{CNW16DCC,LNW15PKC}, it would be possible to obtain a combined zero-knowledge proof that will serve as the backbone of our lattice-based TPBS. We, however, need to address the following challenge. The relation that needs to be proved  by the user $\mathsf{id}$  during the signature generation, must show that the user (1)   
 possesses a valid Boyen signature on $\mathsf{id}\|\mathbf{p}$; (2) has correctly encrypted $\mathsf{id}$ using GPV-IBE;  and (3) there exists a vector $\mathbf{q}$ such that the above policy relation is satisfied for $(\mathbf{p},\mathbf{m})$. While 
 (1) and (2) have been addressed in~\cite{LNW15PKC} and (3)  has been handled in~\cite{CNW16DCC},  it is not straightforward to combine them together to establish our desirable zero-knowledge protocol. 
 Nevertheless, by carefully manipulating the underlying linear equations and applying proper techniques for proving linear constraints in Stern's framework, we are able to prove that (1), (2) and (3) are satisfied simultaneously, namely, the $\mathbf{p}$ involved in (3) was certified together with an $\mathsf{id}$ in (1) and the same $\mathsf{id}$ was encrypted in (2).  In the process, we adapt an enhanced extension-permutation technique  from~\cite{LNRW18TCS,LNWX18PKC}  that allows us to achieve \emph{optimal} permutation size equal to bit size of secret input (denoted as $|\xi|$). This improves   the  suboptimal permutation size ($\mathcal{O}(|\xi|\cdot\log |\xi|)$) if we follow the same extension-permutation techniques in~\cite{LNW15PKC,CNW16DCC} and leads to  (slightly) shorter signatures. Compared to Cheng et al.'s scheme~\cite{CNW16DCC}, our lattice-based TPBS scheme is richer in terms of functionality (with traceability enabled), relies on security assumption $\mathsf{SIVP}_\gamma$ with similar factors $\gamma$, while achieving the same level of asymptotic efficiency (more concretely, our scheme has smaller key sizes but produces slightly larger signature size due to the inclusion of the encryption layer).

\smallskip
\noindent{\sc Related work. }The study of authentication systems supporting anonymity and accountability/traceability started in the early 1980s and is still one of the major research directions nowadays. Below, let us briefly review part of the extensive literature relevant to our work. 

Group signatures (GS)~\cite{CH91EC} allow certified members of a group to anonymously sign messages on behalf of the group, but when needed, any signature can be traced.  Our TPBS can be seen as a GS with fined-grained control on who can sign a message and with simulation-based security requirements. As the opening authority in a GS can violate users' anonymity at will,  there have been several efforts to restrict its power, such as creating a tracing trapdoor for each user~\cite{KTY04EC} or for each message~\cite{SEHK0O12Pairing}, or forcing the authority to decide who to be traced in advance~\cite{KM15PoPETs}. We may as well consider these enhanced mechanisms in the context of TPBS.

Ring signatures (RS)~\cite{RST01AC} enable anonymous authentications within ad-hoc groups and originally does not support any form of user accountability. While absolute anonymity could be a desirable feature in certain scenarios, it could be abused by malicious users.  Therefore, a number of RS variants have been proposed to regulate excessive anonymity, including linkable RS~\cite{LWW04ACISP}, accountable RS~\cite{XY04} and traceable RS~\cite{FS08IEICE}. Similar methods for achieving accountability/traceability have also been deployed in the context of e-cash~\cite{Chaum82C,CHL05EC} and could also be potentially useful for PBS.

Attribute-based signatures (ABS)~\cite{MPR11CT-RSA} allow a user who owns a set of certified attributes to anonymously issue signatures whenever his attributes satisfy a given predicate. Like RS, in its original form, ABS does not support user accountability. To remedy this issue, traceable ABS~\cite{EGK14CT-RSA} was then introduced so that anonymity of misbehaving users can be revoked by a designated authority. Our enhancement here for PBS has similar spirit to that of~\cite{EGK14CT-RSA} for ABS. 

Lattice-based cryptography is currently a mainstream field of research and development, due to a number of advantages over traditional public-key cryptography from factoring and discrete logarithm, most notably, its conjectured security against quantum computers. Designing secure lattice-based anonymous authentication systems supporting accountability/traceability has been an active subfield in the last decade. In fact, one can name various lattice-based GS schemes~\cite{GKV10AC,LLLS13AC,LLNW16EC,PLS18CCS,LNWX19CT-RSA,BCOS20PQC}, linkable~\cite{LAZ19ACNS,TKSSLC19ACISP,LNYWW19ESORICS} and traceble RS~\cite{FLWL20CT-RSA} schemes, as well as compact e-cash systems~\cite{LLNW17AC-ecash,YAZXYW19C}. Lattice-based constructions of ABS, for concrete policies~\cite{BK16Eprint,ZLHZJ19CCS} and general policies~\cite{Tsabary17TCC,KK18PKC}, have also been known, but no scheme supporting traceability, e.g., in the sense of~\cite{EGK14CT-RSA}, has been proposed.

The first lattice-based PBS was put forward by Cheng et al.~\cite{CNW16DCC}. 
Their construction relies on Stern-like protocols and Bonsai signature~\cite{CHKP10EC} and allows delegation of signing.  As discussed above, our lattice-based instantiation of TPBS employs the same approach to policy definition and shares the same framework for designing zero-knowledge proofs. We, however, employ Boyen's signature~\cite{Boyen10PKC} that yields much smaller keys: public key size is reduced by a factor close to $2$ and user signing key by a factor of $(\ell_1+\ell_2)/2$, with $\ell_1,\ell_2$ being bit sizes of user identities and policies, respectively).  Yet, our model of TPBS has not offered delegatability, as suggested by Bellare and Fuchsbauer~\cite{BF14PKC} and achieved by Cheng et al.~\cite{CNW16DCC}.  Extending TPBS to provide delegatability would  be an interesting future work.
On the practical front, similar to~\cite{CNW16DCC}, our lattice-based instantiation is still not practically usable, mainly due to large signature size. Nevertheless, it would certainly enrich the field and be the first step towards more efficient constructions in the near future.

\smallskip
\noindent{\sc Organization.} The rest of the paper is structured as follows. 
 Section~\ref{tpbs-section:tpbs-security-model} introduces and discusses the notion of TPBS and its security requirements. Section~\ref{tpbs-section:tpbs-generic-con} presents our generic construction and its security proofs.  We provide our lattice-based instantiation of TPBS and its underlying zero-knowledge protocol in Section~\ref{tpbs-section:tpbs-lattice-con} and  Section~\ref{tpbs-subsection:main-zk-protocol}, respectively. In Section~\ref{tpbs-section:tpbs-conclusions}, we conclude the paper and mention a few interesting open problems.  Some supplementary materials are deferred to the Appendix. 

\section{Traceable Policy-Based  Signatures}\label{tpbs-section:tpbs-security-model}
{\bf {Notations.}} Let $\mathbb{Z}^+$ denote the set of all positive integers. For $a,b\in \mathbb{Z}^+$, denote $[a,b]$ as the set $\{a,a+1,\ldots,b\}$. In the case where $a=1$ we will simply write $[b]$. All vectors considered in this work are column vectors, unless otherwise stated. When $\mathbf{a}\in \mathbb{R}^n, \mathbf{b}\in \mathbb{R}^m$, for simplicity, we denote $[\mathbf{a}^{\top}|\mathbf{b}^{\top}]^{\top}$ as $(\mathbf{a}\|\mathbf{b})\in \mathbb{R}^{n+m}$. 
 A policy checker $\mathsf{PC}$~\cite{BF14PKC} is an $\mathcal{NP}$-relation $\{0,1\}^*\times \{0,1\}^*$ with the first input being a pair $(p,m)$ representing  a policy $p\in\{0,1\}^*$ and a message $m\in\{0,1\}^*$,  while the second input being a witness $w$.   The associated policy language is  $\mathcal{L}(\mathsf{PC})=\{ (p,m): \exists ~w\in \{0,1\}^*,~\text{such that} ~\mathsf{PC}\big((p,m),w\big)=1\}$. \remove{If $(p,m)\in \mathcal{L}(\mathsf{PC})$, we say that either $m$ satisfies policy $p$ or $p$ permits $m$. The space  of user identities  is denoted as $\mathcal{ID}\subseteq\{0,1\}^*$. For each user $\mathsf{id}$, let $\mathcal{P}_{\mathsf{id}}$ be a set of policies that  user  $\mathsf{id}$ is authorized  to sign. }
\smallskip 

\noindent{\textbf{Syntax of TPBS.}} A traceable policy-based signature (TPBS) scheme extends  PBS  with an  additional feature that the identity of the signer of any signature can be revealed. This is achieved by including the user's identity in relevant algorithms and providing an additional opening algorithm to uncover the identity of the signer of any signature.  A TPBS scheme consists of the following polynomial-time algorithms. 
\begin{description}
	\item[$\mathsf{Setup}(1^{\lambda})$:] This algorithm takes as input $1^{\lambda}$, where $\lambda$ is the security parameter, and outputs  public parameter $\mathsf{pp}$, a master secret key $\mathsf{msk}$, and a master decryption key $\mathsf{mdk}$. 
	
	\item[$\mathsf{KeyGen}(\mathsf{msk},\mathsf{id},\mathcal{P}_{\mathsf{id}})$:] On inputs $\mathsf{msk}$, a user identity $\mathsf{id}\in\{0,1\}^*$, and  a set of policies $\mathcal{P}_{\mathsf{id}}$ on which the user is able to sign messages, this algorithm outputs  a key $\mathsf{usk}_{\mathsf{id}}$ for user $\mathsf{id}$ on all polices in $\mathcal{P}_{\mathsf{id}}$.

	\item[$\mathsf{Sign}(\mathsf{usk}_{\mathsf{id}},m,w)$:] It takes $\mathsf{usk}_{\mathsf{id}}$ of user $\mathsf{id}$, a message $m\in \{0,1\}^*$, and a witness $w\in \{0,1\}^*$ as inputs, and outputs a signature $\sigma$ or $\bot$ if it fails.  
	
	\item[$\mathsf{Verify}(\mathsf{pp},m,\sigma)$:] This algorithm takes $\mathsf{pp}$, $m$, and $\sigma$ as inputs, and outputs~$1$ or~$0$, indicating the validity of the signature $\sigma$ on  message $m$. 
	
	\item[$\mathsf{Open}(\mathsf{mdk},m,\sigma)$:] This algorithm takes $\mathsf{mdk}$, $m$, and $\sigma$ as inputs, and outputs an identity $\mathsf{id}$, or it fails and outputs $\bot$. 
\end{description}
\noindent\textbf{Correctness of TPBS.}  The scheme is said to be correct with respect to a policy checker $\mathcal{PC}$, if for all $\lambda$, all $(\mathsf{pp},\mathsf{msk},\mathsf{mdk})\leftarrow\mathsf{Setup}(1^{\lambda})$, all $\mathsf{id}\in\mathcal{ID}$,   $\mathsf{usk}_{\mathsf{id}}\leftarrow \mathsf{KeyGen}(\mathsf{msk},\mathsf{id},\mathcal{P}_{\mathsf{id}})$, all $(m,w)$ such that $\exists\hspace*{2pt}p\in \mathcal{P}_{\mathsf{id}}, \mathsf{PC}\big((p,m),w\big)=1$, for all $\sigma\leftarrow \mathsf{Sign}(\mathsf{usk}_{\mathsf{id}},m,w)$, we have
\begin{eqnarray*}
\mathsf{Verify}(\mathsf{pp},m,\sigma)=1~~\text{and}~~ \mathsf{Open}(\mathsf{msk},m,\sigma)=\mathsf{id}. 
\end{eqnarray*} 
{\bf{Discussion.}} TPBS extends PBS to include user identities in the $\mathsf{KeyGen}$, $\mathsf{Sign}$ functions, and  adds the $\mathsf{Open}$ function to allow tracing of a signature to its signer.  Security requirements of TPBS are aligned with PBS, ensuring privacy of the policy that is used for signing the message, and unforgeability of the signature. Bellare and Fuchsbauer~\cite{BF14PKC} argued that the traditional notions of indistinguishability and unforgeability are not sufficient for PBS for some applications,  and introduced stronger notions of simulatability and extractability instead. Their argument can be extended to TPBS leading to simulatability and extractability as appropriate security notions. 

Privacy of a TPBS scheme demands that a signature should not  reveal the identity or the  policy that is associated with  the signing key, 
nor it should leak information about the witness used. That is the following indistinguishability conditions must hold: 
(1) for an identity $\mathsf{id}$,  two signatures on a message $m$ generated under two conforming  policies $p_0, p_1$ with witness $w_0,w_1$,  should be indistinguishable; 
(2) for a policy $p$, two signatures on a message $m$ with witness $w$ (that satisfies policy $p$), generated by two users with identities
$\mathsf{id}_0, \mathsf{id}_1$, and  both are authorized to sign $m$ under $p$, cannot
be distinguished. 
Using an argument similar to PBS we note that  there may exist only one policy $p$ for a message $m$ and so the indistinguishability-based definition will not be able to  hide the policy. 
Simulatability-based definition however addresses this problem: it  requires  no $\mathrm{PPT}$ adversary be able to distinguish a simulated signature from a legitimately signed signature. We will use this notion for TPBS schemes.

\smallskip

\begin{figure}[ht]
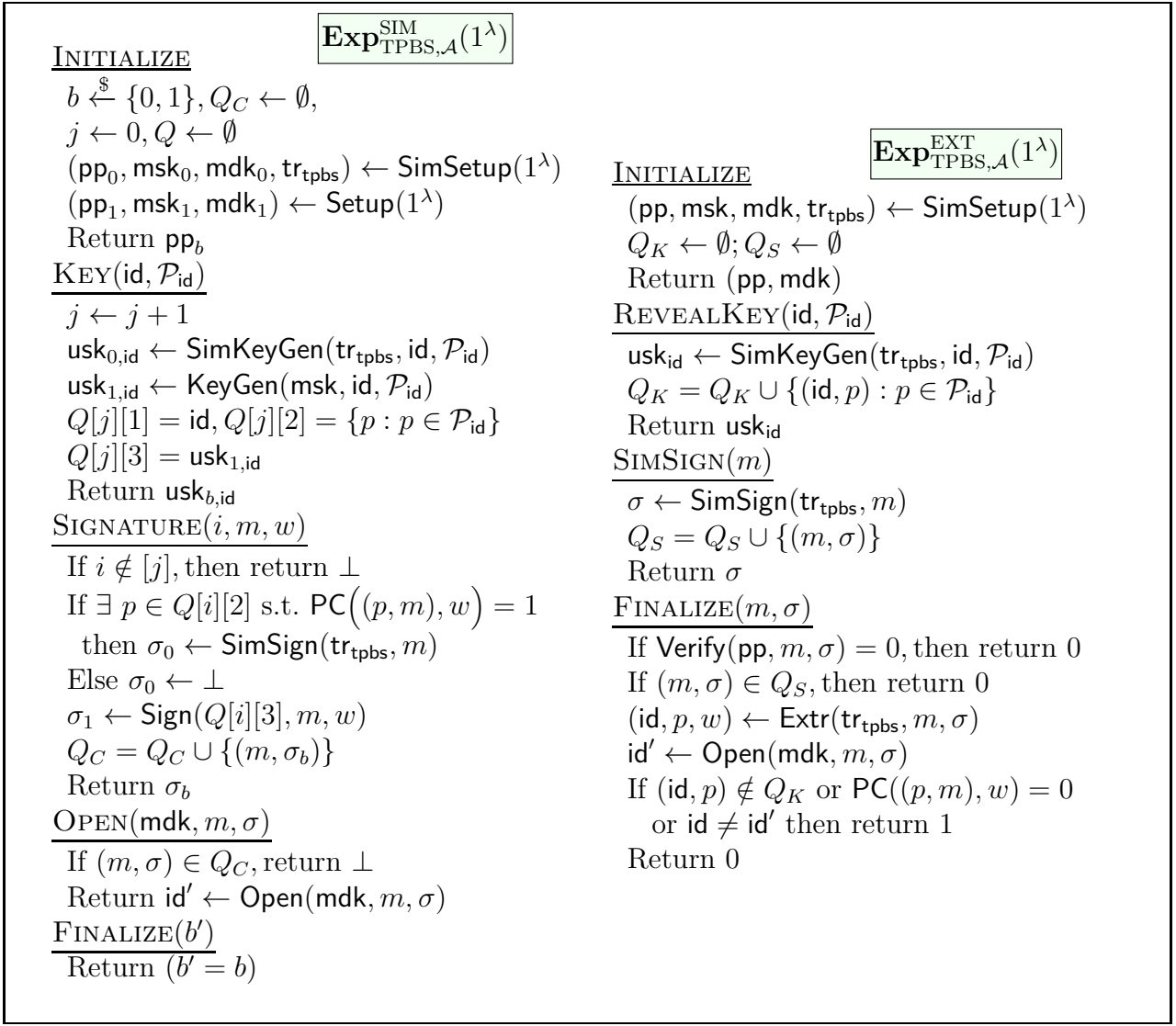

	\begin{center}
		\begin{tabular}{cc}
			\begin{minipage}{7.6cm}
				\underline{\sc Initialize} \hspace*{1.6cm}\tikz \draw[fill=green!5] (2.8,-0.3) rectangle (0.0168cm, 12pt) node[pos=0.5]{$\mathbf{Exp}_{\mathrm{TPBS},\mathcal{A}}^{\mathrm{SIM}}(1^{\lambda})$};\\
				$\hspace*{6pt} b\xleftarrow{\$}\{0,1\}, Q_C\leftarrow \emptyset$, \\
				$\hspace*{6pt}j\leftarrow 0, Q\leftarrow\emptyset$\\
				$\hspace*{6pt}(\mathsf{pp}_0,\mathsf{msk}_0,\mathsf{mdk}_0,\mathsf{tr}_{\mathsf{tpbs}})\leftarrow \mathsf{SimSetup}(1^\lambda)$\\
				$\hspace*{6pt}(\mathsf{pp}_1,\mathsf{msk}_1,\mathsf{mdk}_1)\leftarrow\mathsf{Setup}(1^\lambda)$\\
				$\hspace*{6pt}\text{Return}~\mathsf{pp}_b$
				
				\underline{\sc Key$(\mathsf{id},\mathcal{P}_{\mathsf{id}})$}\vspace*{2.2pt}\\
				$\hspace*{6pt}j\leftarrow j+1$\\
				$\hspace*{6pt}\mathsf{usk}_{0,\mathsf{id}}\leftarrow \mathsf{SimKeyGen}(\mathsf{tr}_{\mathsf{tpbs}},\mathsf{id},\mathcal{P}_{\mathsf{id}})$\\
				$\hspace*{6pt}\mathsf{usk}_{1,\mathsf{id}}\leftarrow \mathsf{KeyGen}(\mathsf{msk},\mathsf{id},\mathcal{P}_{\mathsf{id}})$\\
				$\hspace*{6pt}Q[j][1]=\mathsf{id}, Q[j][2]=\{p:p\in \mathcal{P}_{\mathsf{id}}\}$\\
				$\hspace*{6pt}Q[j][3]=\mathsf{usk}_{1,\mathsf{id}}$\\
				$\hspace*{6pt}\text{Return}~\mathsf{usk}_{b,\mathsf{id}}$

				\underline{\sc Signature$(i,m,w)$}\vspace*{2.2pt}\\
				$\hspace*{6pt}\text{If}~ i\notin [j], \text{then~return}~\bot$\\
				$\hspace*{6pt}\text{If}~\exists~p\in Q[i][2]~\text{s.t.}~\mathsf{PC}\big((p,m),w\big)=1 $ \\
				$\hspace*{12pt}\text{then}~\sigma_0\leftarrow\mathsf{SimSign}(\mathsf{tr}_{\mathsf{tpbs}},m)$ \\
				$\hspace*{6pt}\text{Else} ~\sigma_0\leftarrow \bot$ \\
				$\hspace*{6pt}\sigma_1\leftarrow\mathsf{Sign}(Q[i][3],m,w)$\\
				$\hspace*{6pt}Q_C=Q_C\cup \{(m,\sigma_b)\}$\\
				$\hspace*{6pt} \text{Return}~\sigma_b$

				\underline{\sc Open$(\mathsf{mdk},m,\sigma)$}\vspace*{2.2pt}\\
				$\hspace*{6pt}\text{If}~(m,\sigma)\in Q_C,\text{return}~\bot$\\
				$\hspace*{6pt}\text{Return}~\mathsf{id}'\leftarrow\mathsf{Open}(\mathsf{mdk},m,\sigma)$

				\underline{\sc Finalize$(b')$}\\
				$\hspace*{6pt}\text{Return}~(b'=b)$
			\end{minipage}
			&	
			\begin{minipage}{7.3cm}
				
				\underline{\sc Initialize}\hspace*{1.6cm}\tikz \draw[fill=green!5] (2.9,-0.3) rectangle (0.168cm, 12pt) node[pos=0.5]{$\mathbf{Exp}_{\mathrm{TPBS},\mathcal{A}}^{\mathrm{EXT}}(1^{\lambda})$};\\
				$\hspace*{6pt}(\mathsf{pp},\mathsf{msk},\mathsf{mdk},\mathsf{tr}_{\mathsf{tpbs}})\leftarrow\mathsf{SimSetup}(1^{\lambda})$\\
				$\hspace*{6pt}Q_K\leftarrow \emptyset; Q_S\leftarrow\emptyset$\\ $\hspace*{6pt}\text{Return}~(\mathsf{pp},\mathsf{mdk})$ 
				
				\underline{\sc RevealKey$(\mathsf{id},\mathcal{P}_{\mathsf{id}})$}\vspace*{2.2pt}\\
				$\hspace*{6pt}\mathsf{usk}_{\mathsf{id}}\leftarrow\mathsf{SimKeyGen}(\mathsf{tr}_{\mathsf{tpbs}},\mathsf{id},\mathcal{P}_{\mathsf{id}})$\\
				$\hspace*{6pt}Q_K=Q_K\cup \{(\mathsf{id},p):p\in \mathcal{P}_{\mathsf{id}}\}$\\  $\hspace*{6pt}\text{Return}~\mathsf{usk}_{\mathsf{id}}$
				
				\underline{\sc SimSign$(m)$}\vspace*{2.2pt}\\
				$\hspace*{6pt} \sigma\leftarrow\mathsf{SimSign}(\mathsf{tr}_{\mathsf{tpbs}},m)$\\
				$\hspace*{6pt}Q_S=Q_S\cup \{(m,\sigma)\}$\\
				$\hspace*{6pt}\text{Return}~\sigma$
				
				\underline{\sc Finalize$(m,\sigma)$}\vspace*{2.2pt}\\
				$\hspace*{6pt}\text{If}~\mathsf{Verify}(\mathsf{pp},m,\sigma)=0,\text{then~return}~0$\\
				$\hspace*{6pt}\text{If}~(m,\sigma)\in Q_S,\text{then~return}~0$\\
				$\hspace*{6pt}(\mathsf{id},p,w)\leftarrow \mathsf{Extr}(\mathsf{tr}_{\mathsf{tpbs}},m,\sigma)$\\
				$\hspace*{6pt}\mathsf{id}'\leftarrow \mathsf{Open}(\mathsf{mdk},m,\sigma)$\\
				$\hspace*{6pt}\text{If}~(\mathsf{id},p)\notin Q_K~\text{or}~\mathsf{PC}((p,m),w)=0$ \\
				$\hspace*{12pt}~\text{or}~\mathsf{id}\neq \mathsf{id}'~\text{then~return}~1$\\
				$\hspace*{6pt}\text{Return}~0$
				
			\end{minipage}
		\end{tabular}
	\end{center}
	\caption{Games defining simulatability and extractability of TPBS.}
	\label{fbgs-fig:sim-ext}
\end{figure}

\noindent{\bf Simulatability.} This notion requires that one cannot distinguish a signature generated by a simulator without having access to either the signing key of any identity or witness from a legitimately signed signature. To define simulatability,  we require simulated algorithms $\mathsf{SimSetup}$, $\mathsf{SimKeyGen}$, and $\mathsf{SimSign}$  as in~\cite{BF14PKC}. Algorithm $\mathsf{SimSetup}$ outputs $(\mathsf{pp},\mathsf{msk},\mathsf{mdk},\mathsf{tr}_{\mathsf{tpbs}})$ such that $\mathsf{pp}$ is indistinguishable from the one generated by  $\mathsf{Setup}$. Algorithm $\mathsf{SimKeyGen}$   outputs keys indistinguishable from those produced  by $\mathsf{KeyGen}$. On inputs trapdoor  $\mathsf{tr}_{\mathsf{tpbs}}$  and a message, algorithm $\mathsf{SimSign}$ outputs a signature  indistinguishable from that honestly produced by $\mathsf{Sign}$.  Details of the requirement are modeled in  the experiment $\mathbf{Exp}_{\mathrm{TPBS},\mathcal{A}}^{\mathrm{SIM}}(1^{\lambda})$ in Figure~\ref{fbgs-fig:sim-ext}.  The differences between this  experiment and Bellare and Fuchsbauer's corresponding experiment~\cite[Figure~$2$]{BF14PKC} are the following: (1) we include user identity in the relevant algorithms; (2) we provide an opening oracle $\text{\sc Open}(\mathsf{mdk},\cdot,\cdot)$ to the adversary and require that the queried message signature pair to this oracle is not from the challenge oracle $\text{\sc Signature}(\cdot,\cdot,
\cdot)$. Here (1) is a natural extension of PBS to TPBS, and (2) captures the possibility of an adversary seeing the results of previous openings. Note that the output of  $\mathsf{SimSign}$  does not depend on the user identity (no $\mathsf{id}$ related input) while the output of  $\mathsf{Sign}$ indeed relies on the user identity. Therefore, leaking $\mathsf{mdk}$ to the adversary enables it to run the algorithm $\mathsf{Open}$ and to distinguish a simulated signature  from an honestly generated one trivially. 
\smallskip

Define $\mathbf{Adv}_{\mathrm{TPBS},\mathcal{A}}^{\mathrm{SIM}}(1^{\lambda})=|\mathrm{Pr}[\mathbf{Exp}_{\mathrm{TPBS},\mathcal{A}}^{\mathrm{SIM}}(1^{\lambda})=1]-\frac{1}{2}|$ as the advantage of an adversary $\mathcal{A}$ against simulatability with $\mathbf{Exp}_{\mathrm{TPBS},\mathcal{A}}^{\mathrm{SIM}}(1^{\lambda})$ defined in Figure~\ref{fbgs-fig:sim-ext}. A TPBS scheme is said to be simulatable if  $\mathbf{Adv}_{\mathrm{TPBS},\mathcal{A}}^{\mathrm{SIM}}(1^{\lambda})$ is negligible in $\lambda$ for all  $\mathrm{PPT}$ adversary $\mathcal{A}$.

\smallskip
\noindent{\bf Discussion.} Using a typical unforgeability notion for TPBS would yield the same difficulty that was noted in the case of PBS.  More specifically, the experiment that defines unforgeability will need  the checking of  $(p,m)\in \mathcal{L}(\mathsf{PC})$ to determine if the  adversary has won the game. A problem that may arise is when the proof uses  game hopping, and between two games a distinguisher must efficiently determine whether an adversary has won the game. Using the stronger notion of extractability remedies this problem. 

\smallskip

\noindent {\bf Extractability.}  This is defined by requiring that   whenever a $\mathrm{PPT}$ adversary $\mathcal{A}$ outputs a valid message signature pair $(m,\sigma)$ that is not obtained from an oracle,  there exists an extractor $\mathsf{Extr}$ that  uses trapdoor $\mathsf{tr}_{\mathsf{tpbs}}$ to extract a tuple $(\mathsf{id},p,w)$ such that $\mathcal{A}$ must have queried the key for $(\mathsf{id},p)$, $\mathsf{PC}\big((p,m),w\big)=~1$, and  $\mathsf{Open}(\mathsf{mdk},m,\sigma)= \mathsf{id}$. 
These requirements  are modeled in~$\mathbf{Exp}_{\mathrm{TPBS},\mathcal{A}}^{\mathrm{EXT}}(1^{\lambda})$ in Figure~\ref{fbgs-fig:sim-ext}, where $\mathcal{A}$ receives $(\mathsf{pp},\mathsf{mdk})$ generated by the algorithm $\mathsf{SimSetup}$. In addition, $\mathcal{A}$ can obtain a simulated  user signing key by querying $\text{\sc RevealKey}$ on input $(\mathsf{id},\mathcal{P}_{\mathsf{id}})$, and a simulated  signature on a message $m$ by querying {\sc SimSign}.    Note that compared to the corresponding experiment in~\cite[Figure~$2$]{BF14PKC}, we  include the user identity in the relevant algorithms  and   specify a new case $\mathsf{Open}(\mathsf{mdk},m,\sigma)\neq \mathsf{id}$ in the {\sc Finalize} step that captures the inability of  a $\mathrm{PPT}$ adversary $\mathcal{A}$, who has queried signing keys for a set of users, to  output a valid signature that cannot be opened or  traced to a member of the queried user group. Algorithm $\mathsf{Open}$ serves as a mechanism to prevent misuse of signing ability and enforces user accountability.  

\smallskip

 Define $\mathbf{Adv}_{\mathrm{TPBS},\mathcal{A}}^{\mathrm{EXT}}(1^{\lambda})=\mathrm{Pr}[\mathbf{Exp}_{\mathrm{TPBS},\mathcal{A}}^{\mathrm{EXT}}(1^{\lambda})=1]$ as the advantage of an adversary $\mathcal{A}$ against extractability with the experiment $\mathbf{Exp}_{\mathrm{TPBS},\mathcal{A}}^{\mathrm{EXT}}(1^{\lambda})$  defined in Figure~\ref{fbgs-fig:sim-ext}. A TPBS scheme is  said to be extractable if  $\mathbf{Adv}_{\mathrm{TPBS},\mathcal{A}}^{\mathrm{EXT}}(1^{\lambda})$ is negligible in $\lambda$ for all $\mathrm{PPT}$ adversary $\mathcal{A}$.

\section{Generic Construction of Traceable Policy-Based Signatures} \label{tpbs-section:tpbs-generic-con}

We present a generic construction of TPBS for any $\mathcal{NP}$-relation $\mathsf{PC}$ in Section~\ref{tpbs-subsection:gener-con} and show its simulatability and extractability in Section~\ref{tpbs-subsection:security-proofs}.  Our construction relies on an EU-CMA secure signature scheme, an IND-CCA secure public-key encryption scheme and an SE-NIZK proof system. The standard definitions of these primitives are recalled in~\ref{tpbs-subsection:primitive-for-gener-con}.

\subsection{Generic Construction}\label{tpbs-subsection:gener-con}
To construct a traceable policy-based signature that satisfies simulatability and extractability, our starting point is the generic construction of PBS using SE-NIZK proof~\cite[Figure~$4$]{BF14PKC}. In this construction, the issuer first uses $\mathsf{Setup}$ algorithm to generate a signature key pair $(\mathsf{mvk},\mathsf{msk})$ and a common reference string $\mathsf{crs}$ for an SE-NIZK proof system $\Pi$, and makes $(\mathsf{mvk},\mathsf{crs})$ public. It then runs  $\mathsf{KeyGen}$ to generate a signing key for a policy $p$ by generating a signature on $p$ using $\mathsf{msk}$. A user that holds  a key for the policy $p$  can sign a message $m$ by providing a  zero-knowledge proof $\pi$ that shows possession of a policy $p$ satisfying $(p,m)\in \mathcal{L}(\mathsf{PC})$ and a signature on $p$ which is verifiable under $\mathsf{mvk}$.  The actual signature is $\pi$.

In TPBS, users have identities and will receive keys that are signatures on $\mathsf{id}\|p$ for  all policies that they are authorized to sign.  The signing algorithm uses $\mathsf{id}$ as an  input, and must ensure that the signature can be ``opened''.  To this end, we require that algorithm $\mathsf{Setup}$  generates an encryption key pair $(\mathsf{mek},\mathsf{mdk})$,  and that user first encrypts its identity using $\mathsf{mek}$ when signing a  message $m$ and then proves in zero-knowledge that its encrypted identity is the same as the one for which a valid signing key is known (i.e. a signature on $\mathsf{id}\|p$ is known) such that $(p,m)\in \mathcal{L}(\mathsf{PC})$.  The algorithm $\mathsf{Open}$ on input $\mathsf{mdk}$ and a message signature pair outputs an identity, specifying the originator of the signature.

To define simulatability and extractability, we require four additional algorithms $\mathsf{SimSetup}$, $\mathsf{SimKeyGen}$, $\mathsf{SimSign}$, and $\mathsf{Extr}$. Algorithm $\mathsf{SimSetup}$ is the same as $\mathsf{Setup}$ except that it runs simulated setup algorithm $\mathsf{SimSetup}_{\mathsf{nizk}}$  of the proof system $\Pi$, obtaining a simulated $\mathsf{crs}$ and $\mathsf{tr}$. Algorithm $\mathsf{SimKeyGen}$ is the same as $\mathsf{KeyGen}$.  
For algorithm $\mathsf{SimSign}$, we first encrypt a dummy identity $\mathbf{0}$ and then employ the trapdoor $\mathsf{tr}$ to run $\mathsf{SimProve}$ of the proof system $\Pi$. For algorithm $\mathsf{Extr}$, we run  $\mathsf{Extr}_{\mathsf{nizk}}$ of  the proof system $\Pi$ by utilizing $\mathsf{tr}$ as well. In the following we formalize the above approach. 
For a policy checker $\mathsf{PC}$, define an $\mathcal{NP}$-relation $\rho_\mathsf{tpbs}$ as follows: 
\begin{eqnarray*}
&&\big ((\mathsf{mek},\mathsf{mvk},m,\mathsf{ct}),(\mathsf{id},p,\mathsf{cert}_{\mathsf{id}\|p},w,\mathbf{r})\big) \in \rho_\mathsf{tpbs} \\
\Longleftrightarrow &&\mathsf{ct}=\mathsf{Enc}(\mathsf{mek},\mathsf{id}; \mathbf{r}) \wedge \mathsf{Verify}_{\mathsf{sig}}(\mathsf{mvk},\mathsf{id}\|p,\mathsf{cert}_{\mathsf{id}\|p})=1 \wedge \mathsf{PC}\big((p,m),w\big)=1
\end{eqnarray*}
Let $\mathrm{SIG}=(\mathsf{KeyGen}_{\mathsf{sig}},\mathsf{Sign}_{\mathsf{sig}}, \mathsf{Verify}_{\mathsf{sig}})$ be a signature scheme that is EU-CMA secure, let $\mathrm{PKE}=(\mathsf{KeyGen}_{\mathsf{pke}},\mathsf{Enc},\mathsf{Dec})$ be a public-key encryption scheme that satisfies IND-CCA security, and let $\Pi=(\mathsf{Setup}_{\mathsf{nizk}},\mathsf{Prove},\mathsf{Verify}_{\mathsf{nizk}},\mathsf{SimSetup}_{\mathsf{nizk}},\mathsf{SimProve},\mathsf{Extr}_{\mathsf{nizk}})$ be an SE-NIZK proof for relation~$\rho_{\mathsf{tpbs}}$. 
Our construction of TPBS scheme is depicted in Figure~\ref{tpbs-fig: generic-construction}. 

\smallskip
\noindent{\bf Correctness.} Correctness of our generic construction directly follows from completeness of the underlying proof system, signature scheme, and encryption scheme. 

\begin{figure}[ht]
		\begin{center}
		\begin{tabular}{cc}
			\begin{minipage}{8cm}
				\underline{$\mathsf{Setup}(1^{\lambda})$}\\
				$\hspace*{6pt}\mathsf{crs}\leftarrow\mathsf{Setup}_{\mathsf{nizk}} (1^\lambda)$  \\
				$\hspace*{6pt}(\mathsf{mek},\mathsf{mdk})\leftarrow\mathsf{KeyGen}_{\mathsf{pke}}(1^{\lambda})$\\
				$\hspace*{6pt} (\mathsf{mvk},\mathsf{msk})\leftarrow\mathsf{KeyGen}_{\mathsf{sig}}(1^{\lambda})$\\
				$\hspace*{6pt}\text{Return}~\mathsf{pp}\leftarrow(\mathsf{crs},\mathsf{mek},\mathsf{mvk}), \mathsf{msk}, \mathsf{mdk}$
				
				\underline{$\mathsf{KeyGen}(\mathsf{msk},\mathsf{id},\mathcal{P}_{\mathsf{id}})$}\vspace*{2.2pt}\\
				$\hspace*{6pt}\forall~p\in \mathcal{P}_{\mathsf{id}},\text{compute}$\\ $\hspace*{12pt}\mathsf{cert}_{\mathsf{id}\|p} \leftarrow \mathsf{Sign}_{\mathsf{sig}}(\mathsf{msk},\mathsf{id}\|p)$\\
				$\hspace*{6pt}\text{Set}~\mathsf{usk}_{\mathsf{id}}\leftarrow (\mathsf{id},\{(p, \mathsf{cert}_{\mathsf{id}\|p}): p\in \mathcal{P}_{\mathsf{id}} \})$ \\
				$\hspace*{6pt}\text{Return}~\mathsf{usk}_{\mathsf{id}}$ 
				
				\underline{$\mathsf{Sign}(\mathsf{usk}_{\mathsf{id}},m,w)$}\vspace*{2.2pt}\\
				$\hspace*{6pt}\text{Parse}~\mathsf{usk}_{\mathsf{id}}= (\mathsf{id},\{(p, \mathsf{cert}_{\mathsf{id}\|p}): p\in \mathcal{P}_{\mathsf{id}}\})$\\
				$\hspace*{6pt}\text{If}~\exists~p\in \mathcal{P}_{\mathsf{id}},~\text{s.t.}~\mathsf{PC}\big( (p,m),w\big)=1$\\
				$\hspace*{12pt}\mathbf{r}\leftarrow\{0,1\}^{\mathsf{poly}(\lambda)}, \mathsf{ct}\leftarrow \mathsf{Enc}(\mathsf{mek,\mathsf{id};\mathbf{r}})$\\
				$\hspace*{12pt}\pi \leftarrow\mathsf{Prove}\big(\mathsf{crs}, (\mathsf{mek},\mathsf{mvk},m,\mathsf{ct})$,\\
				$\hspace*{96pt}(\mathsf{id}, p,\mathsf{cert}_{\mathsf{id}\|p},w,\mathbf{r})\big)$\\
				$\hspace*{12pt}\text{Return}~\sigma\leftarrow (\mathsf{ct},\pi)$\\
				$\hspace*{6pt}\text{Else~return}~\bot$
				
				\underline{$\mathsf{Verify}(\mathsf{pp},m,\sigma)$} \vspace*{2.2pt}\\
				$\hspace*{6pt}\text{Parse}~\mathsf{pp}=(\mathsf{crs},\mathsf{mek},\mathsf{mvk}), \sigma=(\mathsf{ct},\pi)$\\
				$\hspace*{6pt}\text{Return}~\mathsf{Verify}_{\mathsf{nizk}}(\mathsf{crs},(\mathsf{mek},\mathsf{mvk},m,\mathsf{ct}),\pi)$
				
				\underline{$\mathsf{Open}(\mathsf{mdk},m,\sigma)$}\vspace*{2.2pt}\\
					$\hspace*{6pt}\text{If}~\mathsf{Verify}(\mathsf{pp},m,\sigma)=0,\text{return}~\bot$\\
				$\hspace*{6pt}\text{Parse}~\sigma=(\mathsf{ct},\pi)$\\
				$\hspace*{6pt}\text{Else~return}~\mathsf{Dec}(\mathsf{mdk},\mathsf{ct})$			
			\end{minipage}
			& 
			\begin{minipage}{8cm}
				\underline{$\mathsf{SimSetup}(1^{\lambda})$}\\
				$\hspace*{6pt}(\mathsf{crs},\mathsf{tr})\leftarrow\mathsf{SimSetup}_{\mathsf{nizk}} (1^\lambda)$  \\
				$\hspace*{6pt}(\mathsf{mek},\mathsf{mdk})\leftarrow\mathsf{KeyGen}_{\mathsf{pke}}(1^{\lambda})$\\
				$\hspace*{6pt} (\mathsf{mvk},\mathsf{msk})\leftarrow\mathsf{KeyGen}_{\mathsf{sig}}(1^{\lambda})$\\
				$\hspace*{6pt}\text{Return}\hspace*{2pt}\mathsf{pp}\hspace*{-2pt}\leftarrow\hspace*{-2pt}(\mathsf{crs},\mathsf{mek},\mathsf{mvk})$\\
				 $\hspace*{42pt}\mathsf{msk}, \mathsf{mdk}$\\
				$\hspace*{42pt}\mathsf{tr}_{\mathsf{tpbs}}\leftarrow(\mathsf{msk},\mathsf{mdk},\mathsf{tr})$

				\underline{$\mathsf{SimKeyGen}(\mathsf{tr}_{\mathsf{tpbs}},\mathsf{id},\mathcal{P}_{\mathsf{id}})$}\vspace*{2.2pt}\\
				$\hspace*{6pt}\text{Parse}~\mathsf{tr}_{\mathsf{tpbs}}=(\mathsf{msk},\mathsf{mdk},\mathsf{tr})$\\
				$\hspace*{6pt}\forall~p\in \mathcal{P}_{\mathsf{id}},\text{compute}$\\
				$\hspace*{12pt}\mathsf{cert}_{\mathsf{id}\|p} \leftarrow \mathsf{Sign}_{\mathsf{sig}}(\mathsf{msk},\mathsf{id}\|p)$\\
				$\hspace*{6pt}\text{Set}~\mathsf{usk}_{\mathsf{id}}\leftarrow (\mathsf{id},\{(p, \mathsf{cert}_{\mathsf{id}\|p}): p\in \mathcal{P}_{\mathsf{id}} \})$ \\
				$\hspace*{6pt}\text{Return}~\mathsf{usk}_{\mathsf{id}}$

				\underline{$\mathsf{SimSign}(\mathsf{tr}_{\mathsf{tpbs}}, m)$}\vspace*{2.2pt}\\
				$\hspace*{6pt}\text{Parse}~\mathsf{tr}_{\mathsf{tpbs}}=(\mathsf{msk},\mathsf{mdk},\mathsf{tr})$\\
				$\hspace*{6pt}\mathbf{r}\leftarrow\{0,1\}^{\mathsf{poly}(\lambda)}, \mathsf{ct}\leftarrow \mathsf{Enc}(\mathsf{mek,\mathbf{0};\mathbf{r}})$\\
				$\hspace*{6pt}\pi \leftarrow\mathsf{SimProve}\big(\mathsf{crs},\mathsf{tr}, (\mathsf{mek},\mathsf{mvk},m,\mathsf{ct})\big)$\\
				$\hspace*{6pt}\text{Return}~\sigma\leftarrow (\mathsf{ct},\pi)$
				
				\underline{$\mathsf{Extr}( \mathsf{tr}_{\mathsf{tpbs}},m,\sigma)$} \vspace*{2.2pt}\\
				$\hspace*{6pt}\text{Parse}~\mathsf{tr}_{\mathsf{tpbs}}=(\mathsf{msk},\mathsf{mdk},\mathsf{tr}), \sigma=(\mathsf{ct},\pi)$\\
				$\hspace*{6pt}(\mathsf{id},p,\mathsf{cert}_{\mathsf{id}\|p},w,\mathbf{r})\leftarrow $\\
				$\hspace*{42pt}\mathsf{Extr}_{\mathsf{nizk}}(\mathsf{tr},(\mathsf{mek},\mathsf{mvk},m,\mathsf{ct}),\pi)$\\
				$\hspace*{6pt}\text{Return}~(\mathsf{id},p,w)$

			\end{minipage}
		\end{tabular}
	\end{center}
	\caption{Generic construction of TPBS based on SE-NIZK. }\label{tpbs-fig: generic-construction}
\end{figure}	

\subsection{Security Analysis}\label{tpbs-subsection:security-proofs}
We prove extractability and simulatability of our scheme in Theorem~\ref{tpbs-thm: EXT} and Theorem~\ref{tpbs-thm: SIM}, respectively. \remove{Proof of extractability is similar to that of~\cite{BF14PKC} except that we add a new case  $\mathsf{Dec}(\mathsf{mdk},\mathsf{ct})\neq \mathsf{id}$ in Type 2 adversary. Due to the encryption of the dummy identity~$\mathbf{0}$ and the opening oracle $\text{\sc Open}(\mathsf{mdk},\cdot,
\cdot)$, the proof of simulatability requires additional security assumptions as compared to~\cite{BF14PKC}. The details are described  in Theorem~\ref{tpbs-thm: EXT} and Theorem~\ref{tpbs-thm: SIM} below. }

\begin{theorem}\label{tpbs-thm: EXT}
	Our construction of traceable policy-based signature scheme depicted in Figure~\ref{tpbs-fig: generic-construction} satisfies extractability as modeled in  $\mathbf{Exp}_{\mathrm{TPBS},\mathcal{A}}^{\mathrm{EXT}}(1^{\lambda})$ in Figure~\ref{fbgs-fig:sim-ext}, if the underlying signature scheme SIG is  EU-CMA secure,  and the proof $\Pi$ is  simulation-sound extractable. 
\end{theorem}
\begin{proof}
	We reduce extractability to the unforgeability of the underlying signature scheme $\mathrm{SIG}$ and simulation-sound extractability of the proof system $\Pi$.  Note that in the {\sc Finalize} step of experiment $\mathbf{Exp}_{\mathrm{TPBS},\mathcal{A}}^{\mathrm{EXT}}(1^{\lambda})$, algorithm $\mathsf{Extr}(\mathsf{tr}_{\mathsf{tpbs}},m,(\mathsf{ct},\pi))$  actually runs $\mathsf{Extr}_{\mathsf{nizk}}(\mathsf{tr},m,(\mathsf{ct},\pi))$; denoting the output by $(\mathsf{id},p,\mathsf{cert}_{\mathsf{id}\|p},w,\mathbf{r})$. We distinguish two types of adversary. 
	\begin{description}
		\item[Type\hspace*{1pt}$1$]\hspace*{-3.6pt} $\mathsf{Verify}_{\mathsf{sig}}(\mathsf{mvk},\mathsf{id}\|p,\mathsf{cert}_{\mathsf{id}\|p})=1$ and $(\mathsf{id},p)\notin Q_K$. 
		
		\item[Type\hspace*{1pt}$2$]\hspace*{-3.6pt} $\mathsf{Verify}_{\mathsf{sig}}(\mathsf{mvk},\mathsf{id}\|p,\mathsf{cert}_{\mathsf{id}\|p})=0$ \hspace*{-1pt}or\hspace*{-1pt} $\mathsf{PC}\big((p,m),w\big)=0$ \hspace*{-1pt}or\hspace*{-1pt} $\mathsf{Dec}(\mathsf{mdk},\mathsf{ct})\hspace*{-1pt}\neq \hspace*{-1pt}\mathsf{id}$. 
	\end{description}	
Note that a winning adversary $\mathcal{A}$ is either Type~$1$ or Type~$2$. We now show that a Type~$1$ adversary can be used to break the unforgeability of SIG and a  Type~$2$ adversary can be utilized to breach simulation-sound extractability of $\Pi$. 

Let $\mathcal{A}$ be of Type~$1$. We construct $\mathcal{B}_s$, which utilizes $\mathcal{A}$ as a subroutine,  against EU-CMA of SIG as follows. The experiment is denoted by $\mathbf{Exp}_{\mathrm{SIG},\mathcal{B}_s[\mathcal{A}]}^{\mathrm{EU-CMA}}(1^{\lambda})$.  To begin with, $\mathcal{B}_s$ receives $\mathsf{mvk}$ from its own environment, computes $(\mathsf{crs},\mathsf{tr})\leftarrow \mathsf{SimSetup}_{\mathsf{nizk}}(1^{\lambda})$ and $(\mathsf{mek},\mathsf{mdk})\leftarrow \mathsf{KeyGen}_{\mathsf{pke}}(1^{\lambda})$, sets $Q_K\leftarrow \emptyset$, $Q_S\leftarrow\emptyset$, $ \mathsf{tr}_{\mathsf{tpbs}}\leftarrow (\cdot,\mathsf{mdk},\mathsf{tr})$, and invokes $\mathcal{A}$ by sending $\mathsf{pp}\leftarrow(\mathsf{crs},\mathsf{mek},\mathsf{mvk})$, $\mathsf{mdk}$. A query of {\sc SimSign} on message $m$ is dealt faithfully since $\mathcal{B}_s$ knows the trapdoor  $\mathsf{tr}$ to run $\mathsf{SimProve}$, which produces a signature $\sigma$. $\mathcal{B}_s$ also adds $(m,\sigma)$ to $Q_S$.  All {\sc RevealKey} queries on $(\mathsf{id}, \mathcal{P}_{\mathsf{id}})$ made by $\mathcal{A}$ can be answered by querying $\mathcal{B}_s$'s signing oracle $\mathsf{Sign}_{\mathsf{sig}}(\mathsf{msk},\cdot)$ on  $(\mathsf{id}\|p)$ for all $p\in\mathcal{P}_{\mathsf{id}} $. Meanwhile, $\mathcal{B}_s$ adds $\{ (\mathsf{id},p): p\in \mathcal{P}_{\mathsf{id}}\}$ to $Q_K$. 

When $\mathcal{A}$ outputs $\big(m, (\mathsf{ct},\pi)\big)$ and wins $\mathbf{Exp}_{\mathrm{TPBS},\mathcal{A}}^{\mathrm{EXT}}(1^{\lambda})$, we obtain the valid message signature pair $\big(m, (\mathsf{ct},\pi)\big)$  which is not in $Q_S$. Then $\mathcal{B}_s$ computes $(\mathsf{id},p,\mathsf{cert}_{\mathsf{id}\|p},w,\mathbf{r})\leftarrow\mathsf{Extr}_{\mathsf{nizk}}(\mathsf{tr},m,(\mathsf{ct},\pi))$. If $\mathcal{A}$ is of Type~$1$, then we have  $\mathsf{Verify}_{\mathsf{sig}}(\mathsf{mvk},\mathsf{id}\|p,\mathsf{cert}_{\mathsf{id}\|p})=1$ and $(\mathsf{id},p)\notin Q_K$. Therefore, $\mathcal{B}_s$ wins experiment $\mathbf{Exp}_{\mathrm{SIG},\mathcal{B}_s[\mathcal{A}]}^{\mathrm{EU-CMA}}(1^{\lambda})$ by outputting $(\mathsf{id}\|p, \mathsf{cert}_{\mathsf{id}\|p})$. Hence we have   
\begin{eqnarray}\label{tpbs-eq:EU-CMA-geq-EXT} \mathbf{Adv}_{\mathrm{SIG},\mathcal{B}_s[\mathcal{A}]}^{\mathrm{EU-CMA}}(1^{\lambda}) \geq\mathbf{Adv}_{\mathrm{TPBS},\mathcal{A}}^{\mathrm{EXT}}(\lambda). 
\end{eqnarray}

Let $\mathcal{A}$ be of Type~$2$. We now construct $\mathcal{B}_\pi$ against simulation-sound extractability of proof  $\Pi$. The experiment is denoted as  $\mathbf{Exp}_{\Pi,\mathcal{B}_{\pi}[\mathcal{A}]}^{\mathrm{SE}}(1^{\lambda})$.  First, $\mathcal{B}_{\pi}$ receives $\mathsf{crs}$ from its own environment, computes $(\mathsf{mek},\mathsf{mdk})\leftarrow \mathsf{KeyGen}_{\mathsf{pke}}(1^{\lambda})$ and $(\mathsf{mvk},\mathsf{msk})\leftarrow \mathsf{KeyGen}_{\mathsf{pke}}(1^{\lambda})$, sets $Q_K\leftarrow \emptyset$, $Q_S\leftarrow\emptyset$, $ \mathsf{tr}_{\mathsf{tpbs}}\leftarrow (\mathsf{msk},\mathsf{mdk},\cdot)$, and invokes $\mathcal{A}$ by sending $\mathsf{pp}\leftarrow(\mathsf{crs},\mathsf{mek},\mathsf{mvk})$, $\mathsf{mdk}$. Then $\mathcal{B}_{\pi}$ answers all queries to {\sc ReveaKey} faitfully by employing key $\mathsf{msk}$ and also maintains the list $Q_K$ as in Type~$1$. When $\mathcal{A}$ queries {\sc SimSign} on message $m$,  $\mathcal{B}_{\pi}$ first samples $\mathbf{r}\leftarrow \{0,1\}^{\mathrm{poly}(\lambda)}$, next computes $\mathsf{ct}\leftarrow \mathsf{Enc}(\mathsf{mek},\mathbf{0},\mathbf{r})$, and then queries its own oracle {\sc SimProve} on $(\mathsf{mek},\mathsf{mvk},m,\mathsf{ct})$ and receives back $\pi$, and finally  forwards $(\mathsf{ct},\pi)$ to $\mathcal{A}$ and adds $\big(m, (\mathsf{ct},\pi)\big)$ to $Q_S$.  

When $\mathcal{A}$ outputs $\big(m, (\mathsf{ct},\pi)\big)$ and wins $\mathbf{Exp}_{\mathrm{TPBS},\mathcal{A}}^{\mathrm{EXT}}(1^{\lambda})$, we have that  $\big(m, (\mathsf{ct},\pi)\big)$ is a valid message-signature pair and  is not in list $Q_S$.  More specifically, it implies the algorithm $\mathsf{Verify}_{\mathsf{nizk}}\big(\mathsf{crs}, (\mathsf{mek},\mathsf{mvk},m,\mathsf{ct}),\pi\big)$ outputs~$1$ and $\big(\mathsf{mek},\mathsf{mvk},m,\mathsf{ct}),\pi\big)$ is not in the list for {\sc SimProve} calls maintained by Experiment $\mathbf{Exp}_{\Pi,\mathcal{B}_{\pi}[\mathcal{A}]}^{\mathrm{SE}}(1^{\lambda})$.  Compute $(\mathsf{id},p,\mathsf{cert}_{\mathsf{id}\|p},w,\mathbf{r})\leftarrow\mathsf{Extr}_{\mathsf{nizk}}(\mathsf{tr},(\mathsf{mek},\mathsf{mvk},m,\mathsf{ct}),\pi))$. If $\mathcal{A}$ is of Type~$2$, either  $\mathsf{Verify}_{\mathsf{sig}}(\mathsf{mvk},\mathsf{id}\|p,\mathsf{cert}_{\mathsf{id}\|p})=0$ or $\mathsf{PC}\big((p,m),w\big)=0$ or $\mathsf{Dec}(\mathsf{mdk},\mathsf{ct})\neq \mathsf{id}$, implying $\big((\mathsf{mek},\mathsf{mvk},m,\mathsf{ct}), (\mathsf{id},p,\mathsf{cert}_{\mathsf{id}\|p},w,\mathbf{r})\big )\notin \rho_{\mathsf{tpbs}}.$ Therefore, $\mathcal{B}_{\pi}$ wins experiment $\mathbf{Exp}_{\Pi,\mathcal{B}_{\pi}[\mathcal{A}]}^{\mathrm{SE}}(1^{\lambda})$ by outputting $\big((\mathsf{mek},\mathsf{mvk},m,\mathsf{ct}),\pi\big)$ and we have 
\begin{eqnarray}\label{tpbs-eq:SE-geq-EXT} \mathbf{Adv}_{\Pi,\mathcal{B}_{\pi}[\mathcal{A}]}^{\mathrm{SE}}(1^{\lambda}) \geq\mathbf{Adv}_{\mathrm{TPBS},\mathcal{A}}^{\mathrm{EXT}}(\lambda).\end{eqnarray}

Combining~(\ref{tpbs-eq:EU-CMA-geq-EXT}) and~(\ref{tpbs-eq:SE-geq-EXT}), we obtain
\[\mathbf{Adv}_{\mathrm{SIG},\mathcal{B}_s[\mathcal{A}]}^{\mathrm{EU-CMA}}(1^{\lambda})+\mathbf{Adv}_{\Pi,\mathcal{B}_{\pi}[\mathcal{A}]}^{\mathrm{SE}}(1^{\lambda}) \geq\mathbf{Adv}_{\mathrm{TPBS},\mathcal{A}}^{\mathrm{EXT}}(\lambda).\] By assumption,  $\mathbf{Adv}_{\mathrm{TPBS},\mathcal{A}}^{\mathrm{EXT}}(\lambda) \leq \mathsf{negl}(\lambda)$, which   concludes the proof. 

\end{proof}

\begin{theorem}\label{tpbs-thm: SIM}
Our construction of traceable policy-based signature scheme depicted in Figure~\ref{tpbs-fig: generic-construction} is simulatable (where simulatability is modeled in experiment $\mathbf{Exp}_{\mathrm{TPBS},\mathcal{A}}^{\mathrm{SIM}}$ in Figure~\ref{fbgs-fig:sim-ext}), if the underlying signature scheme SIG  is  EU-CMA secure,  the encryption scheme $\mathrm{PKE}$ is  IND-CCA secure,  and the proof $\Pi$ is zero-knowledge and simulation-sound extractable. 
\end{theorem}
\begin{proof}
	We show that two runs of $\mathbf{Exp}_{\mathrm{TPBS},\mathcal{A}}^{\mathrm{SIM}}(1^{\lambda})$, one with $b$ set to~$1$, and one with $b$ set to~$0$, are indistinguishable. We proceed by a sequence of indistinguishable games, in which the first game is   $\mathbf{Exp}_{\mathrm{TPBS},\mathcal{A}}^{\mathrm{SIM}-1}(1^{\lambda})$ and the last one  $\mathbf{Exp}_{\mathrm{TPBS},\mathcal{A}}^{\mathrm{SIM}-0}(1^{\lambda})$.  
 
 Let $E_i$ be the event that adversary $\mathcal{A}$ outputs~$1$ in Game~$i$.  
 \begin{description}
 	\item[Game~$1$] This game is  $\mathbf{Exp}_{\mathrm{TPBS},\mathcal{A}}^{\mathrm{SIM}-1}(1^{\lambda})$. Therefore,  $\mathrm{Pr}[\mathbf{Exp}_{\mathrm{TPBS},\mathcal{A}}^{\mathrm{SIM}-1}(1^{\lambda})=1]=\mathrm{Pr}[E_1]$.  
 	\item[Game~$2$] This game modifies Game~$1$ by replacing algorithms  $\mathsf{Setup}_{\mathsf{nizk}}$ and $\mathsf{Prove}$ to algorithms $\mathsf{SimSetup}_{\mathsf{nizk}}$ and $\mathsf{SimProve}$, respectively. Then by the zero-knowledge property of the proof system $\Pi$, Game~$2$ is indistinguishable from Game~$1$.  In other words, we have $|\mathrm{Pr}[E_1]-\mathrm{Pr}[E_2]|\leq \mathbf{Adv}_{\Pi}^{\mathrm{ZK}}(1^{\lambda})$. 
 	
 	\item[Game~$3$] This game is the same as Game~$2$ except that when $\mathcal{A}$ queries {\sc Signature} on inputs $(i,m,w)$, $\mathsf{ct}$ is an encryption of a dummy identity $\mathbf{0}$ instead of $Q[i][1]$. Note that this game is exactly $\mathbf{Exp}_{\mathrm{TPBS},\mathcal{A}}^{\mathrm{SIM}-0}(1^{\lambda})$. Hence $\mathrm{Pr}[E_3]=\mathrm{Pr}[\mathbf{Exp}_{\mathrm{TPBS},\mathcal{A}}^{\mathrm{SIM}-0}(1^{\lambda})=1]$. We show that Game~$2$ and Game~$3$ are indistinguishable to $\mathcal{A}$ by proving the following
 	\begin{eqnarray}\label{tpbs-eq:G2-G3-difference}
 \hspace*{-9pt}	|\mathrm{Pr}[E_2]-\mathrm{Pr}[E_3] | \leq \mathbf{Adv}_{\mathrm{PKE}}^{\mathrm{IND-CCA}}(1^{\lambda})+ \mathbf{Adv}_{\mathrm{TPBS}}^{\mathrm{EXT}}(1^{\lambda}). 
 	\end{eqnarray}
 \end{description} 
Combing the above equations, we then have 
\begin{eqnarray*}
\mid\mathrm{Pr}[E_1]-\mathrm{Pr}[E_3]\mid \leq \mathbf{Adv}_{\Pi}^{\mathrm{ZK}}(1^{\lambda})+\mathbf{Adv}_{\mathrm{PKE}}^{\mathrm{IND-CCA}}(1^{\lambda})+\mathbf{Adv}_{\mathrm{TPBS}}^{\mathrm{EXT}}(1^{\lambda}). 
\end{eqnarray*}
In Theorem~\ref{tpbs-thm: EXT}, we showed that $$\mathbf{Adv}_{\mathrm{TPBS}}^{\mathrm{EXT}}(1^{\lambda})\leq \mathbf{Adv}_{\mathrm{SIG}}^{\mathrm{EU-CMA}}(1^{\lambda})+\mathbf{Adv}_{\Pi}^{\mathrm{SE}}(1^{\lambda}).$$ 
By assumption, our scheme is then simulatable. We now prove equation~(\ref{tpbs-eq:G2-G3-difference}) holds. 

\begin{figure}[ht]
	\begin{center}
		\begin{tabular}{cc}
			\begin{minipage}{7.8cm}
				\underline{$\mathcal{B}(\mathsf{mek}: \text {\sc LR}(\cdot,\cdot),\text{\sc Dec}(\cdot))$}\vspace*{2.2pt}\\
				$\hspace*{6pt}Q_C\leftarrow \emptyset, j\leftarrow 0$\\
				$\hspace*{6pt}(\mathsf{crs},\mathsf{tr}_{})\leftarrow\mathsf{SimSetup}_{\mathsf{nizk}} (1^\lambda)$  \\
				$\hspace*{6pt} (\mathsf{mvk},\mathsf{msk})\leftarrow\mathsf{KeyGen}_{\mathsf{sig}}(1^{\lambda})$\\
				$\hspace*{6pt}\text{Set}\hspace*{2pt}\mathsf{pp}\hspace*{-2pt}\leftarrow\hspace*{-2pt}(\mathsf{crs},\mathsf{mek},\mathsf{mvk})$,\\ $\hspace*{6pt}\mathsf{tr}_{\mathsf{tpbs}}\leftarrow (\mathsf{msk},\cdot,\mathsf{tr})$\\
				$\hspace*{6pt}b'\leftarrow \mathcal{A}\big(\mathsf{pp}: \text{\sc Key}_{\mathcal{B}}(\cdot,\cdot),\\
				\hspace*{12pt}\text{\sc Signature}_{\mathcal{B}}(\cdot,\cdot,\cdot), \text{\sc Open}_{\mathcal{B}}(\cdot,\cdot)\big)$\\
				$\hspace*{6pt}\text{Return}~b'$
				
				\underline{\sc Key$_{\mathcal{B}}(\mathsf{id},\mathcal{P}_{\mathsf{id}})$}\vspace*{2.2pt}\\
				$\hspace*{6pt}j\leftarrow j+1$\\
				$\hspace*{6pt}\mathsf{usk}_{\mathsf{id}}\leftarrow\mathsf{SimKeyGen}(\mathsf{tr}_{\mathsf{tpbs}},\mathsf{id},\mathcal{P}_{\mathsf{id}})$\\
				$\hspace*{6pt}Q[j][1]=\mathsf{id},Q[j][2]=\{p:p\in \mathcal{P}_{\mathsf{id}}\}$\\
				$\hspace*{6pt}Q[j][3]=\mathsf{usk}_{\mathsf{id}}$\\
				$\hspace*{6pt}\text{Return}~\mathsf{usk}_{\mathsf{id}}$
				
				\underline{\sc Signature$_{\mathcal{B}}(i,m,w)$}\vspace*{2.2pt}	\\	
				$\hspace*{6pt}\text{If}~ i\notin [j], \text{then~return}~\bot$\\
				$\hspace*{6pt}\text{If}~\exists~p\in Q[i][2]~\text{s.t.}~\mathsf{PC}\big((p,m),w\big)=1 $ \\
				$\hspace*{10pt}$ \mybox{$\mathsf{ct}^*\leftarrow \text{\sc LR}(\mathbf{0},Q[i][1])$}\\
				$\hspace*{12pt}\pi^* \leftarrow\mathsf{SimProve}\big(\mathsf{crs},\mathsf{tr},(\mathsf{mek},\mathsf{mvk},m,\mathsf{ct}^*)\big)$\\
				$\hspace*{12pt}Q_C=Q_C\cup \{\big(m,(\mathsf{ct}^*,\pi^*)\big)\}$\\
				$\hspace*{12pt}\text{Return}~\sigma^*\leftarrow (\mathsf{ct}^*,\pi^*)$\\
				$\hspace*{6pt}\text{Else} ~\text{return}~\bot$ 
				
				\underline{\sc Open$_{\mathcal{B}}\big(m,(\mathsf{ct},\pi)\big)$}\vspace*{2.2pt}\\
				$\hspace*{6pt}\text{If}~\big(m,(\mathsf{ct},\pi)\big)\in Q_C,~\text{return}~\bot$\\
				$\hspace*{6pt}\text{If}~\mathsf{Verify}(\mathsf{pp},m,(\mathsf{ct},\pi))=0,~\text{return}~\bot$\\
				$\hspace*{6pt}\text{If}~(\mathsf{ct}=\mathsf{ct}^*)\wedge (\pi\neq \pi^*)$\\
				$\hspace*{12pt}\mathcal{B}~\text{halts~and~return}~1$\\
				$\hspace*{6pt}\text{Else~return}$~\mybox{$\text{\sc Dec}(\mathsf{ct})$}
			\end{minipage}	
			&
			\begin{minipage}{8.6cm}
				\vspace*{0.43cm}
				\underline{$\mathcal{B}_E(\mathsf{pp},\mathsf{mdk}: \text{\sc RevealKey}(\cdot,\cdot), \text{\sc SimSign}(\cdot))$}\vspace*{2.2pt}\\
				$\hspace*{6pt}Q_C\leftarrow \emptyset, j\leftarrow 0$\\
				$\hspace*{6pt}b'\leftarrow \mathcal{A}\big(\mathsf{pp}: \text{\sc Key}_{\mathcal{B}_E}(\cdot,\cdot),\\
				\hspace*{12pt}\text{\sc Signature}_{\mathcal{B}_E}(\cdot,\cdot,\cdot), \text{\sc Open}_{\mathcal{B}_E}(\cdot,\cdot)\big)$\\
				\vspace*{0.53cm}
				
				\underline{\sc Key$_{\mathcal{B}_E}(\mathsf{id},\mathcal{P}_{\mathsf{id}})$}\vspace*{2.2pt}\\
				$\hspace*{6pt}j\leftarrow j+1$\\
				$\hspace*{6pt}$\mybox{$\mathsf{usk}_{\mathsf{id}}\leftarrow\text{\sc{RevealKey}}(\mathsf{id},\mathcal{P}_{\mathsf{id}})$}\\
				$\hspace*{6pt}Q[j][1]=\mathsf{id},Q[j][2]=\{p:p\in \mathcal{P}_{\mathsf{id}}\}$\\
				$\hspace*{6pt}Q[j][3]=\mathsf{usk}_{\mathsf{id}}$\\
				$\hspace*{6pt}\text{Return}~\mathsf{usk}_{\mathsf{id}}$
				
				\underline{\sc Signature$_{\mathcal{B}_E}(i,m,w)$}\vspace*{2.2pt}	\\	
				$\hspace*{6pt}\text{If}~ i\notin [j], \text{then~return}~\bot$\\
				$\hspace*{6pt}\text{If}~\exists~p\in Q[i][2]~\text{s.t.}~\mathsf{PC}\big((p,m),w\big)=1 $ \\
				$\hspace*{12pt}$\mybox{$(\mathsf{ct}^*,\pi^*) \leftarrow\text{\sc SimSign}( m)$}\\
				$\hspace*{12pt}Q_C=Q_C\cup \{\big(m,(\mathsf{ct}^*,\pi^*)\big)\}$\\
				$\hspace*{12pt}\text{Return}~\sigma^*\leftarrow (\mathsf{ct}^*,\pi^*)$\\
				$\hspace*{6pt}\text{Else} ~\text{return}~\bot$ \\
				\vspace*{0.24cm}
				
				\underline{\sc Open$_{\mathcal{B}_E}\big(m,(\mathsf{ct},\pi)\big)$}\vspace*{2.2pt}\\
				$\hspace*{6pt}\text{If}~\big(m,(\mathsf{ct},\pi)\big)\in Q_C,~\text{return}~\bot$\\
				$\hspace*{6pt}\text{If}~\mathsf{Verify}(\mathsf{pp},m,(\mathsf{ct},\pi))=0,~\text{return}~\bot$\\
				$\hspace*{6pt}\text{If}~(\mathsf{ct}=\mathsf{ct}^*)\wedge (\pi\neq \pi^*)$ \\
				$\hspace*{12pt}\mathcal{B}_{E}~\text{halts~and~return}~\big(m,(\mathsf{ct}^*,\pi)\big)$\\
				$\hspace*{6pt}\text{Else~return}~\mathsf{Dec}(\mathsf{mdk},\mathsf{ct})$
				
			\end{minipage}
		\end{tabular}
	\end{center}
	\caption{Adversary $\mathcal{B}$ against IND-CCA of $\mathrm{PKE}$ and adversary $\mathcal{B}_E$ against EXT of our TPBS. }	
	\label{tpbs-fig:Bb-against-cca+Be-against-EXT-TPBS}
\end{figure}	

Let $\mathcal{B}$, which utilizes $\mathcal{A}$ that distinguishes Game~$2$ and Game~$3$, be an adversary against IND-CCA of the underlying encryption scheme $\mathrm{PKE}$.  Denote   $\mathbf{Exp}_{\mathrm{PKE},\mathcal{B}[\mathcal{A}]}^{\mathrm{IND-CCA}}(1^{\lambda})$ to be the experiment. We model $\mathcal{B}$ in Figure~\ref{tpbs-fig:Bb-against-cca+Be-against-EXT-TPBS}, in which it receives $\mathsf{mek}$ from its environment, queries two messages to the challenge oracle LR$(\cdot,\cdot)$ and receives back a challenge ciphertext. Meanwhile, $\mathcal{B}$ is also given access to a decryption oracle $\text{\sc Dec}(\cdot)$ where  it is allowed to query any ciphertext except the one obtained from the challenge oracle.  Before we  analyze the behavior of $\mathcal{B}$ in  $\mathbf{Exp}_{\mathrm{PKE},\mathcal{B}[\mathcal{A}]}^{\mathrm{IND-CCA}}(1^{\lambda})$,  let us define event $Q_i$ for $i\in\{2,3\}$ first: $Q_i$ is the event that $\mathcal{A}$ makes a valid  {\sc Open} query $\big(m, (\mathsf{ct},\pi)\big)$ such that $(\mathsf{ct}=\mathsf{ct}^*)\wedge (\pi\neq \pi^*)$ in Game~$i$. 
Let us now consider the experiment $\mathbf{Exp}_{\mathrm{PKE},\mathcal{B}[\mathcal{A}]}^{\mathrm{IND-CCA}-1}(1^{\lambda})$, where oracle $\text{\sc LR}(\cdot, \cdot)$ returns ciphertext of the second input. Note that $\mathcal{B}$ in this experiment perfectly simulates the view of $\mathcal{A}$ in Game~$2$. Therefore 
\begin{eqnarray*}
\mathrm{Pr}[\mathbf{Exp}_{\mathrm{PKE},\mathcal{B}[\mathcal{A}]}^{\mathrm{IND-CCA}-1}(1^{\lambda})=1]&=&\mathrm{Pr}[E_2\wedge \neg Q_2]+\mathrm{Pr}[Q_2]\geq \mathrm{Pr}[E_2]. 
\end{eqnarray*}
On the other hand, $\mathcal{B}$ in the experiment $\mathbf{Exp}_{\mathrm{PKE},\mathcal{B}[\mathcal{A}]}^{\mathrm{IND-CCA}-0}(1^{\lambda})$, where oracle $\text{\sc LR}(\cdot, \cdot)$ returns ciphertext of the first input, perfectly simulates the view of $\mathcal{A}$ in Game~$3$. Therefore 
\begin{eqnarray*}
\mathrm{Pr}[\mathbf{Exp}_{\mathrm{PKE},\mathcal{B}[\mathcal{A}]}^{\mathrm{IND-CCA}-0}(1^{\lambda})=1]&=&\mathrm{Pr}[E_3\wedge \neg Q_3]+\mathrm{Pr}[Q_3]\leq  \mathrm{Pr}[E_3]+\mathrm{Pr}[Q_3]. 
\end{eqnarray*}
Together this yields 
\begin{eqnarray}\label{tpbs-eq:game-2-game-3-one-side-relation}
\mathrm{Pr}[E_2]-\mathrm{Pr}[E_3]
&\leq & \mathbf{Adv}_{\mathrm{PKE},\mathcal{B}[\mathcal{A}]}^{\mathrm{IND-CCA}}(1^{\lambda})+\mathrm{Pr}[Q_3]. 
\end{eqnarray}
Now to lower bound the terms on the left-hand side of equation~(\ref{tpbs-eq:game-2-game-3-one-side-relation}), we define adversary $\bar{\mathcal{B}}$ that behaves as $\mathcal{B}$ with two modifications: (1) it queries $(Q[i][1],\mathbf{0})$ to its challenge oracle $\text{\sc LR}(\cdot,
\cdot)$; (2) when $\mathcal{A}$ makes a valid query $\big(m, (\mathsf{ct},\pi)\big)$ such that $(\mathsf{ct}=\mathsf{ct}^*)\wedge (\pi\neq \pi^*)$, $\bar{\mathcal{B}}$ halts and returns~$0$. Following the above analysis, we get 
\begin{eqnarray*}
	\mathrm{Pr}[\mathbf{Exp}_{\mathrm{PKE},\bar{\mathcal{B}}[\mathcal{A}]}^{\mathrm{IND-CCA}-0}(1^{\lambda})=1]&=&\mathrm{Pr}[E_2\wedge \neg Q_2]\leq \mathrm{Pr}[E_2];\\
	\mathrm{Pr}[\mathbf{Exp}_{\mathrm{PKE},\bar{\mathcal{B}}[\mathcal{A}]}^{\mathrm{IND-CCA}-1}(1^{\lambda})=1]&=&\mathrm{Pr}[E_3\wedge \neg Q_3] \geq  \mathrm{Pr}[E_3]-\mathrm{Pr}[{Q}_3]. 
\end{eqnarray*}

Together this yields 
\begin{eqnarray}\label{tpbs-eq:game-2-game-3-another-side-relation}
\mathrm{Pr}[E_3]-\mathrm{Pr}[E_2]&\leq & \mathbf{Adv}_{\mathrm{PKE},\bar{\mathcal{B}}[\mathcal{A}]}^{\mathrm{IND-CCA}}(1^{\lambda})+\mathrm{Pr}[Q_3]. 
\end{eqnarray}
Combining~(\ref{tpbs-eq:game-2-game-3-one-side-relation}) and (\ref{tpbs-eq:game-2-game-3-another-side-relation}), we have 
\begin{eqnarray}\label{tpbs-eq:game-2-game-3-both-sides-relation}
|\mathrm{Pr}[E_2]-\mathrm{Pr}[E_3]|&\leq & \mathbf{Adv}_{\mathrm{PKE}}^{\mathrm{IND-CCA}}(1^{\lambda})+\mathrm{Pr}[Q_3]. 
\end{eqnarray}
To show (\ref{tpbs-eq:G2-G3-difference}) holds, we are left to show 
\begin{eqnarray}\label{tpbs-eq:Q3-upper-bound-by-ext-of-tpbs}
\mathrm{Pr}[Q_3]\leq \mathbf{Adv}_{\mathrm{TPBS}}^{\mathrm{EXT}}(1^{\lambda}). 
\end{eqnarray}

Towards this goal, we construct $\mathcal{B}_E$ that  breaks EXT of our construction whenever event $Q_3$ occurs. Denote the experiment as  $\mathbf{Exp}_{\mathrm{TPBS},\mathcal{B}_E[\mathcal{A}]}^{\mathrm{EXT}}(1^{\lambda})$. We model $\mathcal{B}_E$ in Figure~\ref{tpbs-fig:Bb-against-cca+Be-against-EXT-TPBS}, in which it receives $\mathsf{pp}$, $\mathsf{mdk}$ from its own environment, and is given access to oracles $\text{\sc RevealKey}(\cdot,\cdot)$ and $\text{\sc SimSign}(\cdot)$. Note that $\mathcal{B}_E$ defined in this way perfectly simulates the view of $\mathcal{A}$ in Game~$3$. 

We claim that $\mathcal{B}_E$ wins  $\mathbf{Exp}_{\mathrm{TPBS},\mathcal{B}_E[\mathcal{A}]}^{\mathrm{EXT}}(1^{\lambda})$ by returning $\big(m,(\mathsf{ct}^*,\pi)\big)$ when $Q_3$ occurs.  In fact, $Q_3$ implies that (1) $\mathsf{Verify}(\mathsf{pp},m,(\mathsf{ct}^*,\pi))=1$; (2)$\big(m, (\mathsf{ct}^*,\pi)\big)$ is not in the list for {\sc SimSign} calls. Compute $$(\mathsf{id},p,\mathsf{cert}_{\mathsf{id}\|p},w,\mathbf{r}) \leftarrow \mathsf{Extr}_{\mathsf{nizk}}(\mathsf{tr}, (\mathsf{mek},\mathsf{mvk},m,\mathsf{ct}^*),\pi).$$
Either  $(\mathsf{id},p)$ was not in the list for {\sc RevealKey} calls or $(\mathsf{id},p)$ was indeed in the list for {\sc RevealKey} calls. The former case immediately implies that $\mathcal{B}_E$ wins $\mathbf{Exp}_{\mathrm{TPBS},\mathcal{B}_E[\mathcal{A}]}^{\mathrm{EXT}}(1^{\lambda})$. In case the latter occurs, there exists $i$ such that $Q[i][1]=\mathsf{id}$. Specifically, it implies $\mathsf{id}\neq \mathbf{0}$. Recall that $\mathsf{ct}^*$ is indeed encryption of~$\mathbf{0}$. Therefore, correctness of the underlying encryption scheme PKE implies $\mathsf{Dec}(\mathsf{mdk},\mathsf{ct}^*)\neq \mathsf{id}$, indicating $\mathcal{B}_E$ wins the experiment as well. To summarize, we get 
$\mathbf{Adv}_{\mathrm{TPBS},\mathcal{B}_E[\mathcal{A}]}^{\mathrm{EXT}}(1^{\lambda})\geq \mathrm{Pr}[Q_3]$. This concludes the proof.

\end{proof}

\section{Lattice-Based Instantiation of TPBS}\label{tpbs-section:tpbs-lattice-con}

This section presents a realization of our generic construction  of TPBS under concrete lattice-based assumptions. Let us first briefly review several lattice-based techniques that will be used in the construction. 

\subsection{Preliminaries on Lattices}\label{tpbs-subsection:preliminary-on-lattices}


We first recall $q$-ary lattices and then review the Gaussian distribution over these lattices. 
Let $q\geq 2$ and $n,m$ be positive integers.  For $\mathbf{A}\in \mathbb{Z}_q^{n\times m}$,  we define the full-rank $q$-ary lattice with dimension $m$ as follows:
$$\Lambda^{\bot}(\mathbf{A})=\{\mathbf{x}\in\mathbb{Z}^m: \mathbf{A}\cdot\mathbf{x}=\mathbf{0} \bmod q\}.$$

For any non-zero vector $\mathbf{u}\in\mathbb{Z}_q^n$ that admits an integral solution to the equation $\mathbf{A}\cdot \mathbf{x}=\mathbf{u}\bmod q$, define the coset $$\Lambda^{\mathbf{u}}(\mathbf{A})=\{\mathbf{x}\in\mathbb{Z}^m: \mathbf{A}\cdot \mathbf{x}=\mathbf{u} \bmod q\}.$$

For any vector $\mathbf{c}\in \mathbb{R}^n$ and any positive real number $s$, define the following: $$\rho_{s,\mathbf{c}}(\mathbf{x})=\exp\left(-\pi\frac{\|\mathbf{x}-\mathbf{c}\|^2}{s^2}\right) ~~~\text{and} ~~~ \rho_{s,\mathbf{c}}(\Lambda)=\sum_{\mathbf{x}\in\Lambda}\rho_{s,\mathbf{c}}(\mathbf{x}).$$ We often omit $\mathbf{c}$ if it is $\mathbf{0}$. Define the distribution over the coset $\Lambda^{\mathbf{u}}(\mathbf{A})$ as $D_{\Lambda^{\mathbf{u}}(\mathbf{A}),s,\mathbf{c}}(\mathbf{x})={\rho_{s,\mathbf{c}}(\mathbf{x})}/{\rho_{s,\mathbf{c}}(\Lambda^{\mathbf{u}}(\mathbf{A}))}$ for any $\mathbf{x}\in\Lambda^{\mathbf{u}}(\mathbf{A})$.

Now let us look at the Gaussian distributions. For any vector $\mathbf{c}\in \mathbb{R}^n$ and any positive real number $s$, define the following: $$\rho_{s,\mathbf{c}}(\mathbf{x})=\exp\left(-\pi\frac{\|\mathbf{x}-\mathbf{c}\|^2}{s^2}\right) ~~~\text{and} ~~~ \rho_{s,\mathbf{c}}(\Lambda)=\sum_{\mathbf{x}\in\Lambda}\rho_{s,\mathbf{c}}(\mathbf{x}).$$ Then the discrete Gaussian distribution  over the lattice $\Lambda$ with parameter $s$ and center $\mathbf{c}$, denoted as $D_{\Lambda,s,\mathbf{c}}$, is defined to be $D_{\Lambda,s,\mathbf{c}}(\mathbf{x})={\rho_{s,\mathbf{c}}(\mathbf{x})}/{\rho_{s,\mathbf{c}}(\Lambda)}$ for any $\mathbf{x}\in \Lambda$. We often omit $\mathbf{c}$ if it is $\mathbf{0}$.

Note that the coset $\Lambda^{\mathbf{u}}(\mathbf{A})$ is not a lattice for any non-zero $\mathbf{u}\in\mathbb{Z}_q^n$, since obviously~$\mathbf{0}$ is not inside this set. However, we can still define the discrete Gaussian distribution over  the coset $\Lambda^{\mathbf{u}}(\mathbf{A})$ in a similar way:  $D_{\Lambda^{\mathbf{u}}(\mathbf{A}),s,\mathbf{c}}(\mathbf{x})={\rho_{s,\mathbf{c}}(\mathbf{x})}/{\rho_{s,\mathbf{c}}(\Lambda^{\mathbf{u}}(\mathbf{A}))}$ for any $\mathbf{x}\in\Lambda^{\mathbf{u}}(\mathbf{A})$.

\noindent Now we are going to recall some well-known facts about the discrete Gaussian distributions.

\begin{lemma}[\cite{GPV08STOC,PR06TCC}]\label{tpbs-lemma:infinity-norm-bound}
	Let $n,q,m$ be some positive integers such that $q\geq 2$ and $m\geq 2n\log q$. Define a positive real number $s$ such that $s\geq \omega(\sqrt{\log m})$.
	\begin{itemize}
		\item Then for all but a $2q^{-n}$ fraction of all matrices $\mathbf{A}$ over $\mathbb{Z}_q^{n\times m}$, the distribution of the syndrome $\mathbf{u}=\mathbf{A}\cdot\mathbf{x}\mod q$ for $\mathbf{x}\hookleftarrow D_{\mathbb{Z}^m,s}$ is statistically close to uniform over $\mathbb{Z}_q^n$. Besides, given $\mathbf{A}\cdot\mathbf{x}=\mathbf{u}\mod q$,  the conditional distribution of $\mathbf{x}\hookleftarrow D_{\mathbb{Z}^m,s}$ is $D_{\Lambda^{\mathbf{u}}(\mathbf{A}),s}$.
		\item Let ${a}\hookleftarrow D_{\mathbb{Z},s}$, $t=\log n$, and $\beta=\lceil s\cdot t\rceil$. Then the probability of $ |{a}|>\beta$ is negligible.
		\item The min-entropy of the distribution $D_{\mathbb{Z}^m,s}$ is at least $m-1$. In other words, for any $\mathbf{x}\in D_{\mathbb{Z}^m,s}$, we have $D_{\mathbb{Z}^m,s}(\mathbf{x})\leq 2^{1-m}$.
	\end{itemize}
\end{lemma}

Let $\mathbf{S}=[\mathbf{s}_1|\cdots|\mathbf{s}_m]\in\mathbb{Z}_q^{m\times m}$  be a full-rank matrix.   The Gram-Schmidt orthogonalization of these $m$ vectors is a sequence of $m$ new orthogonal vectors $\tilde{\mathbf{s}}_1,\tilde{\mathbf{s}}_1,\ldots,\tilde{\mathbf{s}}_m$ computed as $\tilde{\mathbf{s}}_1=\mathbf{s}_1$;   $\tilde{\mathbf{s}}_j=\mathbf{s}_j-\sum_{k=1}^{j-1}a_{j,k}\cdot \tilde{\mathbf{s}}_k$,  where we have  $a_{j,k}=(\mathbf{s}_j^\top\cdot\tilde{\mathbf{s}}_k)/(\tilde{\mathbf{s}}_k^\top\cdot\tilde{\mathbf{s}}_k)$ for $j=2,\ldots, m$.  
Denote $\widetilde{\mathbf{S}}=[\tilde{\mathbf{s}}_1|\cdots|\tilde{\mathbf{s}}_m]$ be its Gram-Schmidt matrix.  
Define $\|\mathbf{S}\|=\mathrm{max}_{i\in[m]}\|\mathbf{s}_i\|$ and $\|\widetilde{\mathbf{S}}\|=\mathrm{max}_{i\in[m]}\|\tilde{\mathbf{s}}_i\|$, where $\|\cdot\|$ is the Euclidean norm. We now recall some algorithms from previous works that will be used  in this work. The $\mathsf{TrapGen}$ algorithm is used to generate a matrix $\mathbf{A}$  that is statistically close to random together with a good basis of the $q$-ary lattice $\Lambda^{\bot}(\mathbf{A})$. The  $\mathsf{SampleD}$ algorithm employs some good basis of the lattice $\Lambda^{\bot}(\mathbf{A})$ to output a short vector in $\Lambda^{\mathbf{u}}(\mathbf{A})$ if $\Lambda^{\mathbf{u}}(\mathbf{A})$ is  not empty. The $\mathsf{ExtBasis}$ algorithm extends a basis of a matrix $\mathbf{A}$ to a basis of any matrix that ontains $\mathbf{A}$ as a submatrix. 
\begin{lemma}[$\mathsf{TrapGen}$~\cite{AP09STACS,MP12EC}]\label{tpbs-lemma:trapgen}
	Let $n,m,q$  be positive integers such that $q\geq 2$ and $m$ is of order $\mathcal{O}(n\log q)$. Then the $\mathrm{PPT}$ algorithm $\mathsf{TrapGen}(n,m,q)$ outputs a tuple $(\mathbf{A},\mathbf{S})$ satisfying the following conditions:
	\begin{itemize}
		\item $\mathbf{A}$ is   within negligibly statistical distance from the uniformly random distribution over $\mathbb{Z}_q^{n\times m}$,
		\item $\mathbf{S}$ is a basis for $\Lambda^{\bot}(\mathbf{A})$, that is, $\mathbf{A}\cdot\mathbf{S}=\mathbf{0}\bmod q$, and
		\item $\|\widetilde{\mathbf{S}}\|\leq \mathcal{O}(\sqrt{n\log q})$.
	\end{itemize}
\end{lemma}

\begin{lemma}[$\mathsf{SampleD}$~\cite{GPV08STOC}]\label{tpbs-lemma:sampled}
	Given   a  basis $\mathbf{S}\in\mathbb{Z}^{m\times m}$ of the full-rank $q$-ary lattice $\Lambda^{\bot}(\mathbf{A})$ for a matrix $\mathbf{A}\in\mathbb{Z}_q^{n\times m}$, a vector $\mathbf{u}$ over $\mathbb{Z}_q^n$, and a positive real number $s\geq\|\widetilde{\mathbf{S}}\|\cdot \omega(\sqrt{\log n})$,  the $\mathrm{PPT}$ algorithm $\mathsf{SampleD}(\mathbf{A},\mathbf{S},\mathbf{u},s)$ outputs a vector $\mathbf{x}\in \Lambda^{\mathbf{u}}(\mathbf{A})$   that is statistically close to the distribution $D_{\Lambda^{\mathbf{u}}(\mathbf{A}),s}$.
\end{lemma}

\begin{lemma}[$\mathsf{ExtBasis}$~\cite{CHKP10EC}]\label{tpbs-lemma:extbasis}
	Given a basis $\mathbf{S}\in\mathbb{Z}^{m\times m}$ of the full-rank $q$-ary lattice $\Lambda^{\bot}(\mathbf{A})$ for some $\mathbf{A}\in\mathbb{Z}_q^{n\times m}$, and a matrix $\mathbf{A}'\in\mathbb{Z}_q^{n\times {m'}}$  containing $\mathbf{A}$ as a submatrix, the $\mathrm{PPT}$ algorithm $\mathsf{ExtBasis}(\mathbf{S},\mathbf{A}')$ outputs a basis $\mathbf{S}'\in\mathbb{Z}^{m'\times m'}$ of the $q$-ary lattice $\Lambda^{\bot}(\mathbf{A}')$ such that  $\|\widetilde{\mathbf{S}'}\|=\|\widetilde{\mathbf{S}}\|$.
\end{lemma}

We then review two lattice problems: short integer solution $(\mathsf{SIS})$ problem  and learning with errors ($\mathsf{LWE}$) problem, together with their hardness results. . 

\begin{definition}[$\mathsf{SIS}^{\infty}_{n,m,q,\beta}$~\cite{Ajtai96STOC,GPV08STOC}]
	Given a uniformly random input matrix $\mathbf{A}$ over $\mathbb{Z}_q^{n\times m}$, the $\mathsf{SIS}^{\infty}_{n,m,q,\beta}$ problem asks to output a vector $\mathbf{x}\in\mathbb{Z}_q^{m}$ such that  $\mathbf{A}\cdot \mathbf{x}=\mathbf{0} \bmod q$ and $0<\|\mathbf{x}\|_{\infty}\leq \beta$.
	
\end{definition}

Let $q > \beta\cdot\widetilde{\mathcal{O}}(\sqrt{n})$ be an integer and $m,\beta$ be polynomials in $n$. Then solving  $\mathsf{SIS}^{\infty}_{n,m,q,\beta}$ problem  is at least as hard as solving $\mathsf{SIVP}_{\gamma}$ problem in the worst case for some approximation factor $\gamma = \beta \cdot \widetilde{\mathcal{O}}(\sqrt{nm})$ (\cite{MR04FOCS,GPV08STOC,MP13C}).
\begin{definition}[$\mathsf{LWE}_{n,q,\chi}$~\cite{Regev05STOC}]
	For  positive integers $n,m,q\geq 2$ and a probability distribution $\chi$ over integers $\mathbb{Z}$, define a distribution  $\mathcal{A}_{\mathbf{s}, \chi}$ over $\mathbb{Z}_q^n \times \mathbb{Z}_q$ for $\mathbf{s}\xleftarrow{\$} \mathbb{Z}_q^n$ as follows: it samples a uniformly random vector $\mathbf{a}$ over $\mathbb{Z}_q^{n}$ and an error element $e$ according to $\chi$, and outputs $(\mathbf{a},\mathbf{a}^\top\cdot \mathbf{s}+e)$. Then the goal of the $\mathsf{LWE}_{n,q,\chi}$ problem is to distinguish $m$ samples chosen from a uniform distribution over $\mathbb{Z}_q^n \times \mathbb{Z}_q$ from $m$ samples chosen from the distribution $\mathcal{A}_{\mathbf{s}, \chi}$ for  some $\mathbf{s}\xleftarrow{\$} \mathbb{Z}_q^n$.
	
\end{definition}
If $q\geq 2$ is an arbitrary modulus, then $\mathsf{LWE}_{n,q,\chi}$ problem is at least as hard as the worst-case problem $\mathsf{SIVP}_{\gamma}$ with $\gamma=\widetilde{\mathcal{O}}(n\cdot q/B)$ 
through an efficient quantum reduction (see, e.g.~\cite{Regev05STOC,PRS17STOC}). Additionally, it is showed that the hardness of the $\mathsf{LWE}$ problem is  maintained when the secret $\mathbf{s}$ is chosen from the error distribution $\chi$ (see~\cite{ACPS09C}). 
 
\subsection{Stern-Like Protocols}\label{tpbs-subsection:Stern-like-protocols}
We will work with statistical zero-knowledge argument of knowledge (ZKAoK) that operates in Stern's framework~\cite{Stern96IT}.   Stern's protocol was originally proposed in the code-based cryptography and was later adapted to the lattice setting~\cite{KTX08AC,LNSW13PKC,LLMNW16AC-dgs,LLNW17AC-ecash} to handle various matrix-vector relations associated with $\mathsf{SIS}$ and $\mathsf{LWE}$ problems. The protocol consists of three moves: commitment, challenge, and response. If a statistically hiding and computationally binding commitment scheme is used in the first step, then one obtains a statistical ZKAoK with perfect completeness and soundness error $2/3$. To achieve negligible soundness error, one can repeat the protocol $\omega(\log \lambda)$ times for security parameter $\lambda$. For our purpose of using the protocol to sign a message, we further apply Fiat-Shamir  transform~\cite{FS86C} to obtain a non-interactive proof. It was shown that the resulting NIZKAoK protocol is simulation-sound extractable~\cite{FKMV12IndoC} in the random oracle model. 
In this paper, we work with a simplified abstracted protocol  from~\cite{LLNW17AC-ecash} that handles two moduli. We now recall this protocol. 

 Let $q_i,,K_i, L_i$ be positive integers such that $q_i \geq 2$, $L_i\geq K_i$,  and let $L=L_1+L_2$. Let $\mathsf{VALID}\subseteq \{-1,0,1\}^L$ and a finite set $\mathcal{S}$, associate every $\eta\in \mathcal{S}$ with a permutation $\Gamma_\eta$ of $L$ elements such that  the following conditions hold:
\begin{eqnarray}\label{eq:zk-equivalence}
\begin{cases}
\mathbf{w} \in \mathsf{VALID} \hspace*{2.5pt} \Longleftrightarrow \hspace*{2.5pt} \Gamma_\eta(\mathbf{w}) \in \mathsf{VALID}~\text{for any}~\eta\in\mathcal{S}, \\
\text{If } \mathbf{w} \in \mathsf{VALID} \text{ and } \eta \text{ is uniform in } \mathcal{S}, \text{ then }  \Gamma_\eta(\mathbf{w}) \text{ is uniform in } \mathsf{VALID}.
\end{cases}
\end{eqnarray}
The target is  to construct a statistical ZKAoK for the following abstract relation:
\begin{eqnarray*}
	\rho_{\mathrm{abstract}} = \big\{&(\mathbf{M}_i, \mathbf{u}_i)_{i\in\{1,2\}}, (\mathbf{w}_1\|\mathbf{w}_2) \in (\mathbb{Z}_{q_i}^{K_i \times L_i} \times \mathbb{Z}_{q_i}^{K_i})_{i\in\{1,2\}} \times \mathsf{VALID}: \\
	&\mathbf{M}_i\cdot \mathbf{w}_i = \mathbf{u}_ i\bmod q_i~\text{for~}i\in\{1,2\}\big\},
\end{eqnarray*}where some entries of $\mathbf{w}_1$ may appear in $\mathbf{w}_2$ and vice versa.  In other words, $\mathbf{w}_1$ and $\mathbf{w}_2$ are mutually related. 
To obtain the desired ZKAoK protocol, one has to prove that $\mathbf{w}=(\mathbf{w}_1\|\mathbf{w}_2)\in\mathsf{VALID}$ and $\mathbf{w}$ satisfies the two linear equations $\mathbf{M}_i\cdot \mathbf{w}_i = \mathbf{u}_i \bmod q_i$ for $i\in\{1,2\}$. To prove $\mathbf{w}\in\mathsf{VALID}$ in a zero-knowledge manner, the prover chooses $\eta \xleftarrow{\$}\mathcal{S}$ and allows the verifier to check $\Gamma_\eta(\mathbf{w}) \in \mathsf{VALID}$. According to the first condition in~(\ref{eq:zk-equivalence}), the verifier should be convinced that $\mathbf{w}$ is indeed from the set $\mathsf{VALID}$. At the same time, the verifier cannot learn any extra information about $\mathbf{w}$ due to the second condition in~(\ref{eq:zk-equivalence}). Furthermore, to prove in ZK that the linear equations hold, the prover   chooses $\{\mathbf{r}_{w_i}\xleftarrow{\$}\mathbb{Z}_{q_i}^{L_i}\}_{i\in\{1,2\}}$ as  masking vectors and then shows the verifier that the equation $\mathbf{M}_i\cdot (\mathbf{w}_i + \mathbf{r}_{w_i}) = \mathbf{M}_i\cdot \mathbf{r}_{w_i} + \mathbf{u}_i \bmod q_i$  holds for $i\in\{1,2\}$. 

In Figure~\ref{Figure:Interactive-Protocol}, we recall in detail  the interaction between two $\mathrm{PPT}$ algorithms prover $\mathcal{P}$ and verifier $\mathcal{V}$. The system utilizes a commitment scheme $\mathsf{COM}$ from~\cite{KTX08AC}: $\mathsf{COM}:\{0,1\}^*\times \{0,1\}^{m}\rightarrow \mathbb{Z}_q^{n}$. It is statistically hiding if $m\geq 2n\log q$ and computationally binding if $\mathsf{SIS}_{n,2m,q,1}^{\infty}$ problem is hard.  For $\mathbf{w}=(\mathbf{w}_1\|\mathbf{w}_2)\in \mathbb{Z}^L$ and $\mathbf{r}=(\mathbf{r}_1\|\mathbf{r}_2)\in\mathbb{Z}^L$, denote $\mathbf{w} \boxplus\mathbf{r}=(\mathbf{w}_1+\mathbf{r}_1 \bmod q_1\|\mathbf{w}_2+\mathbf{r}_2 \bmod q_2)$. Note that for any $\eta\in \mathcal{S}$, we have $\Gamma_{\eta}(\mathbf{w}\boxplus\mathbf{r})=\Gamma_{\eta}(\mathbf{w})\boxplus\Gamma_{\eta}(\mathbf{r})$. 
\begin{theorem}[\cite{LLNW17AC-ecash}]\label{Theorem:zk-protocol}
	Let $\mathsf{COM}$ be a statistically hiding and computationally binding commitment scheme. Then the interactive protocol depicted in Figure~\ref{Figure:Interactive-Protocol} is a statistical \emph{ZKAoK} with perfect completeness, soundness error~$2/3$, and communication cost~$\mathcal{O}(\sum_{i=1}^2L_i\log q_i)$. Specifically:
	\begin{itemize}
		\item There exists an efficient simulator  that on input $\{(\mathbf{M}_i, \mathbf{u}_i)\}_{i\in\{1,2\}}$,  with probability $2/3$ it outputs an accepted transcript that is within statistical distance from the one produced by an honest prover who knows the witness. 
		\item There exists an efficient  algorithm $\mathcal{E}$ that, takes as input $\{(\mathbf{M}_i, \mathbf{u}_i)\}_{i\in\{1.2\}}$ and accepting transcripts $(\mathrm{CMT},1,\mathrm{RSP}_1)$, $(\mathrm{CMT},2,\mathrm{RSP}_2)$,  $(\mathrm{CMT},3,\mathrm{RSP}_3)$, outputs   $\mathbf{w}'=(\mathbf{w}_1'\|\mathbf{w}_2') \in \mathsf{VALID}$ such that $\mathbf{M}_i\cdot \mathbf{w}_i' = \mathbf{u}_i \bmod q_i$ for $i\in \{1,2\}$.  
	\end{itemize}
\end{theorem}
We refer the readers to~\cite{LLNW17AC-ecash} for details  of the proof. 
\begin{figure}[htbp]
	\centering 

	\begin{enumerate}
		\item \textbf{Commitment:} The prover $\mathcal{P}$ samples $\mathbf{r}_{w_1} \xleftarrow{\$} \mathbb{Z}_{q_1}^{L_1}$, $\mathbf{r}_{w_2} \xleftarrow{\$} \mathbb{Z}_{q_2}^{L_2}$, $\eta \xleftarrow{\$} \mathcal{S}$ and randomness $\rho_1, \rho_2, \rho_3$ for $\mathsf{COM}$. Let $\mathbf{r}_w=(\mathbf{r}_{w_1}\|\mathbf{r}_{w_2})$ and $\mathbf{z}=\mathbf{w}\boxplus\mathbf{r}_w$. 
		Then he sends $\mathrm{CMT}= \big(C_1, C_2, C_3\big)$ to $\mathcal{V}$, where
		\begin{gather*}
		C_1 =  \mathsf{COM}(\eta, \{\mathbf{M}_i\cdot \mathbf{r}_{w_i} \bmod q_i\}_{i\in\{1,2\}}; \rho_1), \hspace*{5pt}
		C_2 =  \mathsf{COM}(\Gamma_{\eta}(\mathbf{r}_w); \rho_2), \\
		C_3 =  \mathsf{COM}(\Gamma_{\eta}(\mathbf{z}); \rho_3).
		\end{gather*}
		
		\item \textbf{Challenge:} $\mathcal{V}$ sends back a challenge $Ch \xleftarrow{\$} \{1,2,3\}$ to $\mathcal{P}$.
		\item \textbf{Response:} According to the choice of $Ch$, the prover $\mathcal{P}$ sends $\mathrm{RSP}$ computed in the following way:
		\smallskip
		\begin{itemize}
			\item $Ch = 1$: Let $\mathbf{t}_{w} = \Gamma_{\eta}(\mathbf{w})$, $\mathbf{t}_{r} = \Gamma_{\eta}(\mathbf{r}_w)$, and $\mathrm{RSP} = (\mathbf{t}_w, \mathbf{t}_r, \rho_2, \rho_3)$. \smallskip
			
			\item $Ch = 2$: Let $\eta_2 = \eta$, $\mathbf{z}_2 = \mathbf{z} $, and
			$\mathrm{RSP} = (\eta_2, \mathbf{z}_2, \rho_1, \rho_3)$. \smallskip
			\item $Ch = 3$: Let $\eta_3 = \eta$, $\mathbf{z}_3 = \mathbf{r}_w$, and
			$\mathrm{RSP} = (\eta_3, \mathbf{z}_3, \rho_1, \rho_2)$.
		\end{itemize}
	\end{enumerate}
	\textbf{Verification:}  When receiving $\mathrm{RSP}$ from the verifier $\mathcal{P}$, the prover $\mathcal{V}$ performs as follows:
	\begin{itemize}
		\item $Ch = 1$: Verify that $\mathbf{t}_w \in \mathsf{VALID}$, $C_2 = \mathsf{COM}(\mathbf{t}_r; \rho_2)$, ${C}_3 = \mathsf{COM}(\mathbf{t}_w \boxplus \mathbf{t}_r; \rho_3)$. \smallskip
		
		\item $Ch = 2$: Parse $\mathbf{z}_2=(\mathbf{z}_{2,1}\|\mathbf{z}_{2,2})$ such that $\mathbf{z}_{2,i}\in\mathbb{Z}^{L_i}$ for $i\in\{1,2\}$ and then verify that $C_1 = \mathsf{COM}(\eta_2, \{\mathbf{M}_i\cdot \mathbf{z}_{2,i}- \mathbf{u}_i \bmod q_i\}_{i\in \{1,2\}}; \rho_1)$, ${C}_3 = \mathsf{COM}(\Gamma_{\eta_2}(\mathbf{z}_2); \rho_3)$. \smallskip
		
		\item $Ch = 3$: Parse $\mathbf{z}_3=(\mathbf{z}_{3,1}\|\mathbf{z}_{3,2})$ such that $\mathbf{z}_{3,i}\in\mathbb{Z}^{L_i}$ for $i\in\{1,2\}$ and check that $C_1 =  \mathsf{COM}(\eta_3, \{\mathbf{M}_i\cdot \mathbf{z}_{3,i}\bmod q_i\}_{i\in\{1,2\}}; \rho_1), \hspace*{5pt}
		C_2 =  \mathsf{COM}(\Gamma_{\eta_3}(\mathbf{z}_3); \rho_2).$
		
	\end{itemize}
	In each case, if all the conditions hold, $\mathcal{V}$ outputs $1$.
	\caption{Stern-type statistical ZKAoK for the $\mathcal{NP}$-relation $\rho_{\mathrm{abstract}}$.}
	\label{Figure:Interactive-Protocol}
	      
\end{figure}

We remark that when one works with $\mathsf{SIS}$ or $\mathsf{LWE}$ associated equations, one does not have  the above abstract relation directly. For example, for $\mathsf{LWE}$ related equations, the secret vectors are usually $B$ bounded and there is no direct permutation such that conditions in~(\ref{eq:zk-equivalence}) hold. To solve this issue, Ling et al.~\cite{LNSW13PKC} developed decomposition-extension techniques that are essential to reduce the considered statement to an instance of the above abstract protocol. Looking ahead, in Section~\ref{tpbs-subsection:main-zk-protocol} we reduce the relations considered in Section~\ref{tpbs-subsection:description-of-our-construction} to an instance of the above abstract protocol.

\subsection{Description of Our Scheme}\label{tpbs-subsection:description-of-our-construction}
To obtain a concrete construction of TPBS from lattice assumptions, we will choose the required building blocks specified in our generic construction in Section~\ref{tpbs-subsection:gener-con}. As for zero-knowledge techniques, we employ Stern-like protocols~\cite{KTX08AC,LNSW13PKC}, which are the most promising choice for our purpose due to their versatility and extendability. 

The space of user identities is set as $\mathcal{ID}\subseteq \{0,1\}^{\ell_1}$ with $\ell_1=\mathcal{O}(\log n)$. This is sufficient since we work with polynomial number of users. With regard to the policy language, we follow Cheng et al.~\cite{CNW16DCC}, who came up with an instantiation that captures policies in many real-life scenarios towards construction of their lattice-based PBS.\remove{Specifically, the policy language they designed allows for exponentially many messages and polynomially many policies such that a policy permits many messages while a message satisfies  many polices.} Let $\ell_2=\mathcal{O}(\log n)$ and $d$ be an integer such that $n-\ell_2<d$. A policy checker is specified by two matrices $\mathbf{G}_1\in \mathbb{Z}_2^{n\times \ell_2}$ and $\mathbf{G}_2\in \mathbb{Z}_2^{n\times d}$, denoted as $\mathsf{PC}_{\mathbf{G}_1,\mathbf{G}_2}$. A message $\mathbf{m}\in\{0,1\}^n$ satisfies a policy $\mathbf{p}\in\{0,1\}^{\ell_2}$ if there exists  $\mathbf{q}\in \{0,1\}^{d}$ such that  
\begin{eqnarray}\label{tpbs-eq:policy-checker-equation}
 \mathbf{G}_1\cdot \mathbf{p}+\mathbf{G}_2\cdot \mathbf{q}=\mathbf{m}\bmod 2. 
\end{eqnarray}
The associated language is $$\mathcal{L}(\mathsf{PC}_{\mathbf{G}_1,\mathbf{G}_2})=\{(\mathbf{p},\mathbf{m}): \exists~\mathbf{q}\in\{0,1\}^{d}~\text{s.t.}~\mathbf{G}_1\cdot \mathbf{p}+\mathbf{G}_2\cdot \mathbf{q}=\mathbf{m}\bmod 2\}.$$
Let $\ell=\ell_1+\ell_2$. Instead of using Bonsai signature scheme~\cite{CHKP10EC} as in~\cite{CNW16DCC}, we choose Boyen signature scheme~\cite{Boyen10PKC}. This will reduce the public key size by a factor close to $2$ and user signing key size  by a factor of $\ell/2$. For  encryption scheme, we start with  GPV-IBE~\cite{GPV08STOC},   and then transform it  to an IND-CCA secure encryption  by using a strong one time signature following the CHK technique~\cite{CHK04EC}.  
When user $\mathsf{id}$ signs a message $\mathbf{m}$, it has to generate a zero-knowledge proof showing that  (1)  it possess a valid signature on $\mathsf{id}\|\mathbf{p}$ for the Boyen signature; (2) it has encrypted $\mathsf{id}$ correctly using GPV-IBE; (3) there exists a vector $\mathbf{q}$ such that the above policy relation is satisfied for $(\mathbf{p},\mathbf{m})$. Even a relation similar to the combination of  (1) and (2) is addressed in~\cite{LNW15PKC} and a relation that contains (3) as a sub-statement is addressed in~\cite{CNW16DCC}, it is not straightforward to obtain our zero-knowledge protocol. One reason is that the relation considered in~\cite{LNW15PKC} has encrypted message to be $\mathsf{id}\|\mathbf{p}$ while our relation has encrypted message  $\mathsf{id}$, which makes the considered relation different. Another reason is that we have to show that (1), (2), (3)  are satisfied simultaneously. We manage to do so by carefully utilizing the flexibility and extendability of Stern-like protocols. Furthermore, we employ  the more recent extension-permutation techniques from~\cite{LNRW18TCS,LNWX18PKC} to achieve \emph{optimal} permutation size that is exactly the bit size of the secret input (denoted as $|\xi|$). This improves the signature size slightly and is  preferable to the  suboptimal permutation size ($\mathcal{O}(|\xi|\cdot\log |\xi|)$) if we use the same extension-permutation techniques in~\cite{LNW15PKC,CNW16DCC}.   Details of our zero-knowledge protocol are described in Section~\ref{tpbs-subsection:main-zk-protocol}. For completeness, we  recall Boyen's signature and the GPV-IBE scheme in~\ref{tpbs-subsection:primitives-for-concrete-construction}.

Now that we have established all the building blocks, our construction of TPBS follows smoothly. We present it in the below. 
\begin{description}
	\item[$\mathsf{Setup}(1^{\lambda})$] Given the security parameter $\lambda$, it first specifies public parameters as follows: \begin{itemize}
		\item Message length $n=\mathcal{O}(\lambda)$, user identity length $\ell_1=\mathcal{O}(\log n)$, policy length $\ell_2=\mathcal{O}(\log n)$, witness length $d$ such that $\ell_2+d >n $. Policy specifying matrix $\mathbf{G}_1\in \mathbb{Z}_2^{n\times \ell_2}$ and $\mathbf{G}_2\in \mathbb{Z}_2^{n\times d}$. Define $\ell=\ell_1+\ell_2$. 
		\item Modulus $q=\mathcal{O}(\ell n^2)$, $m\geq 2n\log q$. 
		\item Two real numbers $s=\omega(\sqrt{\log m})\cdot \mathcal{O}(\sqrt{\ell n\log q})$ and $s_1=\omega(\log m)$. Two integer bounds $\beta=\lceil s\cdot \log n\rceil$ and $B=\widetilde{\mathcal{O}}(\sqrt{n})$. Let $\chi$ be an efficiently sample distribution over integers $\mathbb{Z}$ that outputs a sample $e$ with $|e|\leq B$. 
		
		\item Two hash functions $\mathcal{H}_1:\{0,1\}^*\rightarrow \mathbb{Z}_q^{n\times \ell_1}$ and $\mathcal{H}_2:\{0,1\}^*\rightarrow \{1,2,3\}^{\kappa}$, the latter of which will be modeled as random oracle in the security proof.  
		\item A strong one time signature scheme $\mathcal{OTS}=(\mathsf{OGen},\mathsf{OSign},\mathsf{OVerify})$ to apply CHK transform~\cite{CHK04EC}. 
		\item A statistically hiding and computationally binding commitment scheme from~\cite{KTX08AC}: $\mathsf{COM}: \{0,1\}^*\times \{0,1\}^m\rightarrow \mathbb{Z}_q^n$ for our proof system. 
		\item A number of protocol repetitions $\kappa=\omega(\log \lambda)$. \remove{ of our statistical ZKAoK protocol. }
	\end{itemize}

\smallskip
In addition, the algorithm generates key pair $\mathsf{mvk}=(\mathbf{A},\mathbf{A}_0,\ldots, \mathbf{A}_{\ell},\mathbf{u})\in(\mathbb{Z}_q^{n\times m})^{\ell+2}\times \mathbb{Z}_q^{n}$ and $\mathsf{msk}=\mathbf{S}\in\mathbb{Z}^{m\times m}$ for Boyen signature scheme~\cite{Boyen10PKC}, key pair $\mathsf{mek}=\mathbf{B}$ and $\mathsf{mdk}=\mathbf{T}$ for GPV-IBE scheme~\cite{GPV08STOC}, where $(\mathbf{A},\mathbf{S})$ and $(\mathbf{B},\mathbf{T})$ are generated via $\mathsf{TrapGen}(n,m,q)$. The public parameter $\mathsf{pp}$ then contains all parameters together with $\mathsf{mvk},\mathsf{mek}$. It returns $\mathsf{pp},\mathsf{msk},\mathsf{mdk}$. 

\smallskip
\smallskip
	
	\item[$\mathsf{KeyGen}(\mathsf{msk},\mathsf{id}\in\{0,1\}^{\ell_1},\mathcal{P}_{\mathsf{id}})$] For all  $\mathbf{p}\in\{0,1\}^{\ell_2} \in \mathcal{P}_{\mathsf{id}}$, we sign $\mathsf{id}\|\mathbf{p}$ using $\mathsf{msk}=\mathbf{S}$ via algorithm $\mathsf{SampleD}\big( \mathsf{ExtBasis}(\mathbf{S},\mathbf{A}_{\mathsf{id}\|\mathbf{p}}),\mathbf{A}_{\mathsf{id}\|\mathbf{p}}, \mathbf{u},s)\big)$, where $\mathbf{A}_{\mathsf{id}\|\mathbf{p}}=[\mathbf{A}|\mathbf{A}_0+\sum_{j=1}^{\ell_1}\mathsf{id}[j]\cdot\mathbf{A}_j+\sum_{j=1}^{\ell_2}\mathbf{p}[j]\cdot\mathbf{A}_{\ell_1+j}]$, obtaining a signature $\mathbf{v}_{\mathsf{id}\|p}\in \Lambda^{\mathbf{u}}(\mathbf{A}_{\mathsf{id}\|\mathbf{p}})$ satisfying the following: 
	\begin{eqnarray}\label{tpbs-eq:boyen-signature-condition}
	\mathbf{A}_{\mathsf{id}\|\mathbf{p}}\cdot \mathbf{v}_{\mathsf{id}\|\mathbf{p}}=\mathbf{u}\bmod q;\hspace*{12pt} 
	\|\mathbf{v}_{\mathsf{id}\|\mathbf{p}}\|_{\infty} \leq \beta. 
	\end{eqnarray}
	It returns 
	$\mathsf{usk}_{\mathsf{id}}=\big(\mathsf{id}, \{p,\mathbf{v}_{\mathsf{id}\|p}: p\in \mathcal{P}_{\mathsf{id}}\}\big)$.   
\smallskip\smallskip
	
	\item[$\mathsf{Sign}(\mathsf{usk}_{\mathsf{id}},\mathbf{m}\in\{0,1\}^n,\mathbf{q}\in\{0,1\}^d)$] Parse $\mathsf{usk}_{\mathsf{id}}=\big(\mathsf{id}, \{p,\mathbf{v}_{\mathsf{id}\|p}: p\in \mathcal{P}_{\mathsf{id}}\}\big)$.   If $\exists~ \mathbf{p}\in\mathcal{P}_{\mathsf{id}}$ such that $\mathbf{G}_1\cdot \mathbf{p}+\mathbf{G}_2\cdot 
	\mathbf{q}=\mathbf{m}\bmod 2$, then it does following. 
	\begin{itemize}
		\item It first generates a one time key pair $(\mathsf{ovk},\mathsf{osk})\leftarrow \mathsf{OGen}(1^{\lambda})$. 
		\item It next encrypts $\mathsf{id}$ with respect to ``identity'' $\mathsf{ovk}$. Specifically, it computes $\mathbf{G}=\mathcal{H}_1(\mathsf{ovk})\in\mathbb{Z}_q^{n\times \ell_1}$,  samples $\mathbf{s}\hookleftarrow\chi^n$, $\mathbf{e}_1\hookleftarrow\chi^m$,  $\mathbf{e}_2\hookleftarrow\chi^{\ell_1}$, and computes $(\mathbf{c}_1,\mathbf{c}_2)\in \mathbb{Z}_q^m\times \mathbb{Z}_q^{\ell_1}$ as 
		\begin{equation}\label{tpbs-eq:gpv-encryption-condition}
		\mathbf{c}_1=\mathbf{B}^\top \cdot \mathbf{s}+\mathbf{e}_1; ~~	\mathbf{c}_2=\mathbf{G}^\top \cdot \mathbf{s}+\mathbf{e}_2+\mathsf{id}\cdot \lfloor \frac{q}{2}\rfloor.  
		\end{equation} 
		\item Then it generates a NIZKAoK proof $\pi$ to show possession of a tuple $$\xi=(\mathsf{id}\|\mathbf{p},\mathbf{v}_{\mathsf{id}\|\mathbf{p}}, \mathbf{s},\mathbf{e}_1,\mathbf{e}_2,\mathbf{q})$$ such that equations~(\ref{tpbs-eq:boyen-signature-condition}), (\ref{tpbs-eq:gpv-encryption-condition}), (\ref{tpbs-eq:policy-checker-equation}) holds and that 
		\begin{equation}\label{tpbs-eq:gpv-encryptioon-condition-about-bound}
		\|\mathbf{s}\|_{\infty}\leq B, ~~	\|\mathbf{e}_1\|_{\infty}\leq B, ~~	\|\mathbf{e}_2\|_{\infty}\leq B. 
		\end{equation}
		This is done by running the statistical ZKAoK in Section~\ref{tpbs-subsection:main-zk-protocol} with public input 
		$\zeta=(\mathbf{A},\mathbf{A}_0,\ldots,\mathbf{A}_{\ell},\mathbf{u},\mathbf{B},\mathbf{G},\mathbf{c}_1,\mathbf{c}_2,\mathbf{G}_1,\mathbf{G}_2,\mathbf{m})$ and secret input $\xi$ as above. Then the protocol is repeated $\kappa$ times to achieve negligible soundness error and made non-interactive via Fiat-Shamir transform~\cite{FS86C}. The result proof is a triple $\pi=((\mathrm{CMT}_i)_{i=1}^{\kappa},\mathrm{CH},(\mathrm{RSP}_i)_{i=1}^{\kappa})$ with $\mathrm{CH}=\mathcal{H}_{2}(\zeta,(\mathrm{CMT}_i)_{i=1}^{\kappa}).$
		\item Finally, it runs the algorithm $ \mathsf{OSign}(\mathsf{osk}; \mathbf{c}_1, \mathbf{c}_2, \pi)$ to obtain a one time signature $\mathrm{sig}$ on the tuple $(\mathbf{c}_1, \mathbf{c}_2, \pi)$. Let $\sigma= (\mathsf{ovk},\mathbf{c}_1, \mathbf{c}_2, \pi,\mathrm{sig})$. Return~$\sigma$. 
		
	\end{itemize}  
	\smallskip
	Otherwise, this algorithm returns $\bot$. 

\smallskip\smallskip
	\item[$\mathsf{Verify}(\mathsf{pp},\mathbf{m},\sigma)$] Let $\sigma= (\mathsf{ovk},\mathbf{c}_1, \mathbf{c}_2, \pi,\mathrm{sig})$. The algorithm proceeds as follows. 
	\begin{itemize}
		\item  It first runs the verification algorithm  $\mathsf{OVerfiy}(\mathsf{ovk},(\mathbf{c}_1,\mathbf{c}_2,\pi),\mathrm{sig})$. Return~$0$ if $\mathsf{OVerfiy}$ returns~$0$. 
		
		\item It then parses $\pi=((\mathrm{CMT}_i)_{i=1}^{\kappa},\mathrm{CH}=(\mathrm{Ch}_1,\ldots,\mathrm{Ch}_{\kappa}),(\mathrm{RSP}_i)_{i=1}^{\kappa})$.  Return~$0$ if $\mathrm{CH}\neq \mathcal{H}_{2}(\zeta,(\mathrm{CMT}_i)_{i=1}^{\kappa})$. 
		
		\item  Next, for each $i\in[\kappa]$, it runs the verification step of the protocol in Section~\ref{tpbs-subsection:main-zk-protocol} to check the validity of $\mathrm{RSP}_i$ with respect to $\mathrm{CMT}_i$ and $\mathrm{Ch}_i$. Return~$0$ if any of the verification does not pass. Otherwise return~$1$. 
	\end{itemize}
	
	\item[$\mathsf{Open}(\mathsf{mdk},\mathbf{m},\sigma)$] Parse $\sigma= (\mathsf{ovk},\mathbf{c}_1, \mathbf{c}_2, \pi,\mathrm{sig})$. Return $\bot$ if $\mathsf{Verify}(\mathsf{pp},\mathbf{m},\sigma)=0$. Else, this algorithm  opens the signature using $\mathsf{mdk}=\mathbf{T}$ as follows. 
	\begin{itemize}
		\item Compute $\mathbf{G}=\mathcal{H}_1(\mathsf{ovk})\stackrel{\triangle}{=}[\mathbf{g}_1|\cdots|\mathbf{g}_{\ell_1}]$. Run  $\mathbf{f}_i\leftarrow \mathsf{SampleD}(\mathbf{B},\mathbf{T},\mathbf{g}_i,s_1)$ for $i\in[\ell_1]$. Define the decryption key  with respect to ``identity'' $\mathsf{ovk}$ as  $\mathbf{F}_{\mathsf{iden}}=[\mathbf{f}_1|\cdots|\mathbf{f}_{\ell_1}]\in\mathbb{Z}_q^{m\times\ell_1}$.   Note that $\mathbf{B}\cdot \mathbf{F}_{\mathsf{ovk}}=\mathbf{G}\bmod q$. 
		\item Decrypt $(\mathbf{c}_1,\mathbf{c}_2)$ using $\mathbf{F}_{\mathsf{ovk}}$ by computing $$\mathsf{id}'=\lfloor\frac{ \mathbf{c}_2-\mathbf{F}_{\mathsf{ovk}}^{\top}\cdot \mathbf{c}_1}{\lfloor\frac{q}{2}\rfloor}\rceil. $$ Return~$\mathsf{id}'$. 
	\end{itemize}
	
\end{description}
\noindent{\bf Asymptotic Efficiency.} We first analyze the efficiency of our construction with respect to security parameter $\lambda$. 
\begin{itemize}
	\item Public parameter $\mathsf{pp}$ is dominated by the public key of the underlying encryption scheme and signature scheme, which has bit size $\mathcal{O}(\ell \lambda^2 \log^2 q)=\widetilde{\mathcal{O}}(\ell \lambda^2)$. The bit size of $\mathsf{msk}$ and $\mathsf{mdk}$ is $\mathcal{O}(\lambda^2\log^3 \lambda)=\widetilde{\mathcal{O}}(\lambda^2)$.  \smallskip
	
	\item The bit size of  user secret key $\mathsf{usk}_{\mathsf{id}}$ is dominated by those of Boyen signatures, which is $\mathcal{O}(\lambda  \log \lambda\cdot c_{\mathsf{id}})=\widetilde{\mathcal{O}}(c_{\mathsf{id}}\cdot\lambda)$ with $c_{\mathsf{id}}=|\mathcal{P}_{\mathsf{id}}|$.  \smallskip
	\item The bit size of signature is dominated by that of NIZKAoK proof $\pi$, which is $\mathcal{O}(L_1\log q+L_2)\cdot\omega(\log \lambda)=\widetilde{\mathcal{O}}(\ell\lambda)$. Note that $L_1,L_2$ are the bit sizes of witness vectors $\mathbf{w}_1,\mathbf{w}_2$ in Section~\ref{tpbs-subsection:main-zk-protocol}  and $\mathcal{O}(L_1\log q+L_2)=\mathcal{O}(\ell \lambda\log^3 \lambda)$. 
\end{itemize} 
\noindent{\bf Correctness.} Correctness of the above construction relies on the following facts: (1) the underlying zero-knowledge protocol is perfectly complete; (2) the GPV-IBE scheme~\cite{GPV08STOC} for the choice of parameters is correct.   

Correctness of $\mathsf{Verify}$ algorithm directly follows from fact (1). As for the correctness of  $\mathsf{Open}$ algorithm, note that 
\begin{eqnarray*}
\mathbf{c}_2-\mathbf{F}_{\mathsf{ovk}}^{\top}\cdot \mathbf{c}_1&=&\mathbf{G}^\top\cdot \mathbf{s}+\mathbf{e}_2+\mathsf{id}\cdot \lfloor \frac{q}{2}\rfloor -\mathbf{F}_{\mathsf{ovk}}^{\top}\cdot \mathbf{B}^\top\cdot \mathbf{s}-\mathbf{F}_{\mathsf{ovk}}^{\top}\cdot \mathbf{e}_1\\
&=&\mathbf{e}_2-\mathbf{F}_{\mathsf{ovk}}^{\top}\cdot \mathbf{e}_1+\mathsf{id}\cdot \lfloor \frac{q}{2}\rfloor
\end{eqnarray*}
Recall that $\|\mathbf{e}_1\|_{\infty}\leq B$, $\|\mathbf{e}_2\|_{\infty} \leq B$, $B=\widetilde{\mathcal{O}}(\sqrt{n})$, and each column of $\mathbf{F}_{\mathsf{ovk}}$ is obtain via algorithm $\mathsf{SampleD}(\mathbf{B},\mathbf{T},\cdot, s_1)$ with $s_1=\omega(\log m)$. Therefore,  we have $\|\mathbf{F}_{\mathsf{ovk}}\|_{\infty} \leq \lceil s_1\log m \rceil$. Hence $$\|\mathbf{e}_2-\mathbf{F}_{\mathsf{ovk}}^{\top}\cdot \mathbf{e}_1\|_{\infty} \leq B+mB \cdot\lceil s_1\log m \rceil =\widetilde{\mathcal{O}}(n^{1.5}) \leq \lceil\frac{q}{5}\rceil =\mathcal{O}(\ell n^2). $$ Therefore, the rounding algorithm described in $\mathsf{Open}$ returns $\mathsf{id}$ with overwhelming probability. 

\smallskip

\noindent {\bf Security.} We summarize the security of our scheme  in the following theorem.
\begin{theorem}
	In the random oracle model, assuming   hardness of $\mathsf{SIVP}_{\widetilde{\mathcal{O}}(\ell \cdot n^2)}$ in the worst case, our scheme satisfies simulatability and extractability defined in Section~\ref{tpbs-section:tpbs-security-model}. 
\end{theorem}
\begin{proof}
	In Theorem~\ref{tpbs-thm: EXT} and Theorem~\ref{tpbs-thm: SIM}, we showed that the generic construction satisfies extractability and simulatability  if the underlying signature scheme is EU-CMA, the underlying encryption scheme is IND-CCA, and the underlying proof system is SE-NIZK. However, the security of our lattice-based construction is not straightforward due to the random oracle model. More specifically, we do not have a trapdoor $\mathsf{tr}$ when we run $\mathsf{SimSetup}_{\mathsf{nizk}}$, which makes algorithms $\mathsf{SimProve}$ and $\mathsf{Extr}_{\mathsf{nizk}}$ different from those in the standard model. For the sake of presentation, we provide these two algorithms for our proof system in~\ref{tpbs-section:deferred-simprove+extr-algorithms}. The security of our construction will then follow from Theorem~\ref{tpbs-thm: EXT} and Theorem~\ref{tpbs-thm: SIM}. In addition,  
	\begin{itemize}
		\item 	The GPV-IBE scheme, via CHK transform~\cite{CHK04EC} is IND-CCA secure assuming the hardness of $\mathsf{LWE}_{n,q,\chi}$ (see~\cite{GPV08STOC}), which in turn relies on hardness of $\mathsf{SIVP}_{\gamma}$ for $\gamma=\widetilde{\mathcal{O}}(n\cdot q/B)=\widetilde{\mathcal{O}}(\ell\cdot n^2)$. 
		\item Boyen signature is EU-CMA assuming hardness of $\mathsf{SIS}_{n,m,q,\widetilde{\mathcal{O}}(\ell\cdot n)}^{\infty}$ (see~\cite{Boyen10PKC,MP12EC}), which in turn relies on  worst-case hardness of $\mathsf{SIVP}_{\gamma}$ for $\gamma=\widetilde{\mathcal{O}}(\ell\cdot n)\cdot \widetilde{\mathcal{O}}(\sqrt{nm})=\widetilde{\mathcal{O}}(\ell\cdot n)$.
		
		\item The commitment scheme $\mathsf{COM}$ utilized in our proof system is statistically hiding and computationally binding assuming hardness of $\mathsf{SIS}_{n,2m,q,1}^{\infty}$, which in turn depends on  worst-case hardness of $\mathsf{SIVP}_{\widetilde{\mathcal{O}}(n)}$. \remove{Recall that~\cite{FKMV12IndoC} shows that applying Fiat-Shamir transform~\cite{FS86C} to our statistical ZKAoK protocol results in a SE-NIZK protocol.  }
			\end{itemize}
	Therefore, our scheme is simulatable and extractable if $\mathsf{SIVP}_{\widetilde{\mathcal{O}}(\ell \cdot n^2)}$ is hard. 
\end{proof}

\section{The Underlying Zero-Knowledge Argument System}\label{tpbs-subsection:main-zk-protocol}

This section presents our statistical ZKAoK of secret vector $\xi$ such that it satisfies equations~(\ref{tpbs-eq:boyen-signature-condition}), (\ref{tpbs-eq:gpv-encryption-condition}), (\ref{tpbs-eq:policy-checker-equation}), (\ref{tpbs-eq:gpv-encryptioon-condition-about-bound}), which will be invoked by user when signing messages. The target is to reduce those statements to an instance of the abstract relation described in Section~\ref{tpbs-subsection:Stern-like-protocols} such that conditions in~(\ref{eq:zk-equivalence}) hold. To this end, we first recall the decomposition technique from~\cite{LNSW13PKC} to unify our considered statements into equations of the form $\{\widehat{\mathbf{M}}_i\cdot \widehat{\mathbf{w}}_i= \mathbf{u}_i \bmod q_i\}_{i\in \{1,2\}}$ such that $\|\widehat{\mathbf{w}}_i\|_{\infty}\leq 1$ for $i\in\{1,2\}$.  Then we employ the extension-permutation techniques from~\cite{LNRW18TCS,LNWX18PKC} instead of those presented in~\cite{LNW15PKC,CNW16DCC},  which is crucial to achieve optimal permutation size, to specify a set  $\mathsf{VALID}$ that contains our extended secret vector $\mathbf{w}=(\mathbf{w}_1\|\mathbf{w}_2)$ and a permutation $\Gamma_{\eta}$ such that conditions in~(\ref{eq:zk-equivalence}) holds.  
\smallskip 

\noindent{\bf Decomposition.} For a positive integer $B\geq 2$,  let   $\delta_B:=\lfloor \log_2 B\rfloor +1 = \lceil \log_2(B+1)\rceil$ and the sequence $B_1, \ldots, B_{\delta_B}$, where $B_j = \lfloor\frac{B + 2^{j-1}}{2^j} \rfloor$, for any $  j \in [\delta_B]$. It is then verifiable that $\sum_{j=1}^{\delta_B} B_j = B$. In addition, for any integer $a\in[0,B]$, one can decompose $a$ into a vector of the form $\mathsf{idec}_B(a)=[a^{(1)}|a^{(2)}| \cdots| a^{(\delta_B)})^\top \in \{0,1\}^{\delta_B}$, satisfying that $[B_1|B_2|\cdots|B_{\delta_B}]\cdot \mathsf{idec}_B(a)=a$. The procedure of the decomposition is presented below in a deterministic manner.
\begin{enumerate}
	\item $a': = a$.
	\item For $j=1$ to $\delta_B$ do:
	\begin{enumerate}[(i)]
		\item If $a' \geq B_j$ then $a^{(j)}: = 1$, else $a^{(j)}: = 0$;
		\item $a': = a' - B_j\cdot a^{(j)}$.
	\end{enumerate}
	\item Output $\mathsf{idec}_B(a) = [a^{(1)}| a^{(2)}| \cdots| a^{(\delta_B)}]^\top$.
\end{enumerate}
\noindent When dealing with vectors of dimension $m$ and of range $[-B,B]$, we define $\mathsf{vdec}_{m,B}$ that maps $\mathbf{a}=[a_1|a_2|\cdots|a_m]^\top$ to a vector in $\{-1,0,1\}^{m\delta_B}$ of the following form: $$\mathbf{a}'=(\sigma(a_1)\cdot \mathsf{idec}_{B}(|a_1|)\|\sigma(a_2)\cdot \mathsf{idec}_{B}(|a_2|)\|\cdots\|\sigma(a_m)\cdot \mathsf{idec}_B(|a_m|)),$$
where  $\forall j\in[m]$: $\sigma(a_j)=0$ if $a_j=0$;  $\sigma(a_j)=-1$ if $a_j<0$; $\sigma(a_j)=1$ if $a_j>0$.

Define a matrix $\mathbf{G}_{m,B} \in \mathbb{Z}^{m \times m\delta_B}$ to be \begin{eqnarray*}
	\mathbf{G}_{m,B} = \begin{bmatrix} B_1 \ldots  B_{\delta_B} &  & & & \\
		&   &  &  \ddots  &  \\
		&   &  &    & B_1 \ldots  B_{\delta_B}  \\
	\end{bmatrix}
\end{eqnarray*}
Then we have \begin{eqnarray}\label{tpbs-eq:vector-decomposition}
\mathbf{a} = \mathbf{G}_{m,B}\cdot \mathsf{vdec}_{m,B}(\mathbf{a}).
\end{eqnarray}
\noindent {\bf Our statistical ZKAoK.} The goal is to prove knowledge of $\xi$ as described in Section~\ref{tpbs-subsection:description-of-our-construction} so that equations~(\ref{tpbs-eq:boyen-signature-condition}), (\ref{tpbs-eq:gpv-encryption-condition}), (\ref{tpbs-eq:policy-checker-equation}), (\ref{tpbs-eq:gpv-encryptioon-condition-about-bound}) hold. For completeness, we recall it below and then describe our protocol. 
\begin{description}
	\item[Public input $\zeta$:]  	$\mathbf{A},\mathbf{A}_0,\ldots,\mathbf{A}_{\ell}\in\mathbb{Z}_q^{n\times m},\mathbf{u}\in\mathbb{Z}_q^{n},\mathbf{B}\in\mathbb{Z}_q^{n\times m},\mathbf{G}\in\mathbb{Z}_q^{n\times\ell_1}$, $\mathbf{c}_1\in\mathbb{Z}_q^{ m},\mathbf{c}_2\in\mathbb{Z}_q^{\ell_1},\mathbf{G}_1\in\mathbb{Z}_2^{n\times \ell_2},\mathbf{G}_2\in\mathbb{Z}_2^{n\times d}$, $\mathbf{m}\in\{0,1\}^n$.  
	\item[Secret input $\xi$:] $\mathsf{id}\|\mathbf{p}\in\{0,1\}^{\ell_1+\ell_2}$, $\mathbf{v}_{\mathsf{id}\|\mathbf{p}}\in \mathbb{Z}_q^{2m}$, $ \mathbf{s}\in\mathbb{Z}_q^{n},\mathbf{e}_1\in\mathbb{Z}_q^{m},\mathbf{e}_2\in\mathbb{Z}_q^{\ell_1}$, $\mathbf{q}\in\{0,1\}^d$.  
	\item[Prover's goal:]  
	\begin{equation} \label{tpbs-eq:prover-goal-0}
	\begin{cases}
	[\mathbf{A}|\mathbf{A}_0+\sum_{j=1}^{\ell_1}\mathsf{id}[j]\cdot\mathbf{A}_j+\sum_{j=1}^{\ell_2}\mathbf{p}[j]\cdot\mathbf{A}_{\ell_1+j}]\cdot \mathbf{v}_{\mathsf{id}\|\mathbf{p}}=\mathbf{u}\bmod q;\vspace*{2.2pt}\\
	\mathbf{c}_1=\mathbf{B}^\top \cdot \mathbf{s}+\mathbf{e}_1 \bmod q; ~~	\mathbf{c}_2=\mathbf{G}^\top \cdot \mathbf{s}+\mathbf{e}_2+\mathsf{id}\cdot \lfloor \frac{q}{2}\rfloor \bmod q;\\
	\mathbf{G}_1\cdot \mathbf{p}+\mathbf{G}_2\cdot \mathbf{q}=\mathbf{m}\bmod 2; \\
	\|\mathbf{v}_{\mathsf{id}\|\mathbf{p}}\|_{\infty} \leq \beta;~~	\|\mathbf{s}\|_{\infty}\leq B;~~	\|\mathbf{e}_1\|_{\infty}\leq B;~~	\|\mathbf{e}_2\|_{\infty}\leq B. 
	\end{cases}
	\end{equation}
\end{description}
\noindent{\sc Decomposing-Unifying.}  Note that the secret vectors in the equation of third line of~(\ref{tpbs-eq:prover-goal-0}) already have  infinity norm $1$. So we focus on secret vectors whose infinity norm bound is not $1$. Denote $\mathbf{v}_{\mathsf{id}\|\mathbf{p}}=(\mathbf{v}_1\|\mathbf{v}_2)$ such that $\mathbf{v}_1,\mathbf{v}_2\in\mathbb{Z}_q^{m}$. Let 
\begin{equation*}
\begin{cases}
\widehat{\mathbf{v}}_1=\mathsf{vdec}_{m,\beta}(\mathbf{v}_1)\in\{-1,0,1\}^{m\delta_{\beta}},~ \widehat{\mathbf{v}}_2=\mathsf{vdec}_{m,\beta}(\mathbf{v}_2)\in\{-1,0,1\}^{m\delta_{\beta}}; \\
\widehat{\mathbf{A}}=\mathbf{A}\cdot \mathbf{G}_{m,\beta},~~ \widehat{\mathbf{A}}_i=\mathbf{A}_i\cdot \mathbf{G}_{m,\beta}, ~\text{for}~i\in[0,\ell].
\end{cases}
\end{equation*} 
According to~(\ref{tpbs-eq:vector-decomposition}), the first equation in~(\ref{tpbs-eq:prover-goal-0}) is now equivalent to 
\begin{equation}\label{tpbs-eq:prover-goal-1-1}
\widehat{\mathbf{A}}\cdot \widehat{\mathbf{v}}_1+\widehat{\mathbf{A}}_0\cdot \widehat{\mathbf{v}}_2+\sum_{j=1}^{\ell_1}\widehat{\mathbf{A}}_j\cdot\mathsf{id}[j] \widehat{\mathbf{v}}_2+\sum_{j=1}^{\ell_2}\widehat{\mathbf{A}}_{\ell_1+j}\cdot\mathbf{p}[j]\widehat{\mathbf{v}}_{2}=\mathbf{u}\bmod q.\end{equation}
Similarly, we decompose $\mathbf{s},\mathbf{e}_1,\mathbf{e}_2$ and form new matrices $\widehat{ \mathbf{B}}$, $\widehat{ \mathbf{G}}$ as follows. 
\begin{equation*}
\begin{cases}
\widehat{\mathbf{s}}=\mathsf{vdec}_{n,B}(\mathbf{s})\in\{-1,0,1\}^{n\delta_{B}},~
\widehat{\mathbf{e}}_1=\mathsf{vdec}_{m,B}(\mathbf{e}_1)\in\{-1,0,1\}^{m\delta_{B}},\\ 
\widehat{\mathbf{e}}_2=\mathsf{vdec}_{\ell_1,B}(\mathbf{e}_2)\in\{-1,0,1\}^{\ell_1\delta_{B}};\\
\widehat{\mathbf{B}}=\mathbf{B}\cdot \mathbf{G}_{n,B}^\top,~~ \widehat{\mathbf{G}}=\mathbf{G}\cdot \mathbf{G}_{n,B}^\top. 
\end{cases}
\end{equation*}
Therefore, based on equation~(\ref{tpbs-eq:vector-decomposition}), the second line in~(\ref{tpbs-eq:prover-goal-0}) is now equivalent to 
\begin{equation}\label{tpbs-eq:prover-goal-1-2}
\mathbf{c}_1=\widehat{\mathbf{B}}\cdot\widehat{\mathbf{s}}+ \mathbf{G}_{m,B} \cdot\widehat{\mathbf{e}}_1\bmod q;~~ \mathbf{c}_2=\widehat{\mathbf{G}} \cdot\widehat{ \mathbf{s}}+ \mathbf{G}_{\ell_1,B} \cdot\widehat{\mathbf{e}}_2+\mathsf{id}\cdot \lfloor \frac{q}{2}\rfloor\bmod q.
\end{equation}
For simplicity, we let
\begin{eqnarray*}
	\widehat{\mathbf{w}}_{1,1}&\stackrel{\triangle}{=}&\mathsf{mix}(\mathsf{id}\|\mathbf{p},\widehat{ \mathbf{v}}_2)=(\mathsf{id}[1]\widehat{ \mathbf{v}}_2\|\cdots\|\mathsf{id}[\ell_1]\widehat{ \mathbf{v}}_2\|\mathbf{p}[1]\widehat{ \mathbf{v}}_2\|\cdots\|\mathbf{p}[\ell_2]\widehat{ \mathbf{v}}_2)\in\{-1,0,1\}^{\ell m\delta_{\beta}}\vspace*{2.2pt}\\
	\widehat{\mathbf{w}}_{1,2}&=&(\widehat{\mathbf{s}}\|\widehat{\mathbf{e}}_1\|\widehat{\mathbf{e}}_2)\in\{-1,0,1\}^{(n+m+\ell_1)\delta_{B}}
\end{eqnarray*} 
Let $\widehat{L}_1=2m\delta_{\beta}+\ell m\delta_{\beta}+(n+m+\ell_1)\delta_{B}+\ell_1$ and 
form secret vector $\widehat{ \mathbf{w}}_1=(\widehat{ \mathbf{v}}_1\|\widehat{ \mathbf{v}}_2\|\widehat{\mathbf{w}}_{1,1}\|\widehat{ \mathbf{w}}_{1,2}\|\mathsf{id})\in\{-1,0,1\}^{\widehat{L}_1} $. 

Through some basic algebra, we can from a matrix $\widehat{ \mathbf{M}}_1\in\mathbb{Z}_q^{(n+m+\ell_1)\times\widehat{L}_1}$ and vector ${\mathbf{u}}_1=(\mathbf{u}\|\mathbf{c}_1\|\mathbf{c}_2)\in\mathbb{Z}_q^{n+m+\ell_1}$ such that  equations~(\ref{tpbs-eq:prover-goal-1-1}), (\ref{tpbs-eq:prover-goal-1-2}) are equivalent to one equation of the following form 
\[\widehat{\mathbf{M}}_1\cdot\widehat{ \mathbf{w}}_1={\mathbf{u}}_1 \bmod q.\]
Similarly, define $\widehat{L}_2=\ell_2+d$, we can form $\widehat{\mathbf{M}}_2\in\mathbb{Z}_2^{n\times \widehat{L}_2}$,  ${ \mathbf{u}}_2\stackrel{\triangle}{=}\mathbf{m}\in\mathbb{Z}_2^{n}$, and $\widehat{\mathbf{w}}_2=(\mathbf{p}\|\mathbf{q})\in\{0,1\}^{\widehat{L}_2}$ such that $\mathbf{G}_1\cdot \mathbf{p}+\mathbf{G}_2\cdot\mathbf{q}=\mathbf{m}\bmod 2$ is equivalent to  
\[\widehat{\mathbf{M}}_2\cdot\widehat{ \mathbf{w}}_2={\mathbf{u}}_2 \bmod 2.\]
\remove{
	Now we have transformed~(\ref{tpbs-eq:prover-goal-0}) to \begin{equation}\label{tpbs-eq:prover-goal-2}
	\begin{cases}
	\widehat{\mathbf{M}}\cdot\widehat{ \mathbf{w}}=\widehat{\mathbf{u}} \bmod q;\\
	\mathbf{G}_1\cdot \mathbf{p}+\mathbf{G}_2\cdot\mathbf{q}=\mathbf{m}\bmod 2
	\end{cases}
	\end{equation}
}

\remove{Now we will manage to transform~(\ref{tpbs-eq:prover-goal-2}) $\widehat{\mathbf{M}}\cdot\widehat{ \mathbf{w}}=\widehat{\mathbf{u}} \bmod q$ into a equation of the form $\mathbf{M}\cdot \mathbf{w}=\widehat{\mathbf{u}} \bmod q$ such that our secret vector $\mathbf{w}$ fulfills the conditions in~(\ref{eq:zk-equivalence}). }

\noindent{\sc Extending-Permuting.} Now we will manage to transform our secret vector $\widehat{ \mathbf{w}}=(\widehat{\mathbf{w}}_1\|\widehat{\mathbf{w}}_2)$ to $\mathbf{w}=(\mathbf{w}_1\|\mathbf{w}_2)$ so that the latter fulfills the conditions in~(\ref{eq:zk-equivalence}).  
To this end, we employ the following refined extension-permutation techniques in~\cite{LNRW18TCS,LNWX18PKC}.

\noindent{\bf Technique for proving that $\mathbf{z}\in\{0,1\}^{\mathfrak{m}}$.} For any $a \in \{0,1\}$, we denote by $\overline{a}$ the bit $1-a$. The addition operation modulo $2$ is denoted by~$\oplus$.
For any $\mathbf{z}=[z_1|\cdots| z_{\mathfrak{m}}]^{\top}\in\{0,1\}^{{\mathfrak{m}}}$, define an extension of it  as
\[\mathsf{enc_2(\mathbf{z})}=[\bar{z}_1|z_1|\cdots|\bar{z}_{\mathfrak{m}}|z_{\mathfrak{m}}]^{\top}\in\{0,1\}^{2{\mathfrak{m}}}.\]
Now for any vector $\mathbf{b}=[b_1|\cdots|b_{\mathfrak{m}}]^{\top}\in\{0,1\}^{\mathfrak{m}}$, associate a permutation $\phi_{\mathbf{b}}$  that works as follows. When applying to vector $\mathbf{v}=[v_1^{0}|v_1^{1}|\cdots|v_{\mathfrak{m}}^0|v_{\mathfrak{m}}^1]^{\top}\in \mathbb{Z}^{2\mathfrak{m}}$, it permutes $\mathbf{v}$ into the following vector \[[v_1^{b_1}|v_1^{\bar{b}_1}|\cdots|v_{\mathfrak{m}}^{b_1}|v_{\mathfrak{m}}^{\bar{b}_1}]^{\top}.\] For any $\mathbf{z},\mathbf{b}\in\{0,1\}^{\mathfrak{m}}$, it is verifiable that the following equivalence holds.
\begin{eqnarray} \label{tpbs-eq:equivalence-enc-2-vector}
\mathbf{v}=\mathsf{enc}_2(\mathbf{z})\Longleftrightarrow \phi_{\mathbf{b}}(\mathbf{v})=\mathsf{enc}_2(\mathbf{z}\oplus\mathbf{b}).
\end{eqnarray}Define $\mathsf{valid}_2=\{\mathbf{v}: \exists~\mathbf{z}\in\{0,1\}^{\mathfrak{m}}~\text{s.t.}~\mathbf{v}=\mathsf{enc}_2(\mathbf{z})\}$. We have that:  if $\mathbf{v}\in\mathsf{valid}_2$ and $\mathbf{b}$ is uniformly chosen from $\{0,1\}^{\mathfrak{m}}$, then $\phi_{\mathbf{b}}(\mathbf{v})$ is uniform in $\mathsf{valid}_2$. In the Stern's framework, to prove knowledge of $\mathbf{z}\in\{0,1\}^{\mathfrak{m}}$, we first extend $\mathbf{z}$ to $\mathbf{v}\in\mathsf{valid}_2$ and then show that $\mathbf{v}$ is indeed from the set $\mathsf{valid}_2$ through the equivalence observed in~(\ref{tpbs-eq:equivalence-enc-2-vector}). In addition, vector $\mathbf{b}$ acts as a ``one-time pad'' to perfectly hide $\mathbf{v}$, and hence hide $\mathbf{z}$. Moreover, if we need to prove that $\mathbf{z}$ appears somewhere else, we can use the same $\mathbf{b}$ at those places. \smallskip 

\noindent{\bf Technique for proving that $\mathbf{z}\in\{-1,0,1\}^m$.} 
For any integer vector  $\mathbf{a}=[a_1|\cdots|a_\mathfrak{m}]^{\top}\in\mathbb{Z}^{\mathfrak{m}}$,  denote by $[\mathbf{a}]_3$ the vector $\mathbf{a}'=[a_1'|\cdots|a_\mathfrak{m}']^{\top}\in\{-1,0,1\}^{\mathfrak{m}}$, such that $a_i= a_{i}'\bmod 3$ for $i\in[\mathfrak{m}]$. For $\mathbf{z}=[z_1|\cdots|z_{\mathfrak{m}}]^{\top}\in \{-1,0,1\}^{\mathfrak{m}}$, define the $3\mathfrak{m}$-dimensional vector $\mathsf{enc}_3(\mathbf{z})$ as follows:
\[
\mathsf{enc}_3(\mathbf{z}) = \big[[z_1+1]_3| [z_1]_3, [z_1-1]_3|\cdots|[z_{\mathfrak{m}}+1]_3| [z_{\mathfrak{m}}]_3, [z_{\mathfrak{m}}-1]_3\big]^\top  \in \{-1,0,1\}^{3\mathfrak{m}}.
\]

Now, for any  $\mathbf{b}=[b_1|\cdots|b_{\mathfrak{m}}]^{\top} \in \{-1,0,1\}^{\mathfrak{m}}$, define the permutation $\varphi_{\mathbf{b}}$ that transforms vector $\mathbf{v} = [v_1^{(-1)}| v_1^{(0)}|v_1^{(1)}|\cdots|v_{\mathfrak{m}}^{(-1)}| v_{\mathfrak{m}}^{(0)}|v_{\mathfrak{m}}^{(1)}]^\top \in \mathbb{Z}^{3\mathfrak{m}}$ into vector
\[
\varphi_{\mathbf{b}}(\mathbf{v}) = [v_1^{([-e_1-1]_3)}| v_1^{([-e_1]_3)}| v_1^{([-e_1+1]_3)}|\cdots|v_{\mathfrak{m}}^{([-e_{\mathfrak{m}}-1]_3)}| v_{\mathfrak{m}}^{([-e_{\mathfrak{m}}]_3)}| v_{\mathfrak{m}}^{([-e_{\mathfrak{m}}+1]_3)}]^\top.
\]

It is  observed that, for any $\mathbf{z}, \mathbf{b} \in \{-1,0,1\}$, the following equivalence holds.
\begin{eqnarray}\label{tpbs-eq:equivalence-enc-3-vector}
\mathbf{v} = \mathsf{enc}_3(\mathbf{z}) \hspace*{6.8pt}\Longleftrightarrow\hspace*{6.8pt} \varphi_{\mathbf{b}}(\mathbf{v}) = \mathsf{enc}_3([\mathbf{z}+\mathbf{b}]_3).
\end{eqnarray}
Define $\mathsf{valid}_3=\{\mathbf{v}: \exists~\mathbf{z}\in\{-1,0,1\}^{\mathfrak{m}}~\text{s.t.}~\mathbf{v}=\mathsf{enc}_3(\mathbf{z})\}$. We have that:  if $\mathbf{v}\in\mathsf{valid}_3$ and $\mathbf{b}$ is uniformly chosen from $\{-1,0,1\}^{\mathfrak{m}}$, then $\varphi_{\mathbf{b}}(\mathbf{v})$ is uniform in $\mathsf{valid}_3$. Similarly, to prove knowledge of $\mathbf{z}\{-1,0,1\}^{\mathfrak{m}}$, we extend it to $\mathbf{v}\in\mathsf{valid}_3$ and utilize equivalence~(\ref{tpbs-eq:equivalence-enc-3-vector}) to show well-formedness of $\mathbf{v}$. Further, the uniformity of $\mathbf{b}$ perfectly hides the value of $\mathbf{v}$. 
\smallskip 

\noindent
{\bf Technique for proving that $y = t \cdot z$. } 
For any integers $t \in \{0,1\}$ and $z \in \{-1,0,1\}$, construct the $6$-dimensional integer vector $\mathsf{ext}(t,z) \in \{-1,0,1\}^6$ as follows:
\begin{eqnarray*}
	\hspace*{-8pt}
	\mathsf{ext}(t,z) = \big[\hspace*{2.8pt}
	\overline{t}\cdot [z\hspace*{-1.5pt}+\hspace*{-1.5pt}1]_3\hspace*{2.4pt}| \hspace*{2.4pt}
	t \cdot [z\hspace*{-1.5pt}+\hspace*{-1.5pt}1]_3\hspace*{2.4pt}| \hspace*{2.4pt}
	\overline{t} \cdot [z]_3\hspace*{2.4pt}| \hspace*{2.4pt}
	t \cdot [z]_3\hspace*{2.4pt}| \hspace*{2.4pt}
	\overline{t} \cdot [z\hspace*{-1.5pt}-\hspace*{-1.5pt}1]_3\hspace*{2.4pt}| \hspace*{2.4pt}
	t \cdot [z\hspace*{-1.5pt}-\hspace*{-1.5pt}1]_3\hspace*{2.8pt}
	\big]^\top.
\end{eqnarray*}
Now, for any $b \in \{0,1\}$ and $e \in \{-1,0,1\}$, define the permutation $\psi_{b,e}(\cdot)$ that transforms vector $$\mathbf{v} =
\big[v^{(0, -1)}| v^{(1,-1)}| v^{(0,0)}|v^{(1,0)}| v^{(0,1)}| v^{(1,1)}\big]^\top \in \mathbb{Z}^6$$
into vector
\[
\psi_{b,e}(\mathbf{v}) = \big[
v^{(b, [-e-1]_3)}\hspace*{0.8pt}|\hspace*{0.8pt}
v^{(\overline{b}, [-e-1]_3)}\hspace*{0.8pt}|\hspace*{0.8pt}
v^{(b, [-e]_3)}\hspace*{0.8pt}|\hspace*{0.8pt}
v^{(\overline{b}, [-e]_3)}\hspace*{0.8pt}|\hspace*{0.8pt}
v^{(b, [-e+1]_3)}\hspace*{0.8pt}|\hspace*{0.8pt}
v^{(\overline{b}, [-e+1]_3)}
\big]^\top.
\]

We then observe that the following equivalence holds for any $t, b \in \{0,1\}$ and any $z,e \in \{-1,0,1\}$.
\begin{eqnarray}\label{tpbs-eq:equivalence-ext}
\mathbf{v} = \mathsf{ext}(t,z)
\hspace*{6.8pt}\Longleftrightarrow \hspace*{6.8pt}
\psi_{b,e}(\mathbf{v}) = \mathsf{ext}(\hspace*{1.6pt}t \oplus b, \hspace*{1.6pt}[z + e]_3\hspace*{1.6pt}).
\end{eqnarray}
Define $\mathsf{valid}_{e}=\{\mathbf{v}: \exists~{t}\in\{0,1\},~{z}\in\{-1,0,1\}~\text{s.t.}~\mathbf{v} = \mathsf{ext}\big(t, \hspace*{1.6pt}z\big)\}$. We have that:  if $\mathbf{v}\in\mathsf{valid}_e$ and $b,e$ are uniformly chosen from $\{0,1\}$ and $\{-1,0,1\}$ respectively, then $\psi_{b,e}(\mathbf{v})$ is uniform in $\mathsf{valid}_e$. To prove knowledge of $y=t\cdot z$ for $t\in\{0,1\}$ and $z\in\{-1,0,1\}$, we first extend $y$ to $\mathbf{v}\in\mathsf{valid}_e$ and prove well-formedness of $\mathbf{v}$ through equivalence~(\ref{tpbs-eq:equivalence-ext}). The randomness of values $b,e$ hides values $t,z$. Also, if $t,z$ appear elsewhere, we use the same $b,e$ to show that $t,z$ simultaneously satisfy multiple conditions. 

\noindent{\bf Technique for proving that $\mathbf{y}=\mathsf{mix}(\mathbf{t}, \mathbf{z})$. } For any vectors $\mathbf{t}=[t_1|\cdots|t_{\mathfrak{m}_1}]^{\top}\in\{0,1\}^{\mathfrak{m}_1}$ and $\mathbf{z}=[z_1|\cdots|z_{\mathfrak{m}_2}]^{\top}\in\{-1,0,1\}^{\mathfrak{m}_2}$, recall that $\mathbf{y}$ is of the form $[\hspace*{0.8pt}t_1\cdot z_1|\cdots |t_1\cdot z_{\mathfrak{m}_2}|\cdots|t_{\mathfrak{m}_1}\cdot z_1|\cdots| t_{\mathfrak{m}_1}\cdot z_{\mathfrak{m}_2}\hspace*{0.8pt}]^{\top}$. Define  $6\mathfrak{m}_1\mathfrak{m}_2$-dimensional integer vector $\mathsf{Ext}\big(\mathsf{mix}(\mathbf{t}, \mathbf{z})\big)\in\{-1,0,1\}^{6\mathfrak{m}_1\mathfrak{m}_2}$ in the following way:   
\begin{eqnarray*}
	\big(   \mathsf{ext}(t_1,z_1)\|\cdots\|\mathsf{ext}(t_1,z_{\mathfrak{m}_1})\|\cdots\| \mathsf{ext}(t_{\mathfrak{m}_1},z_1)\|\cdots\|\mathsf{ext}(t_{\mathfrak{m}_1},z_{\mathfrak{m}_1}) \big). 
\end{eqnarray*}
Next, for any $\mathbf{b}\in\{0,1\}^{\mathfrak{m}_1}$ and $\mathbf{e}\in\{-1,0,1\}^{\mathfrak{m}_2}$, let the permutation $\Psi_{\mathbf{b},\mathbf{e}}(\cdot)$ act as follows. When applying to vector of form $$\mathbf{v}=\big(\mathbf{v}_{1,1}\|\cdots \|\mathbf{v}_{1,\mathfrak{m}_2}\|\cdots\|\mathbf{v}_{\mathfrak{m}_1,1}\|\cdots \|\mathbf{v}_{\mathfrak{m}_1,\mathfrak{m}_2}\big)\in\mathbb{Z}^{6\mathfrak{m}_1\mathfrak{m}_2},$$  where each block is of size $6$, it transforms $\mathbf{v}$ into $\Psi_{\mathbf{b},\mathbf{e}}(\mathbf{v})$ of form $$\big( \psi_{b_1,e_1}(\mathbf{v}_{1,1})\|\cdots\| \psi_{b_1,e_{\mathfrak{m}_2}}(\mathbf{v}_{1,\mathfrak{m}_2})\|\cdots\| \psi_{b_{\mathfrak{m}_1},e_{1}}(\mathbf{v}_{{\mathfrak{m}_1},1})\|\cdots\|\psi_{b_{\mathfrak{m}_1},e_{\mathfrak{m}_2}}(\mathbf{v}_{{\mathfrak{m}_1},\mathfrak{m}_2})
\big).$$
It then follows from~(\ref{tpbs-eq:equivalence-ext}) that the following equivalence holds for any $\mathbf{t},\mathbf{b}\in\{0,1\}^{\mathfrak{m}_1}$ and any $\mathbf{z},\mathbf{e}\in\{-1,0,1\}^{\mathfrak{m}_2}$: 
\begin{eqnarray}\label{tpbs-eq:equivalence-ext-vector}
\mathbf{v} = \mathsf{Ext}\big(\hspace*{1.6pt}\mathsf{mix}(\mathbf{t},\hspace*{1.6pt}\mathbf{z})\hspace*{1.6pt}\big)
\hspace*{6.8pt}\Longleftrightarrow \hspace*{6.8pt}
\Psi_{\mathbf{b},\mathbf{e}}(\mathbf{v}) = \mathsf{Ext}\big(\hspace*{1.6pt}\mathsf{mix}(\mathbf{t} \oplus \mathbf{b}, \hspace*{1.6pt}[\mathbf{z} + \mathbf{e}]_3\hspace*{1.6pt})\big).
\end{eqnarray}
Being prepared with the above extension-permutation techniques, we are ready to extend our secret vector $\widehat{ \mathbf{w}}$ to vector $\mathbf{w}$ in a set $\mathsf{VALID}$ and define a suitable permutation $\Gamma_{\eta}$ such that we obtain equivalence similar to~(\ref{tpbs-eq:equivalence-enc-2-vector}), (\ref{tpbs-eq:equivalence-enc-3-vector}) and (\ref{tpbs-eq:equivalence-ext-vector}). When a secret input appears more than once in our extended vector $\mathbf{w}$, it is crucial that we use the same randomness to define our permutation. 

Let $L_{1,1}=3m\delta_{\beta}$, $L_{1,2}=L_{1,1}$, $L_{1,3}=6\ell m \delta_{\beta}$,  $L_{1,4}=3(n+m+\ell_1)\delta_B$, $L_{1,5}=2\ell_1$, and $L_1=2L_{1,1}+L_{1,3}+L_{1,4}+L_{1,5}$, $L_2=2(\ell_2+d)$. Now we define the extension of
$\widehat{ \mathbf{w}}_1=(\widehat{ \mathbf{v}}_1\|\widehat{ \mathbf{v}}_2\|\widehat{\mathbf{w}}_{1,1}\|\widehat{ \mathbf{w}}_{1,2}\|\mathsf{id})\in\{-1,0,1\}^{\widetilde{L}_1} $ and $\widehat{ \mathbf{w}}_2=(\mathbf{p}\|\mathbf{q})\in \{0,1\}^{\widehat{ L}_2}$ as follows.
\begin{itemize}
	\item $\mathbf{w}_{1,1}=\mathsf{enc}_3(\widehat{ \mathbf{v}}_1)\in\{-1,0,1\}^{L_{1,1}}$ and $\mathbf{w}_{1,2}=\mathsf{enc}_3(\widehat{ \mathbf{v}}_2)\in\{-1,0,1\}^{L_{1,2}}$.\smallskip
	\item $\mathbf{w}_{1,3}=\mathsf{Ext}(\widehat{ \mathbf{w}}_{1,1})\in\{-1,0,1\}^{L_{1,3}}$, recall that $\widehat{ \mathbf{w}}_{1,1}=\mathsf{mix}(\mathsf{id}\|\mathbf{p},\widehat{\mathbf{v}}_2)$. \smallskip
	\item $\mathbf{w}_{1,4}=\mathsf{enc}_3(\widehat{ \mathbf{w}}_{1,2})\in\{-1,0,1\}^{L_{1,4}}$ and $\mathbf{w}_{1,5}=\mathsf{enc}_2(\mathsf{id})\in\{0,1\}^{L_{1,5}}$. \smallskip 
	\item $\mathbf{w}_{2,1}=\mathsf{enc}_2(\mathbf{p})\in\{0,1\}^{2\ell_2}$ and $\mathbf{w}_{2,2}=\mathsf{enc}_2(\mathbf{q})\in\{0,1\}^{2d}$. 
\end{itemize}
Form  secret vector  $\mathbf{w}_1=(\mathbf{w}_{1,1}\|\mathbf{w}_{1,2}\|\mathbf{w}_{1,3}\|\mathbf{w}_{1,4}\|\mathbf{w}_{1,5})\in\{-1,0,1\}^{L_1}$. In the meanwhile, we add suitable zero-columns to matrix $\widehat{\mathbf{M}}_1$ to obtain matrix $\mathbf{M}_1\in \mathbb{Z}_q^{(n+m+\ell_1)\times L_1}$ so that $\mathbf{M}_1\cdot \mathbf{w}_1={\mathbf{u}}_1\bmod q$. Similarly, form $\mathbf{w}_2=(\mathbf{w}_{2,1}\|\mathbf{w}_{2,2})\in\{0,1\}^{L_2}$ and a suitable matrix $\mathbf{M}_2\in\mathbb{Z}_2^{n\times L_2}$ such that we have $\mathbf{M}_2\cdot \mathbf{w}_2={\mathbf{u}}_2\bmod 2$. Let $L=L_1+L_2$ and $\mathbf{w}=(\mathbf{w}_1\|\mathbf{w}_2)\in\{-1,0,1\}^L$. 

Now we specify the set $\mathsf{VALID}$ that consists of the secret vector $\mathbf{w}$, the set $\mathcal{S}$ and the associated permutation set  $\{\Gamma_{\eta}: \eta\in \mathcal{S}\}$ so that the requirements in~(\ref{eq:zk-equivalence}) are satisfied.

Let $\mathsf{VALID}$ be the set that contains vectors of form $$\mathbf{z}=(\mathbf{z}_{1,1}\|\mathbf{z}_{1,2}\|\mathbf{z}_{1,3}\|\mathbf{z}_{1,4}\|\mathbf{z}_{1,5}\|\mathbf{z}_{2,1}\|\mathbf{z}_{2,2})\in\{-1,0,1\}^{L}$$ 
satisfying:
\begin{itemize}
	\item There exists $\mathbf{y}_{v,1},\mathbf{y}_{v,2}\in\{-1,0,1\}^{m\delta_{\beta}}$ such that $\mathbf{z}_{1,i}=\mathsf{enc}_3(\mathbf{y}_{v,i})$ for $i\in\{1,2\}$.  \smallskip 
	\item There exists $\mathbf{y}_{id}\in\{0,1\}^{\ell_1},\mathbf{y}_p\in\{0,1\}^{\ell_2}$  such that  $\mathbf{z}_{1,3}=\mathsf{Ext}\big(\hspace*{2pt}\mathsf{mix}(\hspace*{2pt}\mathbf{y}_{id}\hspace*{1pt}\|\hspace*{1pt}\mathbf{y}_{p},\hspace*{4pt}\mathbf{y}_{v,2}\hspace*{2pt})\hspace*{2pt}\big)$, $\mathbf{z}_{1,5}=\mathsf{enc}_2(\mathbf{y}_{id})$, $\mathbf{z}_{2,1}=\mathsf{enc}_2(\mathbf{y}_{p})$.  \smallskip 
	\item There exists $\mathbf{y}_{1,4}\in\{-1,0,1\}^{(n+m+\ell_1)\delta_{B}}$ and $\mathbf{y}_q\in\{0,1\}^d$ such that $\mathbf{z}_{1,4}=\mathsf{enc}_3(\mathbf{y}_{1,4})$ and $\mathbf{z}_{2,2}=\mathsf{enc}_2(\mathbf{y}_q)$. 
\end{itemize} It is clear that our secret vector $\mathbf{w}$ belongs to the set $\mathsf{VALID}$. 

Let $\mathcal{S}=(\{-1,0,1\}^{m\delta_{\beta}})^2\times  \{-1,0,1\}^{(n+m+\ell_1)\delta_{B}}\times\{0,1\}^{\ell_1}\times \{0,1\}^{\ell_2}\times \{0,1\}^d$. For any $\eta=(\mathbf{b}_{v,1},\mathbf{b}_{v,2},\mathbf{b}_{1,4},\mathbf{b}_{id},\mathbf{b}_{p},\mathbf{b}_{q})\in\mathcal{S}$, let permutation $\Gamma_{\eta}:\mathbb{Z}^{L}\rightarrow\mathbb{Z}^L$ act as follows. When applying to vector of form $$\mathbf{z}=(\mathbf{z}_{1,1}\hspace*{2pt}\|\hspace*{2pt}\mathbf{z}_{1,2}\hspace*{2pt}\|\hspace*{2pt}\mathbf{z}_{1,3}\hspace*{2pt}\|\hspace*{2pt}\mathbf{z}_{1,4}\hspace*{2pt}\|\hspace*{2pt}\mathbf{z}_{1,5}\hspace*{2pt}\|\hspace*{2pt}\mathbf{z}_{2,1}\hspace*{2pt}\|\hspace*{2pt}\mathbf{z}_{2,2})\in\mathbb{Z}^{L}$$ such that the size of the blocks are $3m\delta_{\beta}$, $3m\delta_{\beta}$,  $6\ell m\delta_{\beta}$, $3(n+m+\ell_1)\delta_B$, $2\ell_1$, $2\ell_2$, $2d$ respectively, it transforms $\mathbf{z}$ into vector $\Gamma_{\eta}(\mathbf{z})$ of the following form
\begin{eqnarray*}
	\Gamma_{\eta}(\mathbf{z})=\big( \hspace*{2pt}
	\varphi_{\mathbf{b}_{v,1}} (\mathbf{z}_{1,1}) &\|&
	\varphi_{\mathbf{b}_{v,2}} (\mathbf{z}_{1,2}) \hspace*{6pt}\| \hspace*{6pt}
	\Psi_{\mathbf{b}_{id}\|\mathbf{b}_{p},\hspace*{2pt}\mathbf{b}_{v,2}} (\mathbf{z}_{1,3})\\
	&\|&\hspace*{6pt} \varphi_{\mathbf{b}_{1,4}} (\mathbf{z}_{1,4})\hspace*{6pt} \|\hspace*{6pt}
	\phi_{\mathbf{b}_{id}} (\mathbf{z}_{1,5}) \hspace*{6pt}\|\hspace*{6pt}
	\phi_{\mathbf{b}_{p}} (\mathbf{z}_{2,1}) \hspace*{6pt}\|\hspace*{6pt}
	\phi_{\mathbf{b}_{q}} (\mathbf{z}_{2,2}) \hspace*{2pt}
	\big).
\end{eqnarray*}
Observing the equivalences in~(\ref{tpbs-eq:equivalence-enc-2-vector}), (\ref{tpbs-eq:equivalence-enc-3-vector}), and (\ref{tpbs-eq:equivalence-ext-vector}), it is verifiable  that $\mathsf{VALID}$, $\mathcal{S}$, $\Gamma_{\eta}$ fulfill the conditions in~(\ref{eq:zk-equivalence}). Therefore, we have successfully reduced the considered statement into an instance of the abstracted relation described in Section~\ref{tpbs-subsection:Stern-like-protocols}. At this point, we can   run the interactive  protocol described in Figure~\ref{Figure:Interactive-Protocol} and obtain the desired statistical ZKAoK protocol, with perfect correctness, soundness error $2/3$, and communication cost $\mathcal{O}(L_1\log q+L_2)$, which is of order $\mathcal{O}(\ell\lambda\log^3 \lambda)$.

\section{Conclusions and Open Questions}\label{tpbs-section:tpbs-conclusions}
In this work, we enhanced the study of PBS, by introducing and formalizing the notion of TPBS that equips PBS with a user tracing feature, providing a generic and modular construction of TPBS that satisfies the stringent security requirements we suggest, and instantiating a TPBS scheme based on concrete, quantum-safe assumptions from lattices. We believe our results will inspire further improvements for PBS - an appealing privacy-enhancing cryptographic primitive that deserves more attention from the community. 

Naturally, our work raises a number of interesting open questions. In terms of functionality, it would be alluring to provide PBS/TPBS with useful features such as support of dynamically growing groups~\cite{BSZ05CT-RSA}, efficient revocation of signing keys~\cite{BS04CCS}, as well as fine-grained tracing mechanisms~\cite{KTY04EC,SEHK0O12Pairing,KM15PoPETs}. It would also be great to design concrete lattice-based instantiations of PBS/TPBS that involve more expressive policy languages (e.g., those related to membership/non-membership of sets and ranges, branching programs and circuits). In terms of security, it would be desirable to obtain post-quantum schemes in the standard model or with security proofs in the QROM~\cite{BDFLSZ11AC}. Finally, in terms of efficiency, it would be interesting to develop lattice-based PBS/TPBS schemes with shorter, practically relevant signature sizes. To this end, a promising starting point would be to adapt into the context of PBS the recently proposed techniques for efficient lattice-based zero-knowledge proofs by Bootle et al.~\cite{BLS19C} and Yang et al.~\cite{YAZXYW19C}.

\section*{Acknowledgements}
Funding: This project is in part supported by Alberta Innovates in the Province of Alberta, Canada. Khoa Nguyen is supported by the NTU -- Presidential Postdoctoral Fellowship 2018. Huaxiong Wang is supported by the National Research Foundation, Prime Minister's Office, Singapore under its Strategic Capability Research Centres Funding Initiative and  Singapore Ministry of Education under Research Grant RG12/19. 

\bibliographystyle{abbrvnat} 

\bibliography{TPBS-bib-infor-science.bib}

\begin{thebibliography}{65}
\providecommand{\natexlab}[1]{#1}
\providecommand{\url}[1]{\texttt{#1}}
\expandafter\ifx\csname urlstyle\endcsname\relax
  \providecommand{\doi}[1]{doi: #1}\else
  \providecommand{\doi}{doi: \begingroup \urlstyle{rm}\Url}\fi

\bibitem[Ajtai(1996)]{Ajtai96STOC}
M.~Ajtai.
\newblock Generating hard instances of lattice problems (extended abstract).
\newblock In G.~L. Miller, editor, \emph{STOC 1996}, pages 99--108. {ACM},
  1996.
\newblock URL \url{https://doi.org/10.1145/237814.237838}.

\bibitem[Alwen and Peikert(2009)]{AP09STACS}
J.~Alwen and C.~Peikert.
\newblock Generating shorter bases for hard random lattices.
\newblock In S.~Albers and J.~Marion, editors, \emph{{STACS} 2009}, volume~3 of
  \emph{LIPIcs}, pages 75--86. Schloss Dagstuhl - Leibniz-Zentrum fuer
  Informatik, Germany, 2009.
\newblock URL \url{https://doi.org/10.4230/LIPIcs.STACS.2009.1832}.

\bibitem[Applebaum et~al.(2009)Applebaum, Cash, Peikert, and Sahai]{ACPS09C}
B.~Applebaum, D.~Cash, C.~Peikert, and A.~Sahai.
\newblock Fast cryptographic primitives and circular-secure encryption based on
  hard learning problems.
\newblock In S.~Halevi, editor, \emph{{CRYPTO} 2009}, volume 5677 of
  \emph{LNCS}, pages 595--618. Springer, 2009.
\newblock URL \url{https://doi.org/10.1007/978-3-642-03356-8\_35}.

\bibitem[Backes et~al.(2016)Backes, Meiser, and Schr{\"{o}}der]{BMS16PKC}
M.~Backes, S.~Meiser, and D.~Schr{\"{o}}der.
\newblock Delegatable functional signatures.
\newblock In C.~Cheng, K.~Chung, G.~Persiano, and B.~Yang, editors, \emph{{PKC}
  2016}, volume 9614 of \emph{LNCS}, pages 357--386. Springer, 2016.
\newblock URL \url{https://doi.org/10.1007/978-3-662-49384-7\_14}.

\bibitem[Bansarkhani and {El Kaafarani}(2016)]{BK16Eprint}
R.~E. Bansarkhani and A.~{El Kaafarani}.
\newblock Post-quantum attribute-based signatures from lattice assumptions.
\newblock \emph{{IACR} Cryptol. ePrint Arch.}, 2016:\penalty0 823, 2016.
\newblock URL \url{http://eprint.iacr.org/2016/823}.

\bibitem[Bellare and Fuchsbauer(2014)]{BF14PKC}
M.~Bellare and G.~Fuchsbauer.
\newblock Policy-based signatures.
\newblock In H.~Krawczyk, editor, \emph{{PKC} 2014}, volume 8383 of
  \emph{LNCS}, pages 520--537. Springer, 2014.
\newblock URL \url{https://doi.org/10.1007/978-3-642-54631-0\_30}.

\bibitem[Bellare et~al.(2003)Bellare, Micciancio, and Warinschi]{BMW03EC}
M.~Bellare, D.~Micciancio, and B.~Warinschi.
\newblock Foundations of group signatures: Formal definitions, simplified
  requirements, and a construction based on general assumptions.
\newblock In E.~Biham, editor, \emph{{EUROCRYPT} 2003}, volume 2656 of
  \emph{LNCS}, pages 614--629. Springer, 2003.
\newblock URL \url{https://doi.org/10.1007/3-540-39200-9\_38}.

\bibitem[Bellare et~al.(2005)Bellare, Shi, and Zhang]{BSZ05CT-RSA}
M.~Bellare, H.~Shi, and C.~Zhang.
\newblock Foundations of group signatures: The case of dynamic groups.
\newblock In A.~Menezes, editor, \emph{{CT-RSA} 2005}, volume 3376 of
  \emph{LNCS}, pages 136--153. Springer, 2005.
\newblock URL \url{https://doi.org/10.1007/978-3-540-30574-3\_11}.

\bibitem[Boneh and Shacham(2004)]{BS04CCS}
D.~Boneh and H.~Shacham.
\newblock Group signatures with verifier-local revocation.
\newblock In V.~Atluri, B.~Pfitzmann, and P.~D. McDaniel, editors, \emph{{CCS}
  2004}, pages 168--177. {ACM}, 2004.
\newblock URL \url{https://doi.org/10.1145/1030083.1030106}.

\bibitem[Boneh et~al.(2011)Boneh, Dagdelen, Fischlin, Lehmann, Schaffner, and
  Zhandry]{BDFLSZ11AC}
D.~Boneh, {\"{O}}.~Dagdelen, M.~Fischlin, A.~Lehmann, C.~Schaffner, and
  M.~Zhandry.
\newblock Random oracles in a quantum world.
\newblock In D.~H. Lee and X.~Wang, editors, \emph{{ASIACRYPT} 2011}, volume
  7073 of \emph{LNCS}, pages 41--69. Springer, 2011.
\newblock URL \url{https://doi.org/10.1007/978-3-642-25385-0\_3}.

\bibitem[Bootle et~al.(2019)Bootle, Lyubashevsky, and Seiler]{BLS19C}
J.~Bootle, V.~Lyubashevsky, and G.~Seiler.
\newblock Algebraic techniques for short(er) exact lattice-based zero-knowledge
  proofs.
\newblock In A.~Boldyreva and D.~Micciancio, editors, \emph{{CRYPTO} 2019},
  volume 11692 of \emph{LNCS}, pages 176--202. Springer, 2019.
\newblock URL \url{https://doi.org/10.1007/978-3-030-26948-7\_7}.

\bibitem[Boschini et~al.(2020)Boschini, Camenisch, Ovsiankin, and
  Spooner]{BCOS20PQC}
C.~Boschini, J.~Camenisch, M.~Ovsiankin, and N.~Spooner.
\newblock Efficient post-quantum snarks for {RSIS} and {RLWE} and their
  applications to privacy.
\newblock In J.~Ding and J.~Tillich, editors, \emph{PQCrypto 2020}, volume
  12100 of \emph{LNCS}, pages 247--267. Springer, 2020.
\newblock URL \url{https://doi.org/10.1007/978-3-030-44223-1\_14}.

\bibitem[Boyen(2010)]{Boyen10PKC}
X.~Boyen.
\newblock Lattice mixing and vanishing trapdoors: {A} framework for fully
  secure short signatures and more.
\newblock In P.~Q. Nguyen and D.~Pointcheval, editors, \emph{{PKC} 2010},
  volume 6056 of \emph{LNCS}, pages 499--517. Springer, 2010.
\newblock URL \url{https://doi.org/10.1007/978-3-642-13013-7\_29}.

\bibitem[Boyle et~al.(2014)Boyle, Goldwasser, and Ivan]{BGI14PKC}
E.~Boyle, S.~Goldwasser, and I.~Ivan.
\newblock Functional signatures and pseudorandom functions.
\newblock In H.~Krawczyk, editor, \emph{{PKC} 2014}, volume 8383 of
  \emph{LNCS}, pages 501--519. Springer, 2014.
\newblock URL \url{https://doi.org/10.1007/978-3-642-54631-0\_29}.

\bibitem[Brickell et~al.(2000)Brickell, Pointcheval, Vaudenay, and
  Yung]{BPVY00PKC}
E.~F. Brickell, D.~Pointcheval, S.~Vaudenay, and M.~Yung.
\newblock Design validations for discrete logarithm based signature schemes.
\newblock In H.~Imai and Y.~Zheng, editors, \emph{{PKC} 2000}, volume 1751 of
  \emph{LNCS}, pages 276--292. Springer, 2000.
\newblock URL \url{https://doi.org/10.1007/978-3-540-46588-1\_19}.

\bibitem[Camenisch and Lysyanskaya(2001)]{CL01EC}
J.~Camenisch and A.~Lysyanskaya.
\newblock An efficient system for non-transferable anonymous credentials with
  optional anonymity revocation.
\newblock In B.~Pfitzmann, editor, \emph{{EUROCRYPT} 2001}, volume 2045 of
  \emph{LNCS}, pages 93--118. Springer, 2001.
\newblock URL \url{https://doi.org/10.1007/3-540-44987-6\_7}.

\bibitem[Camenisch et~al.(2005)Camenisch, Hohenberger, and
  Lysyanskaya]{CHL05EC}
J.~Camenisch, S.~Hohenberger, and A.~Lysyanskaya.
\newblock Compact e-cash.
\newblock In R.~Cramer, editor, \emph{{EUROCRYPT} 2005}, volume 3494 of
  \emph{LNCS}, pages 302--321. Springer, 2005.
\newblock URL \url{https://doi.org/10.1007/11426639\_18}.

\bibitem[Canetti et~al.(2004)Canetti, Halevi, and Katz]{CHK04EC}
R.~Canetti, S.~Halevi, and J.~Katz.
\newblock Chosen-ciphertext security from identity-based encryption.
\newblock In C.~Cachin and J.~Camenisch, editors, \emph{{EUROCRYPT} 2004},
  volume 3027 of \emph{LNCS}, pages 207--222. Springer, 2004.
\newblock URL \url{https://doi.org/10.1007/978-3-540-24676-3\_13}.

\bibitem[Cash et~al.(2010)Cash, Hofheinz, Kiltz, and Peikert]{CHKP10EC}
D.~Cash, D.~Hofheinz, E.~Kiltz, and C.~Peikert.
\newblock Bonsai trees, or how to delegate a lattice basis.
\newblock In H.~Gilbert, editor, \emph{{EUROCRYPT} 2010}, volume 6110 of
  \emph{LNCS}, pages 523--552. Springer, 2010.
\newblock URL \url{https://doi.org/10.1007/978-3-642-13190-5\_27}.

\bibitem[Chaum(1982)]{Chaum82C}
D.~Chaum.
\newblock Blind signatures for untraceable payments.
\newblock In D.~Chaum, R.~L. Rivest, and A.~T. Sherman, editors, \emph{{CRYPTO}
  1982}, pages 199--203. Plenum Press, New York, 1982.
\newblock URL \url{https://doi.org/10.1007/978-1-4757-0602-4\_18}.

\bibitem[Chaum and van Heyst(1991)]{CH91EC}
D.~Chaum and E.~van Heyst.
\newblock Group signatures.
\newblock In D.~W. Davies, editor, \emph{{EUROCRYPT} 1991}, volume 547 of
  \emph{LNCS}, pages 257--265. Springer, 1991.
\newblock URL \url{https://doi.org/10.1007/3-540-46416-6\_22}.

\bibitem[Cheng et~al.(2016)Cheng, Nguyen, and Wang]{CNW16DCC}
S.~Cheng, K.~Nguyen, and H.~Wang.
\newblock Policy-based signature scheme from lattices.
\newblock \emph{Des. Codes Cryptogr.}, 81\penalty0 (1):\penalty0 43--74, 2016.
\newblock URL \url{https://doi.org/10.1007/s10623-015-0126-y}.

\bibitem[del Pino et~al.(2018)del Pino, Lyubashevsky, and Seiler]{PLS18CCS}
R.~del Pino, V.~Lyubashevsky, and G.~Seiler.
\newblock Lattice-based group signatures and zero-knowledge proofs of
  automorphism stability.
\newblock In D.~Lie, M.~Mannan, M.~Backes, and X.~Wang, editors, \emph{{CCS}
  2018}, pages 574--591. {ACM}, 2018.
\newblock URL \url{https://doi.org/10.1145/3243734.3243852}.

\bibitem[{El Kaafarani} and Katsumata(2018)]{KK18PKC}
A.~{El Kaafarani} and S.~Katsumata.
\newblock Attribute-based signatures for unbounded circuits in the {ROM} and
  efficient instantiations from lattices.
\newblock In M.~Abdalla and R.~Dahab, editors, \emph{{PKC} 2018}, volume 10770
  of \emph{LNCS}, pages 89--119. Springer, 2018.
\newblock URL \url{https://doi.org/10.1007/978-3-319-76581-5\_4}.

\bibitem[{El Kaafarani} et~al.(2014){El Kaafarani}, Ghadafi, and
  Khader]{EGK14CT-RSA}
A.~{El Kaafarani}, E.~Ghadafi, and D.~Khader.
\newblock Decentralized traceable attribute-based signatures.
\newblock In J.~Benaloh, editor, \emph{{CT-RSA} 2014}, volume 8366 of
  \emph{LNCS}, pages 327--348. Springer, 2014.
\newblock URL \url{https://doi.org/10.1007/978-3-319-04852-9\_17}.

\bibitem[Faust et~al.(2012)Faust, Kohlweiss, Marson, and Venturi]{FKMV12IndoC}
S.~Faust, M.~Kohlweiss, G.~A. Marson, and D.~Venturi.
\newblock On the non-malleability of the fiat-shamir transform.
\newblock In S.~D. Galbraith and M.~Nandi, editors, \emph{{INDOCRYPT} 2012},
  volume 7668 of \emph{LNCS}, pages 60--79. Springer, 2012.
\newblock URL \url{https://doi.org/10.1007/978-3-642-34931-7\_5}.

\bibitem[Feng et~al.(2020)Feng, Liu, Wu, and Li]{FLWL20CT-RSA}
H.~Feng, J.~Liu, Q.~Wu, and Y.~Li.
\newblock Traceable ring signatures with post-quantum security.
\newblock In S.~Jarecki, editor, \emph{{CT-RSA} 2020}, volume 12006 of
  \emph{LNCS}, pages 442--468. Springer, 2020.
\newblock URL \url{https://doi.org/10.1007/978-3-030-40186-3\_19}.

\bibitem[Fiat and Shamir(1986)]{FS86C}
A.~Fiat and A.~Shamir.
\newblock How to prove yourself: Practical solutions to identification and
  signature problems.
\newblock In A.~M. Odlyzko, editor, \emph{{CRYPTO} 1986}, volume 263 of
  \emph{LNCS}, pages 186--194. Springer, 1986.
\newblock URL \url{https://doi.org/10.1007/3-540-47721-7\_12}.

\bibitem[Fuchsbauer and Pointcheval(2008)]{FP08SCN}
G.~Fuchsbauer and D.~Pointcheval.
\newblock Anonymous proxy signatures.
\newblock In R.~Ostrovsky, R.~D. Prisco, and I.~Visconti, editors, \emph{{SCN}
  2008}, volume 5229 of \emph{LNCS}, pages 201--217. Springer, 2008.
\newblock URL \url{https://doi.org/10.1007/978-3-540-85855-3\_14}.

\bibitem[Fujisaki and Suzuki(2008)]{FS08IEICE}
E.~Fujisaki and K.~Suzuki.
\newblock Traceable ring signature.
\newblock \emph{{IEICE} Trans. Fundam. Electron. Commun. Comput. Sci.},
  91-A\penalty0 (1):\penalty0 83--93, 2008.
\newblock URL \url{https://doi.org/10.1093/ietfec/e91-a.1.83}.

\bibitem[Gentry et~al.(2008)Gentry, Peikert, and Vaikuntanathan]{GPV08STOC}
C.~Gentry, C.~Peikert, and V.~Vaikuntanathan.
\newblock Trapdoors for hard lattices and new cryptographic constructions.
\newblock In C.~Dwork, editor, \emph{STOC 2008}, pages 197--206. {ACM}, 2008.
\newblock URL \url{https://doi.org/10.1145/1374376.1374407}.

\bibitem[Goldwasser et~al.(1988)Goldwasser, Micali, and Rivest]{GMR88JC}
S.~Goldwasser, S.~Micali, and R.~L. Rivest.
\newblock A digital signature scheme secure against adaptive chosen-message
  attacks.
\newblock \emph{{SIAM} J. Comput.}, 17\penalty0 (2):\penalty0 281--308, 1988.
\newblock URL \url{https://doi.org/10.1137/0217017}.

\bibitem[Gordon et~al.(2010)Gordon, Katz, and Vaikuntanathan]{GKV10AC}
S.~D. Gordon, J.~Katz, and V.~Vaikuntanathan.
\newblock A group signature scheme from lattice assumptions.
\newblock In M.~Abe, editor, \emph{{ASIACRYPT} 2010}, volume 6477 of
  \emph{LNCS}, pages 395--412. Springer, 2010.
\newblock URL \url{https://doi.org/10.1007/978-3-642-17373-8\_23}.

\bibitem[Groth(2006)]{Groth06AC}
J.~Groth.
\newblock Simulation-sound {NIZK} proofs for a practical language and constant
  size group signatures.
\newblock In X.~Lai and K.~Chen, editors, \emph{{ASIACRYPT} 2006}, volume 4284
  of \emph{LNCS}, pages 444--459. Springer, 2006.
\newblock URL \url{https://doi.org/10.1007/11935230\_29}.

\bibitem[Kawachi et~al.(2008)Kawachi, Tanaka, and Xagawa]{KTX08AC}
A.~Kawachi, K.~Tanaka, and K.~Xagawa.
\newblock Concurrently secure identification schemes based on the worst-case
  hardness of lattice problems.
\newblock In J.~Pieprzyk, editor, \emph{{ASIACRYPT} 2008}, volume 5350 of
  \emph{LNCS}, pages 372--389. Springer, 2008.
\newblock URL \url{https://doi.org/10.1007/978-3-540-89255-7\_23}.

\bibitem[Kiayias et~al.(2004)Kiayias, Tsiounis, and Yung]{KTY04EC}
A.~Kiayias, Y.~Tsiounis, and M.~Yung.
\newblock Traceable signatures.
\newblock In C.~Cachin and J.~Camenisch, editors, \emph{{EUROCRYPT} 2004},
  volume 3027 of \emph{LNCS}, pages 571--589. Springer, 2004.
\newblock URL \url{https://doi.org/10.1007/978-3-540-24676-3\_34}.

\bibitem[Kohlweiss and Miers(2015)]{KM15PoPETs}
M.~Kohlweiss and I.~Miers.
\newblock Accountable metadata-hiding escrow: {A} group signature case study.
\newblock \emph{PoPETs}, 2015\penalty0 (2):\penalty0 206--221, 2015.
\newblock URL \url{https://doi.org/10.1515/popets-2015-0012}.

\bibitem[Laguillaumie et~al.(2013)Laguillaumie, Langlois, Libert, and
  Stehl{\'{e}}]{LLLS13AC}
F.~Laguillaumie, A.~Langlois, B.~Libert, and D.~Stehl{\'{e}}.
\newblock Lattice-based group signatures with logarithmic signature size.
\newblock In K.~Sako and P.~Sarkar, editors, \emph{{ASIACRYPT} 2013}, volume
  8270 of \emph{LNCS}, pages 41--61. Springer, 2013.
\newblock URL \url{https://doi.org/10.1007/978-3-642-42045-0\_3}.

\bibitem[Libert et~al.(2016{\natexlab{a}})Libert, Ling, Mouhartem, Nguyen, and
  Wang]{LLMNW16AC-dgs}
B.~Libert, S.~Ling, F.~Mouhartem, K.~Nguyen, and H.~Wang.
\newblock Signature schemes with efficient protocols and dynamic group
  signatures from lattice assumptions.
\newblock In J.~H. Cheon and T.~Takagi, editors, \emph{{ASIACRYPT} 2016},
  volume 10032 of \emph{LNCS}, pages 373--403, 2016{\natexlab{a}}.
\newblock URL \url{https://doi.org/10.1007/978-3-662-53890-6\_13}.

\bibitem[Libert et~al.(2016{\natexlab{b}})Libert, Ling, Nguyen, and
  Wang]{LLNW16EC}
B.~Libert, S.~Ling, K.~Nguyen, and H.~Wang.
\newblock Zero-knowledge arguments for lattice-based accumulators:
  Logarithmic-size ring signatures and group signatures without trapdoors.
\newblock In M.~Fischlin and J.~Coron, editors, \emph{{EUROCRYPT} 2016}, volume
  9666 of \emph{LNCS}, pages 1--31. Springer, 2016{\natexlab{b}}.
\newblock URL \url{https://doi.org/10.1007/978-3-662-49896-5\_1}.

\bibitem[Libert et~al.(2017)Libert, Ling, Nguyen, and Wang]{LLNW17AC-ecash}
B.~Libert, S.~Ling, K.~Nguyen, and H.~Wang.
\newblock Zero-knowledge arguments for lattice-based prfs and applications to
  e-cash.
\newblock In T.~Takagi and T.~Peyrin, editors, \emph{{ASIACRYPT} 2017}, volume
  10626 of \emph{LNCS}, pages 304--335. Springer, 2017.
\newblock URL \url{https://doi.org/10.1007/978-3-319-70700-6\_11}.

\bibitem[Ling et~al.(2013)Ling, Nguyen, Stehl{\'{e}}, and Wang]{LNSW13PKC}
S.~Ling, K.~Nguyen, D.~Stehl{\'{e}}, and H.~Wang.
\newblock Improved zero-knowledge proofs of knowledge for the {ISIS} problem,
  and applications.
\newblock In K.~Kurosawa and G.~Hanaoka, editors, \emph{{PKC} 2013}, volume
  7778 of \emph{LNCS}, pages 107--124. Springer, 2013.
\newblock URL \url{https://doi.org/10.1007/978-3-642-36362-7\_8}.

\bibitem[Ling et~al.(2015)Ling, Nguyen, and Wang]{LNW15PKC}
S.~Ling, K.~Nguyen, and H.~Wang.
\newblock Group signatures from lattices: Simpler, tighter, shorter,
  ring-based.
\newblock In J.~Katz, editor, \emph{{PKC} 2015}, volume 9020 of \emph{LNCS},
  pages 427--449. Springer, 2015.
\newblock URL \url{https://doi.org/10.1007/978-3-662-46447-2\_19}.

\bibitem[Ling et~al.(2018{\natexlab{a}})Ling, Nguyen, Roux{-}Langlois, and
  Wang]{LNRW18TCS}
S.~Ling, K.~Nguyen, A.~Roux{-}Langlois, and H.~Wang.
\newblock A lattice-based group signature scheme with verifier-local
  revocation.
\newblock \emph{Theor. Comput. Sci.}, 730:\penalty0 1--20, 2018{\natexlab{a}}.
\newblock URL \url{https://doi.org/10.1016/j.tcs.2018.03.027}.

\bibitem[Ling et~al.(2018{\natexlab{b}})Ling, Nguyen, Wang, and Xu]{LNWX18PKC}
S.~Ling, K.~Nguyen, H.~Wang, and Y.~Xu.
\newblock Constant-size group signatures from lattices.
\newblock In M.~Abdalla and R.~Dahab, editors, \emph{{PKC} 2018}, volume 10770
  of \emph{LNCS}, pages 58--88. Springer, 2018{\natexlab{b}}.
\newblock URL \url{https://doi.org/10.1007/978-3-319-76581-5\_3}.

\bibitem[Ling et~al.(2019)Ling, Nguyen, Wang, and Xu]{LNWX19CT-RSA}
S.~Ling, K.~Nguyen, H.~Wang, and Y.~Xu.
\newblock Accountable tracing signatures from lattices.
\newblock In M.~Matsui, editor, \emph{{CT-RSA} 2019}, volume 11405 of
  \emph{LNCS}, pages 556--576. Springer, 2019.
\newblock URL \url{https://doi.org/10.1007/978-3-030-12612-4\_28}.

\bibitem[Liu et~al.(2004)Liu, Wei, and Wong]{LWW04ACISP}
J.~K. Liu, V.~K. Wei, and D.~S. Wong.
\newblock Linkable spontaneous anonymous group signature for ad hoc groups
  (extended abstract).
\newblock In H.~Wang, J.~Pieprzyk, and V.~Varadharajan, editors, \emph{{ACISP}
  2004}, volume 3108 of \emph{LNCS}, pages 325--335. Springer, 2004.
\newblock URL \url{https://doi.org/10.1007/978-3-540-27800-9\_28}.

\bibitem[Liu et~al.(2019)Liu, Nguyen, Yang, Wang, and Wong]{LNYWW19ESORICS}
Z.~Liu, K.~Nguyen, G.~Yang, H.~Wang, and D.~S. Wong.
\newblock A lattice-based linkable ring signature supporting stealth addresses.
\newblock In K.~Sako, S.~Schneider, and P.~Y.~A. Ryan, editors, \emph{{ESORICS}
  2019}, volume 11735 of \emph{LNCS}, pages 726--746. Springer, 2019.
\newblock URL \url{https://doi.org/10.1007/978-3-030-29959-0\_35}.

\bibitem[Lu et~al.(2019)Lu, Au, and Zhang]{LAZ19ACNS}
X.~Lu, M.~H. Au, and Z.~Zhang.
\newblock Raptor: {A} practical lattice-based (linkable) ring signature.
\newblock In R.~H. Deng, V.~Gauthier{-}Uma{\~{n}}a, M.~Ochoa, and M.~Yung,
  editors, \emph{{ACNS} 2019}, volume 11464 of \emph{LNCS}, pages 110--130.
  Springer, 2019.
\newblock URL \url{https://doi.org/10.1007/978-3-030-21568-2\_6}.

\bibitem[Maji et~al.(2011)Maji, Prabhakaran, and Rosulek]{MPR11CT-RSA}
H.~K. Maji, M.~Prabhakaran, and M.~Rosulek.
\newblock Attribute-based signatures.
\newblock In A.~Kiayias, editor, \emph{{CT-RSA} 2011}, volume 6558 of
  \emph{LNCS}, pages 376--392. Springer, 2011.
\newblock URL \url{https://doi.org/10.1007/978-3-642-19074-2\_24}.

\bibitem[Micciancio and Peikert(2012)]{MP12EC}
D.~Micciancio and C.~Peikert.
\newblock Trapdoors for lattices: Simpler, tighter, faster, smaller.
\newblock In D.~Pointcheval and T.~Johansson, editors, \emph{{EUROCRYPT} 2012},
  volume 7237 of \emph{LNCS}, pages 700--718. Springer, 2012.
\newblock URL \url{https://doi.org/10.1007/978-3-642-29011-4\_41}.

\bibitem[Micciancio and Peikert(2013)]{MP13C}
D.~Micciancio and C.~Peikert.
\newblock Hardness of {SIS} and {LWE} with small parameters.
\newblock In R.~Canetti and J.~A. Garay, editors, \emph{{CRYPTO} 2013}, volume
  8042 of \emph{LNCS}, pages 21--39. Springer, 2013.
\newblock URL \url{https://doi.org/10.1007/978-3-642-40041-4\_2}.

\bibitem[Micciancio and Regev(2004)]{MR04FOCS}
D.~Micciancio and O.~Regev.
\newblock Worst-case to average-case reductions based on gaussian measures.
\newblock In \emph{{FOCS} 2004}, pages 372--381. {IEEE} Computer Society, 2004.
\newblock URL \url{https://doi.org/10.1109/FOCS.2004.72}.

\bibitem[Peikert and Rosen(2006)]{PR06TCC}
C.~Peikert and A.~Rosen.
\newblock Efficient collision-resistant hashing from worst-case assumptions on
  cyclic lattices.
\newblock In S.~Halevi and T.~Rabin, editors, \emph{{TCC} 2006}, volume 3876 of
  \emph{LNCS}, pages 145--166. Springer, 2006.
\newblock URL \url{https://doi.org/10.1007/11681878\_8}.

\bibitem[Peikert et~al.(2017)Peikert, Regev, and
  Stephens{-}Davidowitz]{PRS17STOC}
C.~Peikert, O.~Regev, and N.~Stephens{-}Davidowitz.
\newblock Pseudorandomness of ring-lwe for any ring and modulus.
\newblock In H.~Hatami, P.~McKenzie, and V.~King, editors, \emph{{STOC} 2017},
  pages 461--473. {ACM}, 2017.
\newblock URL \url{https://doi.org/10.1145/3055399.3055489}.

\bibitem[Rackoff and Simon(1991)]{RS91C}
C.~Rackoff and D.~R. Simon.
\newblock Non-interactive zero-knowledge proof of knowledge and chosen
  ciphertext attack.
\newblock In J.~Feigenbaum, editor, \emph{{CRYPTO} 1991}, volume 576 of
  \emph{LNCS}, pages 433--444. Springer, 1991.
\newblock URL \url{https://doi.org/10.1007/3-540-46766-1\_35}.

\bibitem[Regev(2005)]{Regev05STOC}
O.~Regev.
\newblock On lattices, learning with errors, random linear codes, and
  cryptography.
\newblock In H.~N. Gabow and R.~Fagin, editors, \emph{STOC 2005}, pages 84--93.
  {ACM}, 2005.
\newblock URL \url{https://doi.org/10.1145/1060590.1060603}.

\bibitem[Rivest et~al.(2001)Rivest, Shamir, and Tauman]{RST01AC}
R.~L. Rivest, A.~Shamir, and Y.~Tauman.
\newblock How to leak a secret.
\newblock In C.~Boyd, editor, \emph{{ASIACRYPT} 2001}, volume 2248 of
  \emph{LNCS}, pages 552--565. Springer, 2001.
\newblock URL \url{https://doi.org/10.1007/3-540-45682-1\_32}.

\bibitem[Sakai et~al.(2012)Sakai, Emura, Hanaoka, Kawai, Matsuda, and
  Omote]{SEHK0O12Pairing}
Y.~Sakai, K.~Emura, G.~Hanaoka, Y.~Kawai, T.~Matsuda, and K.~Omote.
\newblock Group signatures with message-dependent opening.
\newblock In M.~Abdalla and T.~Lange, editors, \emph{Pairing 2012}, volume 7708
  of \emph{LNCS}, pages 270--294. Springer, 2012.
\newblock URL \url{https://doi.org/10.1007/978-3-642-36334-4\_18}.

\bibitem[Stern(1996)]{Stern96IT}
J.~Stern.
\newblock A new paradigm for public key identification.
\newblock \emph{{IEEE} Trans. Inf. Theory}, 42\penalty0 (6):\penalty0
  1757--1768, 1996.
\newblock URL \url{https://doi.org/10.1109/18.556672}.

\bibitem[Torres et~al.(2019)Torres, Kuchta, Steinfeld, Sakzad, Liu, and
  Cheng]{TKSSLC19ACISP}
W.~A.~A. Torres, V.~Kuchta, R.~Steinfeld, A.~Sakzad, J.~K. Liu, and J.~Cheng.
\newblock Lattice ringct {V2.0} with multiple input and multiple output
  wallets.
\newblock In J.~Jang{-}Jaccard and F.~Guo, editors, \emph{{ACISP} 2019}, volume
  11547 of \emph{LNCS}, pages 156--175. Springer, 2019.
\newblock URL \url{https://doi.org/10.1007/978-3-030-21548-4\_9}.

\bibitem[Tsabary(2017)]{Tsabary17TCC}
R.~Tsabary.
\newblock An equivalence between attribute-based signatures and homomorphic
  signatures, and new constructions for both.
\newblock In Y.~Kalai and L.~Reyzin, editors, \emph{{TCC} 2017}, volume 10678
  of \emph{LNCS}, pages 489--518. Springer, 2017.
\newblock URL \url{https://doi.org/10.1007/978-3-319-70503-3\_16}.

\bibitem[Xu and Yung(2004)]{XY04}
S.~Xu and M.~Yung.
\newblock Accountable ring signatures: {A} smart card approach.
\newblock In J.~Quisquater, P.~Paradinas, Y.~Deswarte, and A.~A.~E. Kalam,
  editors, \emph{CARDIS 2004}, volume 153 of \emph{{IFIP}}, pages 271--286.
  Kluwer/Springer, 2004.
\newblock URL \url{https://doi.org/10.1007/1-4020-8147-2\_18}.

\bibitem[Yang et~al.(2019)Yang, Au, Zhang, Xu, Yu, and Whyte]{YAZXYW19C}
R.~Yang, M.~H. Au, Z.~Zhang, Q.~Xu, Z.~Yu, and W.~Whyte.
\newblock Efficient lattice-based zero-knowledge arguments with standard
  soundness: Construction and applications.
\newblock In A.~Boldyreva and D.~Micciancio, editors, \emph{{CRYPTO} 2019},
  volume 11692 of \emph{LNCS}, pages 147--175. Springer, 2019.
\newblock URL \url{https://doi.org/10.1007/978-3-030-26948-7\_6}.

\bibitem[Zhang et~al.(2019)Zhang, Liu, Hu, Zhang, and Jia]{ZLHZJ19CCS}
Y.~Zhang, X.~Liu, Y.~Hu, Q.~Zhang, and H.~Jia.
\newblock Attribute-based signatures for inner-product predicate from lattices.
\newblock In J.~Vaidya, X.~Zhang, and J.~Li, editors, \emph{{CSS} 2019}, volume
  11982 of \emph{LNCS}, pages 173--185. Springer, 2019.
\newblock URL \url{https://doi.org/10.1007/978-3-030-37337-5\_14}.

\end{thebibliography}
\appendix

\section{Deferred Algorithms for Our Proof System}\label{tpbs-section:deferred-simprove+extr-algorithms}
In this section, we describe algorithms $\mathsf{SimProve}$ and $\mathsf{Extr}_{\mathsf{nizk}}$ for our SE-NIZK proof system.  Note that in Section~\ref{tpbs-subsection:main-zk-protocol} we have successfully transformed the considered statement when generating signatures to an instance of the abstract relation in Section~\ref{tpbs-subsection:Stern-like-protocols}, it thus suffices to consider the following abstract relation 
\begin{eqnarray*}
	\rho_{\mathrm{abstract}} = \big\{&(\mathbf{M}_i, \mathbf{u}_i)_{i\in\{1,2\}}, (\mathbf{w}_1\|\mathbf{w}_2) \in (\mathbb{Z}_{q_i}^{K_i \times L_i} \times \mathbb{Z}_{q_i}^{K_i})_{i\in\{1,2\}} \times \mathsf{VALID}: \\
	&\mathbf{M}_i\cdot \mathbf{w}_i = \mathbf{u}_ i\bmod q_i~\text{for~}i\in\{1,2\}\big\},
\end{eqnarray*} such that the conditions in~(\ref{eq:zk-equivalence}) hold. 

\noindent{\bf The $\mathsf{SimProve}$ algorithm.} Giving input $(\mathbf{M}_i,\mathbf{u}_i)_{i\in\{1,2\}}$, the algorithm proceeds as follows. It first chooses challenge $\mathrm{CH}=(\mathrm{Ch}_1,\ldots,\mathrm{Ch}_{\kappa})$ uniformly at random from the set $\{1,2,3\}^{\kappa}$. Now it computes $(\mathrm{CMT}_1,\ldots,\mathrm{CMT}_{\kappa})$, $(\mathrm{RSP}_1,\ldots,\mathrm{RSP}_{\kappa})$ in the following way. For each $j\in[\kappa]$, 
\begin{itemize}
	\item if $\mathrm{Ch}_j=1$, it samples $\mathbf{w}_j'=(\mathbf{w}_{j,1}'\|\mathbf{w}_{j,2}')\xleftarrow{\$}\mathsf{VALID}$, $\mathbf{r}_{j,w_1}\xleftarrow{\$}\mathbb{Z}_{q_1}^{L_1}$ $\mathcal{P}$, $\mathbf{r}_{j,w_2} \xleftarrow{\$} \mathbb{Z}_{q_2}^{L_2}$, $\eta_j \xleftarrow{\$} \mathcal{S}$ and  $\rho_{j,1}, \rho_{j,2}, \rho_{j,3}\xleftarrow{\$}\{0,1\}^m$ for $\mathsf{COM}$. Let $\mathbf{r}_{j,w}=(\mathbf{r}_{j,w_1}\|\mathbf{r}_{j,w_2})$ and $\mathbf{z}_j'=\mathbf{w}_j'\boxplus\mathbf{r}_{j,w}$. 
	Then it computes $\mathrm{CMT}_j= \big(C_{j,1}', C_{j,2}', C_{j,3}'\big)$   as
	\begin{gather*}
	C_{j,1}' =  \mathsf{COM}(\eta_j, \{\mathbf{M}_i\cdot \mathbf{r}_{j,w_i} \bmod q_i\}_{i\in\{1,2\}}; \rho_{j,1}),\\
	C_{j,2}' =  \mathsf{COM}(\Gamma_{\eta_j}(\mathbf{r}_{j,w}); \rho_{j,2}), \hspace*{5pt}
	C_{j,3}' =  \mathsf{COM}(\Gamma_{\eta_j}(\mathbf{z}_j'); \rho_{j,3}).
	\end{gather*} and $\mathrm{RSP}_j=(\mathbf{t}_{j,w},\mathbf{t}_{j,r}, \rho_{j,2},\rho_{j,3})$, where $\mathbf{t}_{j,w}=\Gamma_{\eta_j}(\mathbf{w}_j')$ and $\mathbf{t}_{j,r}=\Gamma_{\eta_j}(\mathbf{r}_{j,w})$. 
	
	\item if $\mathrm{Ch}_j=2$, it samples 
	$\mathbf{w}_j'=(\mathbf{w}_{j,1}'\|\mathbf{w}_{j,2}')\xleftarrow{\$}\mathsf{VALID}$, $\mathbf{r}_{j,w_1}\xleftarrow{\$}\mathbb{Z}_{q_1}^{L_1}$ $\mathcal{P}$, $\mathbf{r}_{j,w_2} \xleftarrow{\$} \mathbb{Z}_{q_2}^{L_2}$, $\eta_j \xleftarrow{\$} \mathcal{S}$ and  $\rho_{j,1}, \rho_{j,2}, \rho_{j,3}\xleftarrow{\$}\{0,1\}^m$ for $\mathsf{COM}$. Let $\mathbf{r}_{j,w}=(\mathbf{r}_{j,w_1}\|\mathbf{r}_{j,w_2})$ and $\mathbf{z}_j'=\mathbf{w}_j'\boxplus\mathbf{r}_{j,w}$. 
	Then it computes $\mathrm{CMT}_j= \big(C_{j,1}', C_{j,2}', C_{j,3}'\big)$   as
	\begin{gather*}
	C_{j,1}' =  \mathsf{COM}(\eta_j, \{\mathbf{M}_i\cdot (\mathbf{w}_{j,i}'+\mathbf{r}_{j,w_i})-\mathbf{u}_i \bmod q_i\}_{i\in\{1,2\}}; \rho_{j,1}), \\
	C_{j,2}' =  \mathsf{COM}(\Gamma_{\eta_j}(\mathbf{r}_{j,w}); \rho_{j,2}), \hspace*{5pt}
	C_{j,3}' =  \mathsf{COM}(\Gamma_{\eta_j}(\mathbf{z}_j'); \rho_{j,3}).
	\end{gather*} and $\mathrm{RSP}_j=(\eta_{j,2},\mathbf{z}_{j,2}, \rho_{j,1},\rho_{j,3})$, where $\eta_{j,2}=\eta_j$ and $\mathbf{z}_{j,2}=\mathbf{z}_{j}'$.  
	\item if $\mathrm{Ch}_j=3$, it first computes $\mathbf{w}_{j,i}'$ such that $\mathbf{M}_i\cdot \mathbf{w}_{j,i}'=\mathbf{u}_i\bmod q_i$ for each $i\in\{1,2\}$.  Then it samples $\mathbf{r}_{j,w_1}\xleftarrow{\$}\mathbb{Z}_{q_1}^{L_1}$ $\mathcal{P}$, $\mathbf{r}_{j,w_2} \xleftarrow{\$} \mathbb{Z}_{q_2}^{L_2}$, $\eta_j \xleftarrow{\$} \mathcal{S}$ and  $\rho_{j,1}, \rho_{j,2}, \rho_{j,3}\xleftarrow{\$}\{0,1\}^m$ for $\mathsf{COM}$. Let $\mathbf{w}_j'=(\mathbf{w}_{j,1}'\|\mathbf{w}_{j,2}')$, $\mathbf{r}_{j,w}=(\mathbf{r}_{j,w_1}\|\mathbf{r}_{j,w_2})$ and $\mathbf{z}_j'=\mathbf{w}_j'\boxplus\mathbf{r}_{j,w}$.   Compute $\mathrm{CMT}_j= \big(C_{j,1}', C_{j,2}', C_{j,3}'\big)$   as
	\begin{gather*}
	C_{j,1}' =  \mathsf{COM}(\eta_j, \{\mathbf{M}_i\cdot \mathbf{r}_{j,w_i} \bmod q_i\}_{i\in\{1,2\}}; \rho_{j,1}), \\
	C_{j,2}' =  \mathsf{COM}(\Gamma_{\eta_j}(\mathbf{r}_{j,w}); \rho_{j,2}),  \hspace*{5pt}
	C_{j,3}' =  \mathsf{COM}(\Gamma_{\eta_j}(\mathbf{z}_j'); \rho_{j,3}).
	\end{gather*} and $\mathrm{RSP}_j=(\eta_{j,3},\mathbf{z}_{j,3}, \rho_{j,1},\rho_{j,2})$, where $\eta_{j,3}=\eta_{j}$ and $\mathbf{z}_{j,3}=\mathbf{r}_{j,w}$. 
\end{itemize}
Finally, it outputs the proof $\pi=\big((\mathrm{CMT}_i)_{i\in[\kappa]}, \mathrm{CH},(\mathrm{RSP}_i)_{i\in[\kappa]}\big)$ and program the random oracle as $\mathcal{H}_2((\mathbf{M}_i,\mathbf{u}_i)_{i\in\{1,2\}},(\mathrm{CMT}_i)_{i\in[\kappa]})=\mathrm{CH}$.  

\noindent{\bf The $\mathsf{Extr}_{\mathsf{nizk}}$ algorithm.} Let  $\mathcal{A}$ be an algorithm that outputs a tuple $\big((\mathbf{M}_{i}^{*},\mathbf{u}_i^*)_{i\in\{1,2\}},\pi^* \big)$  after querying simulated proofs for the input $\{(\mathbf{M}_{1,i},\mathbf{u}_{1,i})_{i\in\{1,2\}},\ldots, (\mathbf{M}_{Q_s,i},\mathbf{u}_{Q_s,i})_{i\in\{1,2\}}\}$, where $Q_s$ is the total number of statements queried by $\mathcal{A}$. The $\mathsf{Extr}_{\mathsf{nizk}}$ algorithm outputs $\bot$ if $\pi^*$ is not valid or the outputted tuple is one of the queried tuples.  Otherwise, it proceeds as follows. 

Parse $\pi^*=\big((\mathrm{CMT}_i^*)_{i\in[\kappa]}, (\mathrm{Ch}_i^*)_{i\in[\kappa]},(\mathrm{RSP}_i^*)_{i\in[\kappa]}\big)$. We claim that $\mathcal{A}$ queried to the random oracle $\mathcal{H}_2$ the tuple $h^*\stackrel{\triangle}{=}((\mathbf{M}_i^*,\mathbf{u}_i^*)_{i\in\{1,2\}},(\mathrm{CMT}_i^*)_{i\in[\kappa]})$. Otherwise, guessing correctly this value occurs with probability $3^{-\kappa}$, which is negligible since $\kappa=\omega(\log \lambda)$. Let  $Q_{H_2}$ be  total number of queries to $\mathcal{H}_2$ and $h^*$ be the $t^*$th hash query. Now the $\mathsf{Extr}_{\mathsf{nizk}}$ algorithm replays $\mathcal{A}$ for polynomial-number times with the same random tape and input as in the original run. For each replay, $\mathsf{Extr}_{\mathsf{nizk}}$ behaves exactly the same as the original execution except the difference in the hash replies from the $t^*$th  query onwards. More precisely, it replies $\mathrm{CH}_1,\ldots, \mathrm{CH}_{t^*-1}$ as in the original run while replying fresh random values $\mathrm{CH}_{t^*}',\ldots,\mathrm{CH}_{Q_{H_2}}'$. As such, the $t^*$th hash query to $\mathcal{H}_2$ for all new runs is $h^*$. By the improved forking lemma~\cite{BPVY00PKC}, with probability at least $1/2$ the $\mathsf{Extr}_{\mathsf{nizk}}$ algorithm can obtain a three-fork involving the same tuple $h^*$ with pairwise distinct values $\mathrm{CH}_{t^*}^{(1)}, \mathrm{CH}_{t^*}^{(2)},\mathrm{CH}_{t^*}^{(3)}$. It is verifiable that with probability $1-(\frac{7}{9})^{\kappa}$, there exists $j\in[\kappa]$ such that the $j$th entry of $\mathrm{CH}_{t^*}^{(1)}, \mathrm{CH}_{t^*}^{(2)},\mathrm{CH}_{t^*}^{(3)}$, denoted as $\mathrm{Ch}_{t^*,j}^{(1)}, \mathrm{Ch}_{t^*,j}^{(2)},\mathrm{Ch}_{t^*,j}^{(3)}$ forms the set $\{1,2,3\}$. Without loss of generality, assume $\mathrm{Ch}_{t^*,j}^{(i)}=i$ for $i\in\{1,2,3\}$. We now show how to extract a witness $\mathbf{w}^*=(\mathbf{w}_1^*\|\mathbf{w}_2^*)$ from the corresponding (valid) responses $\mathrm{RSP}_{t^*,j}^{(1)},\mathrm{RSP}_{t^*,j}^{(2)},\mathrm{RSP}_{t^*,j}^{(3)}$ with respect to $\mathrm{CMT}_j^*$.  Suppose $\mathrm{RSP}_{t^*,j}^{(1)}=(\hspace*{-1pt}\mathbf{t}_{j,w},\mathbf{t}_{j,r}, \rho_{j,2}^{(1)},\rho_{j,3}^{(1)})$, 
$\mathrm{RSP}_{t^*,j}^{(2)}\hspace*{-1pt}=\hspace*{-1pt}(\hspace*{-1pt}\eta_{j,2},\mathbf{z}_{j,2}, \rho_{j,1}^{(2)},\rho_{j,3}^{(2)})$, 
$\mathrm{RSP}_{t^*,j}^{(3)}\hspace*{-1pt}=\hspace*{-1pt}(\hspace*{-1pt}\eta_{j,3},\mathbf{z}_{j,3}, \rho_{j,1}^{(3)},\rho_{j,2}^{(3)})$,  $\mathrm{CMT}_j^*=(C_{j,1}^*,C_{j,2}^*,C_{j,3}^*)$. Let $\mathbf{z}_{j,2}=(\mathbf{z}_{j,2,1}\|\mathbf{z}_{j,2,2})$ and $\mathbf{z}_{j,3}=(\mathbf{z}_{j,3,1}\|\mathbf{z}_{j,3,2})$, where $\mathbf{z}_{j,2,i},\mathbf{z}_{j,3,i}\in\mathbb{Z}_{q_i}^{L_i}$ for $i\in\{1,2\}$. 
From the validity of the responses, we have \begin{eqnarray*}
	\begin{cases}
		\mathbf{t}_{j,w}\in \mathsf{VALID}; \\
		C_{j,1}^*=\mathsf{COM}(\eta_{j,2}, \{\mathbf{M}_i^*\cdot \mathbf{z}_{j,2,i}-\mathbf{u}_i^* \bmod q_i\}_{i\in\{1,2\}}, \rho_{j,1}^{(2)});\\ \vspace*{3pt}
		C_{j,1}^*=\mathsf{COM}(\eta_{j,3},\{\mathbf{M}_i^*\cdot \mathbf{z}_{j,3,i} \bmod q_i\}_{i\in\{1,2\}}, \rho_{j,1}^{(3)});\\\vspace*{3pt}
		C_{j,2}^*=\mathsf{COM}(\mathbf{t}_{j,r},\rho_{j,2}^{(1)}) =\mathsf{COM}(\Gamma_{\eta_{j,3}}(\mathbf{z}_{j,3}),\rho_{j,2}^{(3)});\\\vspace*{3pt}
		C_{j,3}^*=\mathsf{COM}(\mathbf{t}_{j,w}\boxplus\mathbf{t}_{j,r},\rho_{j,3}^{(1)})=\mathsf{COM}(\Gamma_{\eta_{j,2}}(\mathbf{z}_{j,2}),\rho_{j,3}^{(2)}). 
	\end{cases}
\end{eqnarray*}
Due to the binding property of the commitment scheme $\mathsf{COM}$, we have 
\begin{eqnarray*}
	& \mathbf{t}_{j,w}\in \mathsf{VALID}; \hspace*{12pt}\eta_{j,2}=\eta_{j,3};  \\
	&\mathbf{t}_{j,r}=\Gamma_{\eta_{j,3}}(\mathbf{z}_{j,3}); \hspace*{12pt}
	\mathbf{t}_{j,w}\boxplus\mathbf{t}_{j,r}=\Gamma_{\eta_{j,2}}(\mathbf{z}_{j,2});\\
	&	\mathbf{M}_i^*\cdot \mathbf{z}_{j,2,i}-\mathbf{u}_i^* =
	\mathbf{M}_i^*\cdot \mathbf{z}_{j,3,i} \mod q_i~\text{for}~i\in\{1,2\}.
\end{eqnarray*}
Let $\mathbf{w}^*=(\mathbf{w}_1^*\|\mathbf{w}_2^*) =\Gamma_{\eta_{j,2}}^{-1}(\mathbf{t}_{j,w})$, where $\mathbf{w}_i^*\in\mathbb{Z}^{L_i}$ for $i\in\{1,2\}$. Then we have $\mathbf{w}^*\in\mathsf{VALID}$. Note that we now have  $\Gamma_{\eta_{j,2}}(\mathbf{w}^*)\boxplus\Gamma_{\eta_{j,2}}(\mathbf{z}_{j,3})=\Gamma_{\eta_{j,2}}(\mathbf{z}_{j,2})$, implying $\mathbf{w}^*\boxplus \mathbf{z}_{j,3}=\mathbf{z}_{j,2}$. In other words, we have $\mathbf{w}_i^*+\mathbf{z}_{j,3,i}=\mathbf{z}_{j,2,i} \bmod q_i$ for $i\in\{1,2\}$. Next, for $i\in\{1,2\}$, the following equation holds: 
\begin{eqnarray*}
	\mathbf{M}_i^*\cdot \mathbf{w}_i^*= \mathbf{M}_i^*\cdot (\mathbf{z}_{j,2,i}-\mathbf{z}_{j,3,i})=\mathbf{u}_i^* \bmod q_i. 
\end{eqnarray*}
Finally, the $\mathsf{Extr}_{\mathsf{nizk}}$ algorithm outputs $\mathbf{w}^*$.

\section{Deferred Building Blocks}\label{tpbs-appexdix:deferred-building-blocks}
\subsection{Building Blocks for the Generic Construction} \label{tpbs-subsection:primitive-for-gener-con}
\noindent{\textbf{Digital Signature Schemes.}}  A digital signature scheme SIG~\cite{GMR88JC} consists of the following three polynomial-time algorithms: $\mathsf{KeyGen}$, $\mathsf{Sign}$, $\mathsf{Verify}$. \begin{description}
	\item[$\mathsf{KeyGen}$] This algorithm takes as input $1^{\lambda}$, and outputs a public-secret key pair $(\mathsf{vk},\mathsf{sk})$.
	\item[$\mathsf{Sign}$] This algorithm takes as input $\mathsf{sk}$ and a message $m$, and outputs a signature~$\sigma$. 
	\item[$\mathsf{Verify}$] This algorithm takes as input $\mathsf{vk}$, message-signature pair  $(m,\sigma)$, and outputs a bit. 
\end{description}
\noindent{\textbf{Correctness.}} A SIG scheme is said to be correct if for all $\lambda$, all $(\mathsf{vk},\mathsf{sk})\leftarrow \mathsf{KeyGen}(1^{\lambda})$ and all $m$, we have $\mathsf{Verify}(\mathsf{vk},m,\mathsf{Sign}(\mathsf{sk},m))=1$.

\noindent{\textbf{EU-CMA.}} Existential unforgeability under chosen message attacks (EU-CMA) of a SIG scheme  is formalized using the following experiment $\mathbf{Exp}_{\mathrm{SIG},\mathcal{A}}^{\mathrm{EU-CMA}}(1^{\lambda})$.    First, a challenger $\mathcal{C}$ runs $\mathsf{KeyGen}$ to obtain $(\mathsf{vk},\mathsf{sk})$. Then $\mathcal{C}$ invokes an adversary $\mathcal{A}$ by sending $\mathsf{vk}$. $\mathcal{A}$ is also given a signing oracle $\mathsf{Sign}(\mathsf{sk},\cdot)$, where it could ask for signatures of any number of messages. Finally, $\mathcal{A}$ outputs an attempted forgery $(m^*,\sigma^*)$. The experiment returns~$1$ if $m^*$ was not queried to $\mathsf{Sign}(\mathsf{sk},\cdot)$ by $\mathcal{A}$ and that $\mathsf{Verify}(\mathsf{vk},m^*,\sigma^*)=1$. Define the advantage of $\mathcal{A}$ as $\mathbf{Adv}_{\mathrm{SIG},\mathcal{A}}^{\mathrm{EU-CMA}}(1^{\lambda})=\mathrm{Pr}[\mathbf{Exp}_{\mathrm{SIG},\mathcal{A}}^{\mathrm{EU-CMA}}(1^{\lambda})=1]$. 
A SIG  scheme is said to be existentially unforgeable under chosen message attacks if for any $\mathrm{PPT}$ $\mathcal{A}$,  $\mathbf{Adv}_{\mathrm{SIG},\mathcal{A}}^{\mathrm{EU-CMA}}(1^{\lambda})$ is negligible in $\lambda$.  
\smallskip 

\noindent {\textbf{Public Key Encryption Schemes. }} A public key encryption scheme PKE~\cite{RS91C} consists of three polynomial-time algorithms: $\mathsf{KeyGen}$, $\mathsf{Enc}$, $\mathsf{Dec}$. 
\begin{description}
	\item[$\mathsf{KeyGen}$] On input security parameter $\lambda$, this algorithm outputs a public-secret key pair $(\mathsf{ek},\mathsf{dk})$. 
	\item[$\mathsf{Enc}$] On input encryption key $\mathsf{ek}$ and a message $m$, this algorithm outputs a ciphertext $\mathsf{ct}$. 
	\item[$\mathsf{Dec}$] On input decryption key $\mathsf{dk}$ and $\mathsf{ct}$, this algorithm outputs $m'$ or $\bot$ indicating decryption failure. 
\end{description}
\noindent {\textbf{Correctness. }} A PKE scheme is said to be correct if for all $\lambda$, $(\mathsf{ek},\mathsf{dk})\leftarrow \mathsf{KeyGen}(1^{\lambda})$, all $m$, we have $\mathsf{Dec}(\mathsf{dk},\mathsf{Enc}(\mathsf{ek},m))=m$. 

\noindent \textbf{IND-CCA}. Indistinguishability under chosen ciphertext attacks (IND-CCA) of a PKE scheme is modeled using the following experiment $\mathbf{Exp}_{\mathrm{PKE},\mathcal{A}}^{\mathrm{IND-CCA}}(1^{\lambda})$. 
To begin with, a challenger $\mathcal{C}$ chooses a random bit $b\in\{0,1\}$ and runs the key generation algorithm to obtain $(\mathsf{ek},\mathsf{dk})$. An adversary $\mathcal{A}$ is then invoked by given $\mathsf{ek}$ and a decryption oracle $\mathsf{Dec}(\mathsf{dk},\cdot)$, from which it could query decryption of any number of ciphertexts.  When $\mathcal{A}$ outputs two messages  $m_0,m_1$ to $\mathcal{C}$, the latter returns a ciphertext $\mathsf{ct}$ that is an encryption of $m_b$. After receiving the challenged ciphertext $\mathsf{ct}$,  $\mathcal{A}$ still has access to the decryption oracle but is not allowed to submit $\mathsf{ct}$ as a query. Finally, $\mathcal{A}$ outputs a bit $b'$ guessing  $\mathsf{ct}$ is an encryption of $m_{b'}$. The experiment returns~$1$ if $b'=b$. Define the advantage of $\mathcal{A}$ as $\mathbf{Adv}_{\mathrm{PKE},\mathcal{A}}^{\mathrm{IND-CCA}}(1^{\lambda})=\mathrm{Pr}[\mathbf{Exp}_{\mathrm{PKE},\mathcal{A}}^{\mathrm{IND-CCA}}(1^{\lambda})=1]$. A PKE scheme is said to be indistinguishable under chosen ciphertext attacks if $\mathbf{Adv}_{\mathrm{PKE},\mathcal{A}}^{\mathrm{IND-CCA}}(1^{\lambda})$ is negligible for any probabilistic polynomial-time algorithm $\mathcal{A}$. 

\noindent\textbf{Simulation-Sound Extractable Non-Interactive Zero-Knowledge Proof Systems.} Fix an $\mathcal{NP}$-relation $\rho$, a simulation-sound extractable non-interactive zero-knowledge (SE-NIZK) proof system $\Pi$~\cite{Groth06AC} for the relation $\rho$ consists the following polynomial-time algorithms: $\mathsf{Setup}$, $\mathsf{Prove}$, $\mathsf{Verify}$. 
\begin{description}
	\item[$\mathsf{Setup}$] This algorithm takes $1^{\lambda}$ as input and  returns a common reference string $\mathsf{crs}$. 
	
	\item[$\mathsf{Prove}$] This algorithm takes $\mathsf{crs}$ and a statement-witness pair $(x,w)\in \rho$, and  outputs a proof $\pi$. 
	
	\item[$\mathsf{Verify}$] Given as input $\mathsf{crs}$ and $(x,\pi)$, it returns a bit. 
	
\end{description} 
\noindent\textbf{Completeness.} $\Pi$ is said to be complete if for all $\lambda$,  all $(x,w)\in \rho$, we have
$\mathrm{Pr}[\mathsf{crs}\leftarrow \mathsf{Setup}(1^{\lambda}), \pi \leftarrow\mathsf{Prove}(\mathsf{crs},x,w):  \mathsf{Verify}(\mathsf{crs},x,\pi)=1]=1$.

\noindent\textbf{Soundness.} $\Pi$ is said to be sound if for all $\lambda$, all $\widehat{\mathsf{Prove}}$,   all $x\notin \mathcal{L}_{\rho}$, we have $\mathrm{Pr}[\mathsf{crs}\leftarrow \mathsf{Setup}(1^{\lambda}), \pi \leftarrow\widehat{\mathsf{Prove}}(\mathsf{crs},x):  \mathsf{Verify}(\mathsf{crs},x,\pi)=1]\leq 2^{-\lambda}$. 

To define zero-knowledge (ZK) and simulation-sound extractability (SE), we need the following three polynomial-time algorithms: $\mathsf{SimSetup}$, $\mathsf{SimProve}$, and $\mathsf{Extr}$. 
\begin{description}
	\item[$\mathsf{SimSetup}$] This algorithm takes $1^{\lambda}$ as input and  returns a simulated common reference string $\mathsf{crs}$ together with a trapdoor  $\mathsf{tr}$.
	
	\item[$\mathsf{SimProve}$]  It takes as input $\mathsf{tr}$ and a statement $x$ and outputs a simulated proof $\pi$. Note that $x$ may not be in language $\mathcal{L}_{\rho}$. 
	
	\item[$\mathsf{Extr}$] On input a tuple $(\mathsf{tr},x,\pi)$, this algorithm outputs a witness $w$. 
\end{description} 
\noindent\textbf{Zero-Knowledge}. Informally speaking, ZK implies that a simulated $\mathsf{crs}$ is indistinguishable from one produced by $\mathsf{Setup}$ and that a simulated proof is indistinguishable from one generated by $\mathsf{Prove}$. 
We model the definition of ZK in  experiment $\mathbf{Exp}_{\Pi,\mathcal{A}}^{\mathrm{ZK}}(1^{\lambda})$ in Figure~\ref{tpbs-fig:zk+se}.  $\Pi$ is said to satisfy (computational) zero-knowledge if $\mathbf{Adv}_{\Pi,\mathcal{A}}^{\mathrm{ZK}}(1^{\lambda})=\mathrm{Pr}[\mathbf{Exp}_{\Pi,\mathcal{A}}^{\mathrm{ZK}}(1^{\lambda})=1]$ is negligible in~$\lambda$ for all $\mathrm{PPT}$ adversary $\mathcal{A}$. 

\noindent\textbf{Simulation-Sound Extractability.} We model the definition of SE using experiment $\mathbf{Exp}_{\Pi,\mathcal{A}}^{\mathrm{SE}}(1^{\lambda})$ in Figure~\ref{tpbs-fig:zk+se}. $\Pi$ is said to simulation-sound extractable if for all  $\mathrm{PPT}$ adversary $\mathcal{A}$, we have $\mathbf{Adv}_{\Pi,\mathcal{A}}^{\mathrm{SE}}(1^{\lambda})=\mathrm{Pr}[\mathbf{Exp}_{\Pi,\mathcal{A}}^{\mathrm{SE}}(1^{\lambda})=1]$ negligible in~$\lambda$. 
\begin{figure}
	\begin{center}
		\begin{tabular}{cc}
			\begin{minipage}{8cm}
				\underline{\sc Initialize}\hspace*{1.6cm}\tikz \draw[fill=green!5] (2.4,-0.3) rectangle (0.068cm, 12pt) node[pos=0.5]{$\mathbf{Exp}_{\Pi,\mathcal{A}}^{\mathrm{ZK}}(1^{\lambda})$};\\
				$\hspace*{6pt}b\xleftarrow{\$}\{0,1\}$\\
				$\hspace*{6pt}(\mathsf{crs}_0,\mathsf{tr})\leftarrow \mathsf{SimSetup}(1^{\lambda})$ \\
				$\hspace*{6pt}\mathsf{crs}_1\leftarrow \mathsf{Setup}(1^{\lambda})$ \\
				$\hspace*{6pt}\text{Return}~\mathsf{crs}_b$
				
				\underline{\sc Proof$(x,w)$}\vspace*{2.2pt}\\
				$\hspace*{6pt}\text{If}~(x,w)\in \rho, \text{then}~\pi_0\leftarrow \mathsf{SimProve}(\mathsf{tr},x)$\\
				$\hspace*{68pt}\text{else}~\pi_0\leftarrow \bot$\\
				$\hspace*{6pt}\pi_1\leftarrow\mathsf{Prove}(\mathsf{crs},x,w)$\\
				$\hspace*{6pt}\text{Return}~\pi_b$
				
				\underline{\sc Finalize$(b')$}\vspace*{2.2pt}\\
				$\hspace*{6pt}\text{Return}~(b'=b)$
			\end{minipage}
			&
			\begin{minipage}{8cm}
				\underline{\sc Initialize}\hspace*{1.6cm}\tikz \draw[fill=green!5] (2.4,-0.3) rectangle (0.068cm, 12pt) node[pos=0.5]{$\mathbf{Exp}_{\Pi,\mathcal{A}}^{\mathrm{SE}}(1^{\lambda})$};\\
				$\hspace*{6pt}(\mathsf{crs}_0,\mathsf{tr})\leftarrow \mathsf{SimSetup}(1^{\lambda}),Q\leftarrow \emptyset$ \\
				$\hspace*{6pt}\text{Return}~\mathsf{crs}$ 
				
				\underline{\sc SimProve$(x)$}\vspace*{2.2pt}\\
				$\hspace*{6pt}\pi\leftarrow \mathsf{SimProve}(\mathsf{crs},\mathsf{tr},x)$\\ $\hspace*{6pt}Q=Q\cup \{(x,\pi)\}$\\
				$\hspace*{6pt}\text{Return}~\pi$
				
				\underline{\sc Finalize$(x,\pi)$}\vspace*{2.2pt}\\
				$\hspace*{6pt}w\leftarrow \mathsf{Extr}(\mathsf{tr},x,\pi)$\\
				$\hspace*{6pt}\text{Return}~1~\text{if the following hold}:$\\
				$\hspace*{24pt}(x,\pi)\notin Q$\\
				$\hspace*{24pt}\mathsf{Verify}(\mathsf{crs},x,\pi)=1$\\
				$\hspace*{24pt}(x,w)\notin \rho$
			\end{minipage}
		\end{tabular}
	\end{center}
	\caption{Games defining zero-knowledge and simulation-sound extractability of a proof system~$\Pi$.}
	\label{tpbs-fig:zk+se}
\end{figure}
\subsection{Building Blocks for the Concrete Construction} \label{tpbs-subsection:primitives-for-concrete-construction}
Below we recall Boyen signature scheme~\cite{Boyen10PKC}. The scheme takes the following parameters:  security parameter $\lambda$,  message length $\ell$, a integer  $n=\mathcal{O}(\lambda)$,  a sufficient large modulus $q=\mathrm{poly}(n)$, a integer $m\geq 2 n\log q$,  a real number  $s=\omega(\sqrt{\log m})\cdot \mathcal{O}(\sqrt{\ell n\log q})$, and a integer bound $\beta=\lceil s\cdot \log n\rceil$. The verification key is a tuple $(\mathbf{A},\mathbf{A}_0,\ldots,\mathbf{A}_{\ell},\mathbf{u})$ while the signing key a matrix $\mathbf{S}$, where $(\mathbf{A},\mathbf{S})$ is generated by $\mathsf{TrapGen}(n,m,q)$ described in Lemma~\ref{tpbs-lemma:trapgen} and matrices $\mathbf{A}_0,\ldots,\mathbf{A}_{\ell}$ and vector $\mathbf{u}$ are all uniformly at random from $\mathbb{Z}_q^{n\times m}$ and $\mathbb{Z}_q^n$, respectively. 

To sign a message $\mathfrak{m}=[\mathfrak{m}[1]|\cdots| \mathfrak{m}[{\ell}]]^{\top}\in\{0,1\}^{\ell}$, the signer forms matrix $\mathbf{A}_{\mathfrak{m}}$ in the following way: $\mathbf{A}_{\mathfrak{m}}=[\mathbf{A}|\mathbf{A}_0+\sum_{j=1}^{\ell}\mathfrak{m}[j]\cdot\mathbf{A}_j]\in\mathbb{Z}_q^{n\times 2m}$. It then outputs a vector $\mathbf{v}\in \Lambda^{\mathbf{u}}(\mathbf{A}_{\mathfrak{m}})$ via $\mathsf{SampleD}\big( \mathsf{ExtBasis}(\mathbf{S},\mathbf{A}_{\mathfrak{m}}),\mathbf{A}_{\mathfrak{m}}, \mathbf{u},s)\big)$. To verify the validity of a signature $\mathbf{v}$, it suffices to check that $\mathbf{A}_{\mathfrak{m}}\cdot \mathbf{v}=\mathbf{u}\bmod q$ and $\|\mathbf{v}\|_{\infty}\leq \beta$.  The above scheme is shown to be EU-CMA secure if $\mathsf{SIS}_{n,m,q,\ell\widetilde{\mathcal{O}}(n)}^{\infty}$ is hard~\cite{Boyen10PKC,MP12EC}. 

We now review the GPV-IBE scheme~\cite{GPV08STOC} in the following. The scheme has the following parameters: security parameter $\lambda$, message length $\ell_1$, a integer $n=\mathcal{O}(\lambda)$, a prime modulus $q=\mathcal{O}(n^2)$, a integer $m\geq 2 n\log q$, a real number $s_1=\omega(\log m)$, an integer bound $B=\widetilde{\mathcal{O}}(\sqrt{n})$ and an efficiently sample distribution $\chi$ over integer $\mathbb{Z}$ that outputs a sample $e$ such that $|e|\leq B$ with overwhelming probability, a hash function $\mathcal{H}_1:\{0,1\}^*\rightarrow \mathbb{Z}_q^{n\times \ell_1}$. The master public-secret key pair is $(\mathbf{B},\mathbf{T})$ generated via $\mathsf{TrapGen}(n,m,q)$. For a identity $\mathsf{iden}\in\{0,1\}^*$, the extraction algorithm first hashes $\mathsf{iden}$ to a matrix $\mathbf{G}=[\mathbf{g}_1|\cdots|\mathbf{g}_{\ell_1}]\in \mathbb{Z}_q^{n\times \ell_1}$ using $\mathcal{H}_1$ and then runs  $\mathbf{f}_i\leftarrow \mathsf{SampleD}(\mathbf{B},\mathbf{T},\mathbf{g}_i,s_1)$ for $i\in[\ell_1]$. Define the decryption key of user $\mathsf{iden}$ as $\mathbf{F}_{\mathsf{iden}}=[\mathbf{f}_1|\cdots|\mathbf{f}_{\ell_1}]\in\mathbb{Z}_q^{m\times\ell_1}$.  

To encrypt a message $\mathfrak{m}=[\mathfrak{m}[1]|\cdots|\mathfrak{m}[\ell_1]]^{\top}\in\{0,1\}^{\ell_1}$ under identity $\mathsf{iden}$, the sender first hashes $\mathsf{iden}$ to obtain $\mathbf{G}$ as above and computes ciphertext as $(\mathbf{c}_1,\mathbf{c}_2)=(\mathbf{B}^{\top}\cdot \mathbf{s}+\mathbf{e}_1,\mathbf{G}^{\top}\cdot \mathbf{s}+\mathbf{e}_2+\mathfrak{m}\cdot \lfloor \frac{q}{2}\rfloor)\in \mathbb{Z}_q^{m}\times \mathbb{Z}_q^{\ell_1}$, where $\mathbf{s}\hookleftarrow\chi^n$, $\mathbf{e}_1\hookleftarrow\chi^m$, and $\mathbf{e}_2\hookleftarrow \chi^{\ell_1}$. To decrypt a ciphertext of the above form, the receiver computes $\mathfrak{m}'= \mathbf{c}_2-\mathbf{F}_{\mathsf{iden}}^{\top}\cdot \mathbf{c}_1\in \mathbb{Z}_q^{\ell_1}$. Then for each $i\in[\ell_1]$, set $\mathfrak{m}[i]=1$ if $\mathfrak{m}'[i]$ is closer to $\lfloor\frac{q}{2} \rfloor$ than to~$0$; otherwise, set $\mathfrak{m}[i]=0$.  Return $\mathfrak{m}=[\mathfrak{m}[1]|\cdots|\mathfrak{m}[\ell_1]]^{\top}\in\{0,1\}^{\ell_1}$. The scheme is proven to be chosen plaintext attacks (CPA) secure if $\mathsf{LWE}_{n,q,\chi}$ is hard~\cite{GPV08STOC}.

\end{document}